\documentclass[draft,12pt,reqno,a4paper]{amsart}

\usepackage[margin=3cm,footskip=1cm]{geometry}

\usepackage{amssymb} \usepackage{amsmath} \usepackage{amsfonts}
\usepackage{amsthm} \usepackage{mathtools}

\usepackage{enumerate} \usepackage[latin1]{inputenc}
\usepackage[shortlabels]{enumitem}

\usepackage{color} 
\usepackage{graphicx} \usepackage{verbatim}

\DeclareMathOperator*{\swslim}{s--w^\star--lim}

\DeclareMathOperator*{\slim}{s--lim}

\DeclareMathOperator*{\vGlim}{{\mathcal G}--lim}

\DeclarePairedDelimiter\abs\lvert\rvert
\DeclarePairedDelimiter\norm\lVert\rVert
\DeclarePairedDelimiter\set{\{}{\}}

 \newcommand{\cs}{{\rm the Cauchy--Schwarz inequality }}

\newcommand{\N}{{\mathbb{N}}} 
\renewcommand{\S}{{\mathbb{S}}}
\newcommand{\R}{{\mathbb{R}}}

 \renewcommand{\c}{{\rm c}}
\newcommand{\e}{{\rm e}} 
 \renewcommand{\i}{{\rm i}}
\renewcommand{\d}{{\rm d}}

\DeclareFontFamily{U}{mathx}{\hyphenchar\font45}
\DeclareFontShape{U}{mathx}{m}{n}{
 <5> <6> <7> <8> <9> <10>
 <10.95> <12> <14.4> <17.28> <20.74> <24.88>
 mathx10
 }{}
\DeclareSymbolFont{mathx}{U}{mathx}{m}{n}
\DeclareFontSubstitution{U}{mathx}{m}{n}
\DeclareMathAccent{\widecheck}{0}{mathx}{"71}

\DeclarePairedDelimiter\inp\langle\rangle

\newcommand\parb[2][]{#1 \big ( #2#1\big )} \newcommand\parbb[2][]{#1
 \Big ( #2#1\Big )}

 \renewcommand{\exp}{{\rm exp}}
\newcommand{\mand}{\text{ and }}

\newcommand{\vA}{{\mathcal A}} \newcommand{\vB}{{\mathcal B}}
\newcommand{\vC}{{\mathcal C}} \newcommand{\vD}{{\mathcal D}}
\newcommand{\vE}{{\mathcal E}} 
\newcommand{\vG}{{\mathcal G}} 
 \newcommand{\vH}{{\mathcal H}}
 \newcommand{\vL}{{\mathcal L}}
\newcommand{\vM}{{\mathcal M}} 
\newcommand{\vO}{{\mathcal O}}

\theoremstyle{plain}
\newtheorem{thm}{Theorem}[section]
\newtheorem{proposition}[thm]{Proposition}
\newtheorem{lemma}[thm]{Lemma} \newtheorem{corollary}[thm]{Corollary}
\newtheorem{cond}[thm]{Condition}
\theoremstyle{definition}

\newtheorem{defn}[thm]{Definition}

 \newtheorem{remark}[thm]{Remark}
\newtheorem{remarks}[thm]{Remarks}
\newtheorem{defns}[thm]{Definitions} \newtheorem*{remarks*}{Remarks}
\newtheorem*{remark*}{Remark}

\numberwithin{equation}{section}

\usepackage[all]{onlyamsmath}
\makeatletter
\newcommand\myparagraph[1]{\par\medskip\noindent\begingroup%
 \def\tagform@##1{\maketag@@@{\bfseries(\ignorespaces##1\unskip\@@italiccorr)}}%
 \textbf{\ignorespaces#1}%
 \endgroup%
 \enspace}
\makeatother

\title{Scattering theory for $C^2$ long-range potentials}

\thanks{K.I. is supported by JSPS KAKENHI, grant nr.\ JP17K05325 and JP23K03163. 
E.S. is supported by the Danish Council for Independent Research $|$ Natural Sciences, 
grant nr.\ DFF-4181-00042. }

\author{K. Ito}
\address[K. Ito]{Graduate School of Mathematical Sciences, The University of Tokyo\\
3-8-1 Komaba, Meguro-ku, Tokyo 153-8914, Japan}
\email{ito@ms.u-tokyo.ac.jp}

\author{E. Skibsted}
\address[E. Skibsted]{Institut for Matematiske Fag\\
Aarhus Universitet\\ Ny Munkegade, 8000 Aarhus C, Denmark}
\email{skibsted@math.au.dk}

\date{\today}
\begin{document}

\begin{abstract} 
We develop a complete stationary scattering theory for
 Schr\"odinger operators on $\R^d$, $d\ge 2$, with $C^2$
long-range potentials. 
 This extends former results in the literature, in particular \cite{Is,
 Is2, II, GY}, which all require a higher degree of smoothness. In this
sense the spirit of our paper is similar to \cite[Chapter XXX]{H1},
which also develops a scattering theory under the $C^2$ condition,
however being very different from ours. While the Agmon--H\"ormander
theory is based on the Fourier transform, our theory is not and may
be seen as more related to our previous approach to scattering theory on
manifolds \cite{IS2,IS3, IS4}. The $C^2$ regularity is
natural in the Agmon--H\"ormander
theory as well as in our theory, in fact probably being `optimal' in
the Euclidean setting.
We prove equivalence of the stationary and time-dependent theories 
by giving stationary representations of associated time-dependent
wave operators. Furthermore we develop a related stationary scattering theory at
fixed energy in terms of asymptotics of generalized eigenfunctions of
minimal growth. 
A basic ingredient of our approach is a solution to the eikonal equation constructed from the geometric variational scheme of \cite{CS}. 
Another key ingredient is strong radiation condition bounds for the
limiting resolvents originating in \cite{HS}. 
They improve formerly known ones \cite{Is, Sa} and considerably
simplify the stationary approach. We obtain the bounds by a new 
commutator scheme whose elementary form allows a small degree of smoothness.
\end{abstract}

\allowdisplaybreaks
\maketitle
\tableofcontents

\section{Introduction}\label{introduction}

\subsection{Setting}

In the present paper we construct a stationary long-range scattering theory for 
the Schr\"odinger operator 
\begin{equation}
H
=-\tfrac12\Delta+V+q
\label{eq:23011416}
\end{equation}
on $\vH=L^2(\mathbb R^d)$ with $d\ge 2$. 
Here $\Delta$ is the ordinary Laplacian on $\mathbb R^d$, 
and we shall often write 
\begin{equation*}
-\Delta=p\cdot p=p_ip_i
;\quad 
p_i=-\mathrm i\partial_i,\ \ i=1,\dots,d ,
\end{equation*}
with the Einstein convention being adopted without tensorial
superscripts.
 
We shall address the problem of constructing such theory under a minimal regularity
condition on the long-range part $V$ of the potential $V+q$. The
second term
$q$ is a standard short-range potential. This
corresponds to taking $l=2$ in the following $C^l$ long-range type condition. For technical reasons 
we consider below, and throughout the
paper, the following more general condition in which $l\geq 2$ is
arbitrary (however typically given as $l= 2$). It is a trivial consequence of the condition that $H$ is self-adjoint.

Let $\N_0=\N\cup \set{0}$, $\R_+=(0,\infty)$ and $\langle x\rangle=(1+|x|^2)^{1/2}$ for
$x\in\mathbb R^d$. For given Banach spaces $X$ and $Y$ we denote by
$\vL(X,Y)$ and $\vC(X,Y)$ the set of bounded and compact
 operators $T:X\to Y$, respectively, and for $Y=X$ we abbreviate $\vL(X)=\vL(X,Y)$ and $\vC(X)=\vC(X,Y)$.

\begin{subequations}
\begin{cond}\label{cond:220525}
Let $V\in C^l(\mathbb R^d;\mathbb R)$ for some $l\in\{2,3,\ldots\}$, 
and assume there exist $\sigma\in (0,1)$, $\rho\in (0,1]$ and $C>0$ 
such that for any $\alpha\in\mathbb N_0^d$ with $|\alpha|\le l$ and $x\in\mathbb R^d$
\begin{align}\label{eq:cond22bb}
|\partial^\alpha V(x)|
\le 
C\langle x\rangle^{-m(|\alpha|)};
\quad 
m(k)=
\begin{cases}
\sigma+k &\text{for }k=0,1,2,\\
\sigma+2+\tfrac{\rho+1}2(k-2)&\text{for }k=2,\dots, l.
\end{cases} 
 \end{align} 
 In addition, let $q\colon\mathbb R^d\to \mathbb R$ be measurable, and assume 
there exists $\tau\in (0,1)$ such that
 \begin{equation}\label{eq:shortrange}\langle x\rangle^{1+\tau}q(x)(-\Delta +1)^{-1}\in \vC(\vH).
 \end{equation} 
Finally, assume the operator $H
=-\tfrac12\Delta+V+q$ does not have positive eigenvalues.
\end{cond}
\begin{remarks}
\begin{enumerate}[1)]
\item
For $l=2$ the above $V+q$ is called a \textit{$2$-admissible
potential}, here adapting the terminology of \cite [Definition~30.1.3]{H1}. 
Several of our main theorems require only $l=2$. 
However for an intermediate key estimate of independent interest we need $l= 4$ 
(or with a modification possibly only $l= 3$), see
Theorem~\ref{thm:proof-strong-bound} (and Remark~\ref{rem:22110219}~\ref{item:bound3}). This estimate 
will be used for a certain regularized $4$-admissible
potential  constructed from a
given $2$-admissible
potential, see
 Remark \ref{rem:22110219}
\ref{item:2342619}. (For the regularized potential the
 parameter $\sigma$ is the same and $\rho<\sigma$, arbitrarily.) In this sense indeed the key estimate serves as an intermediate result
for our study of $2$-admissible
potentials.
\item We  call $V$ a \textit{classical $C^l$ long-range
 potential} if \eqref{eq:cond22bb} holds with $\sigma\in (0,1)$ and either     $l=2$ (in
 which case  $\rho$ is irrelevant)  or  $l\geq
 3$ and in this case $\rho=1$. If  this is fulfilled 
 for all $l\geq 2$ we call $V$ a
 \emph{classical $C^\infty$ long-range potential}. Obviously the
 expression for the 
 \emph{order of decay function} $m:\N_0\to \R_+$ is simplest (in
 fact being case-independent)
 for a classical potential. Nevertheless, in general, its particular
 form is well-suited for an induction argument to be used in the
 proof of the basic result Theorem \ref{thm:main result2} (stated
 below).
\item\label{item:singularities} The local singularities allowed in 
\eqref{eq:shortrange} is not an important issue/difficulty in this
paper. The last assumption on absence of positive
eigenvalues is very weak and can be omitted for example if \eqref{eq:shortrange} is replaced by
assuming boundedness of the function $\langle x\rangle^{1+\tau}q(x)$.
\end{enumerate}\end{remarks}
\end{subequations}

Under Condition~\ref{cond:220525} for $l= 2$ we succeed in fully developing a stationary scattering theory: 
 We characterize the generalized eigenfunctions of
 minimal growth by their asymptotics 
and construct the 
\emph{stationary scattering matrix} as well as the \textit{generalized Fourier transforms}, 
which unitarily diagonalize the (absolutely) continuous part of $H$. 
Such results were formerly obtained by Isozaki~\cite{Is, Is2}, Ikebe--Isozaki~\cite{II} and G\^{a}tel--Yafaev~\cite{GY}
for classical $C^4$ or $C^3$ long-range potentials. In this paper we extend
them to the category of $2$-admissible
potentials.
Note that a similar $C^2$ condition very naturally appears 
in the classical long-range scattering theory, see for example
\cite[Theorem~2.7.1]{DG}. In fact we believe that 
 our condition should be considered as `optimal', although this
 terminology will not be justified in the paper.
Finally we shall prove that 
the adjoints of the generalized Fourier transforms coincide with some constructed time-dependent wave operators, 
verifying that the stationary and the time-dependent approaches are equivalent to each other. 

The long-range scattering theory requires a non-trivial comparison dynamics due to
non-negligible effects from $V$ at infinity. 
To virtually eliminate such effects we solve the
\textit{stationary eikonal equation}, or simply the \textit{eikonal equation}, 
\begin{subequations}
\begin{equation}\label{eq:4c}
\tfrac12\abs{\nabla_x S(\lambda,x)}^2+V(x)=\lambda,\quad\lambda>0,
\end{equation}
in the stationary theory, and the \textit{time-dependent eikonal equation},
which is more commonly called the \textit{Hamilton--Jacobi equation}, 
\begin{equation}\label{eq:4cb}
\partial_tK(t,x)+\tfrac12|\nabla_xK(t,x)|^2+V(x)=0,\quad t>0,
\end{equation}
\end{subequations}
in the time-dependent theory.
In this paper we solve the former equation \eqref{eq:4c} 
by the geometric method of \cite{CS} and derive global estimates
of the solution. 
For comparison we mention that Isozaki \cite{Is, Is2} in his
construction of a solution 
used a classical
PDE-method solving a Cauchy problem, and 
the cited papers \cite{II,GY} rely on Isozaki's solution. 
Once a proper solution to \eqref{eq:4c} is obtained, 
one can solve \eqref{eq:4cb} by using the Legendre transform, see also \cite{IS3}.

Another important technical tool is a strong version of the radiation condition bounds for the limiting resolvents,
which in fact considerably simplifies the stationary scattering theory. 
Such `strong radiation condition bounds' were first established
by Herbst--Skibsted
\cite {HS} for classical $C^\infty$ long-range potentials. For a more
restrictive class of classical $C^\infty$ long-range potentials
(defined by a virial condition) the 
bounds were derived uniformly in non-negative energies \cite {Sk},
yielding a stationary scattering theory at fixed energy including 
the threshold zero. 
In this paper we present a procedure of proof that works within a
low regularity framework (in particular being independent of
pseudodifferential operator theory). 
 A similar procedure was invented and applied ealier to the short-range Stark
 Hamiltonian \cite{AIIS}, however our setup is different and we need
 to proceed rather 
 independently.
Since the proof still requires fourth (or possibly 
 only third) derivatives of the potential,
 we shall regularize it 
up to an error of short-range type 
using a regularization scheme of
H\"ormander \cite[Lemma~30.1.1]{H1}. This leads to the study of
radiation condition bounds of a new classical $C^2$ long-range
potential which conforms with \eqref{eq:cond22bb} for an $l\geq3$
{but} possibly \emph{fails} to be a classical $C^3$ long-range
potential. 
The short-range error from the H\"ormander decomposition
 (see Lemma \ref{lem:230111} for the version to be used in our paper) along with the potential $q$ will be treated by the second resolvent identity. 
Note that the H\"ormander decomposition was also employed in \cite{II},
however Ikebe and Isozaki considered only classical $C^4$ long-range potentials. 

These two ingredients occupy a considerable part of the paper, 
and 
in addition to the entailing stationary and time-dependent scattering theories we consider them as main results
 of independent interest. 

 With these preliminaries done we derive a 
 complete stationary scattering theory. The strong radiation
 condition bounds yield a very fast construction of generalized
 Fourier transforms using the `spherical eikonal coordinates'. These
 coordinates were first used in the context of stationary scattering
 theory in \cite{ACH} for a different setting and with different
 proofs. Then we `pull the results back' to assertions in the
 ordinary spherical coordinates. After a complete stationary theory
 is obtained, the time-dependent theory follows naturally, still
 being quite non-trivial though. Finally we mention that our low
 regularity theory has an application to the 3-body problem. We
 shall briefly discuss this aspect in
 Subsection~\ref{subsubsec:three-body-problem}. 
 Finally, we discuss potential applications to scattering theory in
 Subsection \ref{subsubsec:A perspective:
 Generalization to manifolds}.

\subsection{Main results}

Now we present a series of main results of the paper. 

\subsubsection{Stationary eikonal equation}

Let us first solve the stationary eikonal equation \eqref{eq:4c} outside a large ball. 
Take any $\chi\in C^\infty(\mathbb{R};\mathbb R)$ such that 
\begin{align}
\chi(t)
=\begin{cases}
0 &\mbox{for } t \le 4/3, \\
1 &\mbox{for } t \ge 5/3,
\end{cases}
\quad
\chi'\geq 0,
\label{eq:14.1.7.23.24}
\end{align}
and set for any $R>0$ and $x\in\mathbb R^d$
\begin{align}
\chi_R(x)=\chi(|x|/R). 
\label{eq:14.1.7.23.24b}
\end{align}

 \begin{thm}\label{thm:main result2} 
Suppose Condition~\ref{cond:220525} for some 
$l\geq 2$. Fix any closed interval $I\subset \R_+$. 
Then it follows that for all $R\geq R_0$ for some $R_0>0$, there
exists a real $S\in C^l(I\times (\mathbb R^d\setminus\{0\}))$ and 
$s\in C^l(I\times\mathbb R^d)$ such that: 
\begin{enumerate}
\item\label{item:22120919a}
The function $S$ solves 
\begin{equation}
\tfrac12\abs{\nabla_x S}^2+\chi_RV
=\lambda\ \ \text{on }I\times (\mathbb R^d\setminus\{0\}).
\label{eq:2212116}
\end{equation}
\item\label{item:22120919}
For any $\lambda\in I$, $S(\lambda,\cdot)$ coincides with the geodesic distance 
from the origin with respect to the Riemannian metric 
$g=2\bigl(\lambda-\chi_RV\bigr)\,\d x^2$.

\item\label{item:22120919c}
The functions $S$ and $s$ are related as 
\begin{equation*}
S=\sqrt{2\lambda}|x|(1+s)
\ \ \text{on }I\times (\mathbb R^d\setminus\{0\}),
\end{equation*} and $s$ vanishes on $I\times \{|x|\le R\}$. 
\item \label{item:22120919d}
There exists $C=C(I,R_0)>0$ (being independent of
 $R\geq R_0$) such that 
for any $k+|\alpha|\le l$ and $(\lambda,x)\in I\times\mathbb R^d$
 \begin{align}
 \label{eq:errestza2}
\bigl|\partial^k_\lambda\partial^\alpha_x s(\lambda,x)\bigr|
\le 
C \lambda^{-1-k}\langle x\rangle^{-m(k+|\alpha|)+k}
.
\end{align}
\end{enumerate}
\end{thm}

\begin{remarks}\label{remark:lambda_high-order-deriv}
\begin{enumerate}[1)]

\item 
 For a variation of \eqref{eq:errestza2} see also Corollary~\ref{cor:epsSmall}. 
 \item \label{item:lambda_high-order-deriv1} 
We can extend $S$ and $s$ to be smoothly defined for all $\lambda>0$,
 however, allowing $R=R(\lambda)$ to be $\lambda$-dependent possibly with $R(\lambda)\to\infty$ as $\lambda\to 0_+$. 
 On the other hand, 
 a bound corresponding to \eqref{eq:errestza2} can be kept uniform in $\lambda>0$. 
This can be seen from the proof, but we shall not elaborate on it. 
\item \label{item:lambda_high-order-deriv2bbb}
Our bound \eqref{eq:errestza2} is stronger
than \cite[Theorem~4.1]{Is}, where Isozaki obtained 
\eqref{eq:errestza2}, except for $(k,\alpha)=(3,0)$, for  a classical $C^3$ long-range
  potential. 
G\^atal--Yafaev asserted
\eqref{eq:errestza2} for a classical $C^3$ long-range
  potential  in \cite[Lemma~3.1]{GY}, 
however it was not proved there, and in fact it was indicated in a paragraph subsequent to it that the
assertion requires a classical $C^4$ long-range potential.
\end{enumerate}
\end{remarks}

\subsubsection{Stationary scattering theory at fixed energy}

Next we construct a stationary scattering theory at fixed energy 
through the WKB approximation of the limiting resolvent.

Let us briefly review the \textit{limiting absorption principle}
for the resolvent
\[R(z)=(H-z)^{-1};\quad z\in \mathbb C\setminus\sigma(H). \]
It is a basic and well-studied topic, for details see e.g.\ 
\cite{AIIS2} and the references there.
Recall the 
\textit{Besov spaces}, or the \textit{Agmon--H\"ormander spaces}, defined as 
\begin{align*}
\mathcal B&=
\bigl\{\psi\in L^2_{\mathrm{loc}}\,\big|\,\|\psi\|_{\mathcal B}<\infty\bigr\},\quad 
\|\psi\|_{\mathcal B}=\sum_{m=0}^\infty 2^{m/2}
\|1_m\psi\|_{{\mathcal H}},\\
\mathcal B^*&=
\bigl\{\psi\in L^2_{\mathrm{loc}}\,\big|\, \|\psi\|_{\mathcal B^*}<\infty\bigr\},\quad 
\|\psi\|_{\mathcal B^*}=\sup_{m\ge 0}2^{-m/2}\|1_m\psi\|_{{\mathcal H}},
\\
\mathcal B^*_0
&=
\Bigl\{\psi\in \mathcal B^*\,\Big|\, \lim_{m\to\infty}2^{-m/2}\|1_m\psi\|_{{\mathcal H}}=0\Bigr\}.
\end{align*}
Here we let
\begin{align}
1_0&=1\bigl(\{ |x|<1\}\bigr)\quad \text{and} \quad
1_m=1\bigl(\{2^{m-1}\le |x|<2^m\}\bigr)
\ \ \text{for }m\in\mathbb N,
\label{eq:230328}
\end{align}
with $1(A)$ being the sharp characteristic function of a subset $A\subseteq \mathbb R^d$. 
It is worthwhile recalling that if we define the standard \emph{weighted $L^2$ spaces} as 
\begin{equation*}
L_s^2=\langle x\rangle^{-s}\vH\ \ \text{for }s\in\mathbb R 
,
 \end{equation*} 
then for any $s>1/2$ 
\begin{equation*}
 L^2_s\subset \mathcal B\subset L^2_{1/2}
\subset \mathcal H
\subset L^2_{-1/2}\subset \mathcal B^*_0\subset \mathcal
 B^*\subset L^2_{-s}.
\end{equation*}
It is proved in \cite{AIIS2} that 
locally uniformly in $\lambda>0$ there exist the \textit{limiting resolvents}
\begin{align*}
R(\lambda\pm \i 0)
=\swslim_{z\to \lambda\pm\mathrm i0_+}R(z)\ \ \text{in }\mathcal L(\mathcal B,\mathcal B^*)
,
\end{align*}
or equivalently stated, for any $\psi,\phi\in\mathcal B$ 
\[
\langle \phi, R(\lambda\pm \i 0)\psi\rangle
=\lim_{z\to \lambda\pm\mathrm i0_+}\langle \phi,R(z)\psi\rangle
,
\] 
respectively (see Theorem \ref{thm:221105} for more 
general assertions). In particular the singular continuous spectrum of
$H$ is empty,
$\sigma_{\mathrm{sc}}(H)=\emptyset$.

For these limiting resolvents we discuss the \textit{WKB approximations} as follows, 
using here  and throughout the paper  the notation $\mathcal G=L^2\parb{\S^{d-1}}$.

\begin{thm}\label{thm:comp-gener-four} 
Suppose Condition~\ref{cond:220525} for $l=2$.
Let $I\subset \R_+$ be a closed interval, 
and let $R>0$. 
Assume there exists real $S=\sqrt{2\lambda}|x|(1+s)\in C\bigl(I;C^2(\{|x|>R\})\bigr)$ satisfying: 
\begin{enumerate}[(i)]
\item\label{item:ge1}
For each $\lambda\in I$, $S(\lambda,\cdot)$ solves \eqref{eq:4c} on $\{|x|>R\}$. 
\item\label{item:ge2}
For any compact subset $I'\subseteq I$ there exist $\epsilon,C>0$ 
such that for any $|\alpha|\le 2$, $\lambda\in I'$ and $|x|>R$ 
 \begin{align*}
\left| \partial_x^\alpha s(\lambda,x)\right|
&\le C\langle x\rangle^{-\epsilon-|\alpha|}.
\end{align*}
\end{enumerate}
In addition, for any $\xi\in \mathcal G$ and $(\lambda,x)\in I\times \mathbb R^d$ set 
\begin{align}
\phi_\pm^S[\xi](\lambda,x)
=\tfrac{(2\pi)^{1/2}}{(2\lambda)^{1/4}}
\chi_R(x)
|x|^{-(d-1)/2}\e^{\pm\i S(\lambda, x)} \xi(\hat x),\quad \hat x=|x|^{-1}x
,
 \label{eq:230626}
\end{align}
respectively, where $\chi_R\in C^\infty(\mathbb{R}^d)$ is from \eqref{eq:14.1.7.23.24b}. 
Then the following assertions hold. 

\begin{enumerate}
\item\label{item:23062110}
For any $\lambda\in I$ 
there exist unique $F^\pm(\lambda)\in\vL(\vB,\vG)$ such that for any
$\psi\in \mathcal B$
\begin{align}
R(\lambda \pm\i 0)\psi
-\phi_\pm^S[F^\pm(\lambda)\psi](\lambda,\cdot)
\in \vB^*_0
, 
\label{eq:221206}
\end{align}
respectively,

\item\label{item:23062111}
The mappings $F^\pm\colon I\times \mathcal B\to \mathcal G$
are continuous.

\item\label{item:23062112}
For any $\lambda\in I$ one has the identities
\begin{equation*}
(H-\lambda)F^\pm(\lambda)^*=0
\quad\mand \quad
 F^\pm(\lambda)^*F^\pm(\lambda)
=\delta(H-\lambda), 
\end{equation*} 
respectively, 
where $\delta(H-\lambda)=\pi^{-1}\mathop{\mathrm{Im}}R(\lambda+\mathrm i0)
\in \mathcal L(\mathcal B,\mathcal B^*)$.
\item\label{item:23062113}
For any $\lambda\in I$ the ranges $F^\pm(\lambda)\mathcal B\subseteq \mathcal G$ are dense. 
\end{enumerate}
\end{thm}
\begin{remarks}\label{rem:230203}
\begin{enumerate}[1)]
\item\label{item:230426}
The existence of such $S$ for large $R$ is
guaranteed by Theorem~\ref{thm:main result2}, but here it can be
slightly more general 
(thanks to Theorem~\ref{thm:main result2} the
 condition \ref{item:ge2} is fulfilled with
 $\epsilon=\sigma$ for $I'=I$). Note also that according to Remark \ref{remark:lambda_high-order-deriv}
\ref{item:lambda_high-order-deriv1} 
we may let $I=\R_+$ 
if we allow a {$\lambda$-dependent} $R$.

\item
Clearly for any $\xi\in\mathcal G$ and $\lambda\in
 I$, the functions $\phi_\pm^S[\xi](\lambda,\cdot)\in\mathcal B^*$.
We may think of these \emph{quasi-modes} as purely outgoing/incoming distorted spherical waves, respectively. 
\item 
In Theorem~\ref{thm:char-gener-eigenf-1} below we will 
see the `completeness property' $F^\pm(\lambda)\mathcal B=\mathcal G$
for any $\lambda\in I$. 
\end{enumerate}
\end{remarks}

Now we have stationary versions of scattering quantities. 

\begin{defns}
In the setting of Theorem~\ref{thm:comp-gener-four}, let $\lambda\in I$. 
\begin{enumerate}
\item
The operators $F^\pm(\lambda)\colon \mathcal B\to\mathcal G$ are the
\emph{restricted stationary wave operators}
at energy $\lambda$. 
\item
The adjoints $F^\pm(\lambda)^*\colon \mathcal G\to \vB^*$ are the \emph{stationary wave matrices}. 
\item
The \emph{stationary scattering matrix at energy} $\lambda\in I$ 
is the unitary operator $\mathsf S(\lambda)\colon \mathcal G\to\mathcal G$
obeying
\begin{align}\label{eq:scattering_matrix}
F^+(\lambda)=\mathsf S(\lambda)F^-(\lambda).
\end{align} 
\end{enumerate}
\end{defns}
Indeed the scattering matrix is well-defined, stated as follows.
\begin{corollary}\label{cor:230623}
In the setting of Theorem~\ref{thm:comp-gener-four}, at any energy
$\lambda\in I$ the stationary scattering matrix $\mathsf S(\lambda)$ uniquely exist. 
Moreover, the mapping $I \ni\lambda\mapsto \mathsf S(\lambda) \in \vL(\mathcal G)$ is strongly continuous.
\end{corollary}

\begin{remarks}\label{rem:freeS}
\begin{enumerate}[1)]
\item\label{item:S1}
The stationary scattering matrix $\mathsf S(\lambda)$ is defined pointwise for \textit{all} $\lambda\in I$, not only for a.e.\ $\lambda\in I$, 
unlike in the abstract construction commonly adopted in the
time-dependent approach.
\item \label{item:S2}Although the usual well-known short-range
 scattering theory (defined for $V+q=q$) is a different subject, let us remark that in this
 case it is more conventional to define the quasi-modes
 \eqref{eq:230626} slightly different, more precisely with
 \eqref{eq:230626} modified by the factor $\e^{\mp \i \pi
 (d-3)/4}$. In particular if also the short-range potential $q$
 vanishes,  it then follows that in fact $\mathsf S(\lambda)\equiv I$.
\item \label{item:S3} The dependence of the 
 restricted stationary wave operators on the given function
 $S(\lambda,x)$ is almost canonically given by \eqref{eq:230626} and \eqref{eq:221206}. It is given by an explicit multiplication operator
 of modulus one, see Remark \ref {remark:inverse-asympt-norm} and
 its appearance in Theorem \ref{thm:221207} (\ref{item:230205}).

\end {enumerate} 
\end{remarks}
The stationary scattering theory is intimately related to the
 asymptotics of the \emph{minimal generalized eigenfunctions} $\phi\in
 \vE_\lambda$, 
 where 
\begin{equation*}
 \vE_\lambda=\{\phi\in \vB^*\,|\, (H-\lambda)\phi=0\ \text{in the distributional sense}\};\quad \lambda>0. 
\end{equation*}
It is a non-trivial subspace of $\mathcal B^*$ by Theorem~\ref{thm:comp-gener-four},
and it is minimal in the sense that $\mathcal E_\lambda\cap \mathcal
B^*_0=\{0\}$ (see Remark \ref{remark:construction-general} \ref{item:2re}).

\begin{thm}\label{thm:char-gener-eigenf-1}
 In the setting of Theorem~\ref{thm:comp-gener-four}, let $\lambda\in I$. 
\begin{subequations}
 \begin{enumerate}[(1)]
 \item\label{item:14.5.13.5.40} For any one of $\phi\in \vE_\lambda$ or $\xi_\pm \in \mathcal G$ the
 two other quantities in $\{\phi,\xi_+,\xi_-\}$ uniquely exist
 such that
 \begin{align}\label{eq:gen1}
 \phi -\phi_+^S[\xi_+](\lambda,\cdot)+\phi_-^S[\xi_-](\lambda,\cdot)\in
 \vB_0^*.
 \end{align}

 \item \label{item:14.5.13.5.41} 
The above correspondences $\xi_\pm\to \phi$ and
 $\xi_\mp\to \xi_\pm$ are given by the formulas
 \begin{align}\label{eq:aEigenfw}
 \phi&=2\pi\mathrm i F^\pm(\lambda)^*\xi_\pm\quad \text{and}\quad\xi_+=\mathsf S(\lambda)\xi_-.
 \end{align}

 \item\label{item:14.5.13.5.42} The wave matrices
 $F^\pm(\lambda)^*$ 
are topological linear isomorphisms as 
$\mathcal G\to \vE_\lambda\subseteq \vB^*$.
 In addition, for any $\phi\in \vE_\lambda$ and $\xi_\pm \in \mathcal G$ satisfying \eqref{eq:gen1}
 one has 
 \begin{align}\label{eq:aEigenf2w}
 \|\xi_\pm\|_{\mathcal G}=\tfrac{(2\lambda)^{1/4}}{(2\pi)^{1/2}}\lim_{m \to \infty}2^{-m/2}\norm{1_m\phi}_{\mathcal H},
 \end{align}
where $1_m$ is from \eqref{eq:230328}. 

 \item\label{item:14.5.14.4.17} 
The operators
$F^\pm(\lambda)\colon\mathcal B\to\mathcal G$ and 
$\delta(H-\lambda)\colon \vB\to \vE_\lambda$ are surjective. 

 \end{enumerate}
 \end{subequations}
\end{thm}
\begin{remarks}
\begin{enumerate}[1)]
\item
We can also express $\xi_\pm\in\mathcal G$ as simple oscillatory weak limits of $\phi\in\mathcal E_\lambda$ at infinity, 
see Step~III of the proof. 
\item 
The above result extends in the Euclidean setting \cite{GY,IS2} to 2-admis\-sib\-le potentials. 
See also \cite{II}. 
\end{enumerate}
\end{remarks}

 \subsubsection{Generalized Fourier transforms}
The stationary scattering theory at fixed energy
 applies to the construction of 
the \textit{generalized (or distorted) Fourier transforms}, also
referred to as the \textit{stationary wave operators}, 
which unitarily transform the continuous part of the Schr\"odinger
operator $H$ into a simple multiplication operator. 

 As in Theorem~\ref{thm:comp-gener-four} the subset $I\subset \R_+$
 denotes a closed interval, and we let $P_H(I)$ denote the
 corresponding spectral projection for $H$. We introduce the notation 
\begin{align*}
H_I=H_{|\vH_I},\quad 
\vH_I={P_H(I)\vH}, \quad
 \widetilde{\vH}_I= L^2(I,\d \lambda;\vG)
. 
\end{align*}
Thanks to Theorem~\ref{thm:comp-gener-four} we can also introduce the operators
\begin{equation*}
 \mathcal F_0^\pm =\int _{I} ^\oplus F^\pm(\lambda)\,\d \lambda 
 \colon \mathcal B\to C(I;\mathcal G)\cap \widetilde{\vH}_I. 
\end{equation*} 
These operators can be extended as to be acting from $P_H(I)\mathcal B$ and then in
turn to operators acting from $\vH_I$. These assertions are part of the
following main theorem (see also Step II in the proof). 

\begin{thm}\label{thm:221207} In the setting of
 Theorem~\ref{thm:comp-gener-four} the following asssertions hold:
\begin{enumerate}
\item\label{item:230205a}
The operators $\mathcal F_0^\pm$ 
induce unitary operators $\mathcal F^\pm\colon\vH_I\to \widetilde{\vH}_I$,
respectively.
\item\label{item:230205b}
The induced unitary operators $\mathcal F^\pm$ satisfy 
\[\mathcal F^\pm H_I(\mathcal F^\pm)^*=M_\lambda ,\] 
respectively, 
where $M_\lambda$ denotes the operator of multiplication by $\lambda$ on $\widetilde\vH_I$. 
\item\label{item:230205}
Suppose also $S_1$ satisfies the assumptions of Theorem~\ref{thm:comp-gener-four}, 
and let $\mathcal F^\pm_1$ be the associated unitary operators as above. 
Then there exists the limit
 \begin{subequations}
 \begin{equation}\label{eq:cov1}
 \Theta(\lambda,\omega):=\lim_{r\to \infty}\parb{ S_1(\lambda,r\omega)-S(\lambda,r\omega)}
 \end{equation} taken locally uniformly in 
$(\lambda,\omega)\in I\times\mathbb
 S^{d-1}$, and it follows that
 \begin{equation}
 \label{eq:cov2}\mathcal F^\pm_1=\mathrm e^{\mp\mathrm i\Theta}\mathcal F^\pm,
 \end{equation} 
 \end{subequations}
respectively.
\end{enumerate}
 \end{thm}
\begin{remarks}\label{rem:23020315}
\begin{enumerate}[1)]
\item\label{item:230706}
Under the conditions of Theorem \ref{thm:main result2} one can easily extend the assertion to $I=\R_+$, so that 
the whole absolutely continuous part 
$H_{\mathrm{ac}}=H_{|\mathcal H_{\mathrm{ac}}},\ 
\mathcal H_{\mathrm{ac}}=P_H(\R_+)\mathcal H$, is diagonalized.
One has only to cover $\R_+$ with disjoint intervals
and take a direct sum of the associated generalized Fourier transforms, or 
to adopt a function $S$ defined for all $\lambda>0$ in Theorem~\ref{thm:comp-gener-four}, 
see
Remarks~\ref{remark:lambda_high-order-deriv} \ref{item:lambda_high-order-deriv1}
and \ref{rem:230203} \ref{item:230426}. 
This is straightforward.

\item
In \cite{II,GY} the diagonalization was carried out successfully 
for classical $C^4$ or $C^3$ long-range potentials. Our result extends
these previous results to the $C^2$ case. 
For related works we refer the reader to \cite{Sa,IS2},
still
 the present $C^2$ case is not covered in these
works. On the other hand \cite{H0} and the Agmon--H\"ormander theory
 \cite[Chapter XXX]{H1} cover 2-admissible
 potentials in some form (see \cite[Proposition 3.4.2]{DG}
 for a somewhat 
relevant comparison within their theory). However their approach (specialized
 to the Schr\"odinger operator) is very different from
 ours, for example it is fundamentally based on the Fourier
 transform. Moreover the stationary scattering theory of
 \cite[Chapter XXX]{H1} is only
 developed up to the point of showing asymptotic completeness.

\end{enumerate}
\end{remarks}

\subsubsection{Time-dependent scattering theory}

We present our results on ~time-dependent
 theory, which relate to the time-dependent eikonal equation
 \eqref{eq:4cb}. For that purpose we introduce 
 the space-time regions
 \begin{align*}
 \Omega_\mu&=\bigl\{(t,x)\in \R_+\times \mathbb R^d \mid \abs{x}\ge
 \mu t\bigr\},\quad \mu>0,\\
\Omega_{\mu,T}&=\{(t,x)\in\Omega_\mu\mid t>T>0\},\quad T>0.
 \end{align*}

\begin{thm}\label{thm:2307032}
Suppose Condition~\ref{cond:220525} for some 
$l\geq 2$. 
Fix any $\mu>\mu'>0$
and let $S$ be given as in Theorem~\ref{thm:main result2} with
$I=I_{\mu'}=[\mu'^2/2,\infty)$ (assuming here and for other purposes that $R>0$ is sufficiently large). 
Then for each $(t,x)\in\Omega_\mu$ there exists a unique critical point $\lambda_c=\lambda_c(t,x)\in I$ 
of 
 \begin{equation}
\widetilde K(\lambda,t,x)=S(\lambda,x)-\lambda t 
\label{eq:230705}
\end{equation}
considered as a function of $\lambda \in I$. In addition, if one sets $K=\widetilde K(\lambda_\c,\cdot,\cdot)$, then: 
\begin{enumerate}
\item\label{item:22120919ax}
The function $K\in C^l(\Omega_\mu)$, and it solves
\begin{equation}
 \partial_tK +\tfrac12\abs{\nabla_x K}^2+\chi_RV=0\ \ \text{on }\Omega_\mu.
\label{eq:eik223}
\end{equation}

\item\label{item:22120919dx}
There exists $C>0$ (being independent of all large
 $R$) such that 
for any $k+|\alpha|\le 2$ and $(t,x)\in \Omega_\mu$
 \begin{align}
 \label{eq:errestza2x}
\bigl|\partial^k_t\partial^\alpha_x \bigl(K(t,x)-x^2/(2t)\bigr)\bigr|
\le 
C t^{1-k}\langle x\rangle^{-\sigma-|\alpha|}
.
\end{align}
\end{enumerate}
\end{thm}
\begin{remark}
For $l\geq 3$ it is possible to extend \eqref{eq:errestza2x} to higher
order derivatives (given by $3\leq k+|\alpha|\le l$) 
with appropriate exponents on the right-hand side. Since 
such estimates are not
needed in the present paper, they will not be presented.
\end{remark} 

\begin{defn}
In the setting of Theorem~\ref{thm:2307032}, in particular under the
fixed condition $\mu>\mu'>0$, the function $K\in C^l(\Omega_\mu)$ is called the
\textit{Legendre transform} of $S\in C^l(I_{\mu'}\times (\mathbb
R^d\setminus\{0\}))$. 
\end{defn}

\begin{thm}\label{thm:230703}
Suppose Condition~\ref{cond:220525} with $l=2$. 
Let $\mu,T>0$, and assume $K\in C^2(\Omega_{\mu,T};\R)$ satisfies: 
\begin{enumerate}[(i)]
\item
$K$ solves the Hamilton--Jacobi equation \eqref{eq:4cb} on $\Omega_{\mu,T}$.
\item
There exist $\epsilon, C>0$ such that for any $|\alpha|\le 2$ and $(t,x)\in \Omega_{\mu,T}$
 \begin{align*}
\bigl|\partial^\alpha_x \bigl(K(t,x)-x^2/(2t)\bigr)\bigr|
\le 
C t\langle x\rangle^{-\epsilon-|\alpha|}
.
\end{align*}
\end{enumerate}
In addition, let $J=J_\mu=[\mu^2/2,\infty)$ and define 
isometries
$U^\pm(t)\colon \widetilde{\vH}_J\to \mathcal H$ as follows: 
For any $(t,h)\in (T,\infty)\times \widetilde{\vH}_J $, 
\begin{equation*}
 (U^\pm(t)h)(x)=\e^{\mp 3\pi\i/4}t^{-1}\abs{x}^{1-d/2}\e^{\pm \i
 K(t,x)}h\parb{ x^2/(2t^2), \hat x };\quad \hat x=|x|^{-1}x 
. 
\end{equation*} 
Then the following two assertions hold. 
\begin{enumerate}
\item\label{item:23071721}
There exist the strong limits 
\begin{equation}\label{eq:wave7}
 W^\pm:=\slim _{t\to \infty} \e^{\pm \i tH}U^\pm(t)
\colon\widetilde{\mathcal H}_J\to \mathcal H,
\end{equation} 
respectively. They are isometries, mapping
 $\widetilde{\mathcal H}_{J'}\subseteq \widetilde{\mathcal H}_{J}$ into $ 
 {\mathcal H}_{J'}\subseteq\mathcal H$
 for any closed subinterval $J'\subseteq J$. 
\item\label{item:23071722}
Let $K_1$ also satisfy the same assumptions as $K$, and let $W_1^\pm$
be the associated isometries as above.
Then there exists the limit
 \begin{subequations}
 \begin{equation}\label{eq:coV21}
 \Phi(\lambda,\omega):=\lim_{t\to \infty}\parb{ K_1(t, (2\lambda)^{1/2}t\omega)-K(t, (2\lambda)^{1/2}t\omega)}
 \end{equation} taken locally uniformly in 
$(\lambda,\omega)\in J\times\mathbb
 S^{d-1}$, and it follows that
 \begin{equation}\label{eq:coV22}
 W_1^\pm =W^\pm \mathrm e^{\pm\mathrm i\Phi},
 \end{equation} 
 \end{subequations}
respectively. 
\end{enumerate}
\end{thm}
\begin{remarks}
\begin{enumerate}[1)]
\item
Thanks to Theorem~\ref{thm:2307032} such $K$ always exists for $T$ large enough. 
\item
Our choice of \textit{free comparison dynamics} $U^\pm(t)$ is 
motivated in Remark~\ref{rem230718}. 
\item
 Under the conditions of Theorem \ref{thm:2307032} we can extend
 $U^\pm(t)$, and hence $W^\pm$ as to be acting on $\widetilde{\mathcal H}_{\R_+}$,
 by covering
 $\R_+$ with disjoint intervals (as in Remark~\ref{rem:23020315} \ref{item:230706}) and then taking a direct sum of the
 restrictions to these intervals.
\end{enumerate}
\end{remarks} 

\begin{defn}
The limits $W^\pm$ from \eqref{eq:wave7} are called the \textit{time-dependent wave operators}. 
They are \textit{asymptotically complete on $J' $} (for a
closed subinterval $J'\subseteq J$), if they are unitary operators mapping $\widetilde{\vH}_{J'}$ onto
${\vH}_{J'}$. 
\end{defn}

 The following result shows that the stationary and the
 time-dependent wave operators are essentially mutually
 inverses. First we state a general result for classes of solutions
 to the stationary and the time-dependent eikonal equations, then we
 specialize to the concrete ones constructed geometrically and by
 the Legendre transform, in which case they are indeed mutually
 inverses.

\begin{thm}\label{thm:time-depend-theory2} 
Suppose Condition~\ref{cond:220525} with $l=2$. 
\begin{enumerate} 
 \item \label{item:23071912}
 Let $\mathcal F^\pm$, along with a closed interval $I\subset \R_+$, $R>0$ and a
 stationary 
 eikonal equation solution
 $S$, be given as in 
 Theorem~\ref{thm:221207}. 
Let $W^\pm$, along with $\mu,T>0$, $J=[\mu^2/2,\infty)$ and a
time-dependent eikonal equation solution $K$,
be given as in Theorem~\ref{thm:230703}. Then there exists $\Psi\in C\bigl((I\cap J)\times \mathbb S^{d-1};\mathbb R\bigr)$
such that 
 \begin{equation}\label{eq:fund2} 
 \parb{W^\pm}^*=\mathrm e^{\mp\mathrm i\Psi}\mathcal F^\pm\colon\mathcal H_{I\cap J}\to\widetilde{\mathcal H}_{I\cap J},
 \end{equation}
 respectively. 
In particular, $W^\pm$ are asymptotically complete on $I\cap J$. 
\item \label{item:23071914} Under the conditions of Theorems
 \ref{thm:main result2} and \ref{thm:2307032}, let $K\in
 C^2(\Omega_\mu)$ denote the Legendre transform of $S\in
 C^2\bigl(I_{\mu'}\times (\mathbb R^d\setminus\{0\}\bigr)$;
 $I_{\mu'}=[\mu'^2/2,\infty)$, $\mu>\mu'>0$. Then
 \eqref{eq:fund2} holds with $I=I_{\mu'}=[\mu'^2/2,\infty)$,
 $J=[\mu^2/2,\infty)$ and with $\Psi$ taken
   identically zero, i.e.
\begin{equation}\label{eq:fund22} 
 \parb{W^\pm}^*=\mathcal F^\pm\colon\mathcal H_{J}\to\widetilde{\mathcal H}_{J}.
 \end{equation} 
 In particular, $W^\pm$ are asymptotically complete on $ J$. 
\end{enumerate} 
 \end{thm}
\begin{remark}
The above completeness results resemble asymptotic completeness for
 2-admissible potentials as stated in \cite [Theorem 30.5.10]{H1},
 however there given with a different free dynamics. 
Our approach is more related to our previous study in a geometric setting
\cite{IS2,IS3}, though more regularity of $V$
is imposed there.
In the geometric short-range setting \cite{IS4} we employed only time-dependent methods, 
but the third order derivative was required for the
analogous $K$, see
\cite[Condition~1.3 and (1.8b)]{IS4}.
\end{remark}

\subsubsection{Application to the 3-body problem}\label{subsubsec:three-body-problem}

Our results apply to a recent development in the stationary
scattering theory for $3$-body long-range Hamiltonians \cite{Sk1}. 
Let us see how the application comes about in a slightly simplified form. 
Below we only present a brief outline, and refer the reader to
\cite{Sk1} for precise definitions, terminologies and procedure. See also Remark 
\ref{remark:inverse-asympt-norm}. 

Let $\mathbf X$ be a finite dimensional real inner product space, and $\{\mathbf X^a\}_{a\in\mathcal A}$
a family of subspaces of $\mathbf X$ closed under addition. 
Consider the $3$-body problem, i.e.\ assume that $\# a_{\min}=3$. 
Let $V^a\in C^\infty(\mathbf X^a)$, $a\in\vA\setminus
 \{a_{\min},a_{\max}\}$, be pair potentials,
and assume there exists $\mu\in (\sqrt 3 -1,1)$ such that 
for any $\alpha\in \N_0^{\dim \mathbf X^a}$
\begin{align*}
 \partial^\alpha V^a(x^a)=\vO\parb{\langle x^a\rangle^{-\mu-|\alpha|}}.
 \end{align*} Note that $\sqrt 3 -1\approx 0,732$. 
We then define cut-off pair potentials $W^a\in C^\infty(\mathbf X)$ as 
\begin{align*}
W^a(x)=\chi\bigl(\abs{x^a}/|x|^\mu\bigr)V^a(x^a),
\end{align*} 
where $\chi\in C^\infty(\mathbb{R})$ is chosen in agreement with 
\eqref{eq:14.1.7.23.24}.

\begin{lemma}\label{lemma:application-3-body}
The cut-off pair potentials $W^a$, 
$a\in\vA\setminus\{a_{\min},a_{\max}\}$, satisfy
Condition~\ref{cond:220525} with $l=2$. 
\end{lemma}
\begin{proof}
For any $\alpha\in \N_0^{\dim \mathbf X}$ we have $\partial^\alpha W^a(x)=\vO\parb{\langle x\rangle^{-\mu(|\alpha|+\mu)}}.$
 The condition 
 $\mu(2+\mu)>2$ (required for $ |\alpha|=2$) is fulfilled exactly for $\mu>\sqrt 3 -1$. 
\end{proof}

\begin{remark*}
As a consequence of Lemma \ref{lemma:application-3-body} our   results
apply to the $3$-body Hamiltonian 
using cut-off pair potentials of long-range type.
In particular one can in a stationary manner derive asymptotic
completeness for the $3$-body
long-range Hamiltonian 
along with formulas for scattering quantities, cf. \cite{Sk1} and  Remark 
\ref{remark:inverse-asympt-norm}. 
Note that the previous methods in the literature do not apply at this
point since the third derivatives of $W^a$ 
may not be of the form $\mathcal O(\langle x\rangle^{-3-\sigma})$
for some $\sigma\in(0,1]$.
\end{remark*}

\subsubsection{A further perspective: Generalization to manifolds}\label{subsubsec:A perspective: Generalization to manifolds}

Our arguments are not really dependent on a specific structure of the Euclidean space, 
but rather on solutions to the eikonal equations \eqref{eq:4c} and
\eqref{eq:4cb}, and the estimates of their derivatives. 
In fact we do not even use the (ordinary) Fourier transform, and
neither pseudodifferential operators nor advanced functional calculus. 
In our previous related works \cite{IS2,IS3} we studied in the same
spirit long-range scattering theory on a manifold with ends, 
employing approximate solutions to the eikonal equations. We did not develop a $C^2$ regularity theory using exact solutions as done in the
present paper in the Euclidean setting. Moreover we used weaker radiation conditions bounds 
entailing a more complicated construction of the stationary scattering
theory than presented here. For completeness of presentation let us
note that our older work \cite{IS4} may be seen as a $C^3$
scattering theory on a manifold in that it involves an exact $C^3$
solution to the (`free')
time-dependent eikonal equation, however this theory is entirely 
time-dependent and allows only short-range potentials.

If one could construct a `good solution' on a more general manifold,
say 
for a suitable $C^2$ perturbation (by the variational method of this
paper or by any other means), 
then the elementary techniques of the present paper would conceivably work there. 
Hence the methods of this paper potentially could also contribute to
 stationary scattering theory on manifolds, in particular to
 developing a more refined low
 regularity theory.

\subsection{Key bounds}

In our stationary scattering theory it is a major challenge to verify the WKB approximation \eqref{eq:221206}.
The following strong radiation condition bounds constitute a powerful tool for that verification,
and are themselves (along with Theorem~\ref{thm:main result2}) the most important technical novelty of the paper.

Let $S\in C^l(I\times (\mathbb R^d\setminus\{0\}))$ be given in
agreement with Theorem~\ref{thm:main result2}, 
and define the \textit{gamma observables} (or alternatively referred
to as
\emph{radiation observables}) 
\begin{align}
\begin{split}
\gamma&
=(\gamma_1,\dots, \gamma_d)
=p\mp (\nabla_x \chi_1S)
,\\ 
\gamma_{\|}
&
=\mathop{\mathrm{Re}}\bigl((\nabla_x \chi_1 S)\cdot \gamma\bigr)
=(\nabla_x\chi_1 S)\cdot \gamma
-\tfrac{\mathrm i}2(\Delta_x\chi_1 S)
,
\end{split}
\label{eq:6}
\end{align}
where we have included $\chi_1$ from \eqref{eq:14.1.7.23.24b} 
just to cut-off a singularity at the origin. 
We also set 
\begin{align}
\beta_c=\min\{2,1+\sigma+\rho\}
.
\label{eq:22120912}
\end{align}

\begin{thm}\label{thm:proof-strong-bound} 
Suppose Condition~\ref{cond:220525} with $l= 4$ and $q\equiv 0$. 
Let 
$I\subset \R_+$ be a closed interval and let 
 $S\in C^4(I\times(\mathbb R^d\setminus\{0\})$ be the 
 function from Theorem~\ref{thm:main result2} given for each $R\geq R_0$ (for some $R_0>0$),
and define correspondingly $\gamma_{j}$ and $\gamma_{ \|}$ as
above. Then the 
following bounds hold for any compact subset
$I'\subseteq I\subset \R_+$:
 \begin{subequations}
\begin{enumerate}
\item
For any $\beta\in (0,\beta_c)$ 
and $R\geq R_0$ there exists $C>0$, such that
 for all 
$\lambda\in I'$ and $\psi\in L^2_{\beta+1/2}$
\begin{align}\label{eq:gamma1a}
 \bigl\|\gamma_{\|} R(\lambda\pm\i 0)\psi\bigr\|_{L^2_{\beta-1/2}}&\leq
 {C}\norm{\psi}_{L^2_{\beta+1/2}},\\
\bigl\|\gamma_{i}\gamma_{j}R(\lambda\pm\i
 0)\psi\bigr\|_{L^2_{\beta-1/2}}&\leq
 {C}\norm{\psi}_{L^2_{\beta+1/2}};\quad i,j=1,\dots,d.\label{eq:gamma1b}
\end{align} 
\item
For any 
$\beta'\in (0,\beta_c/2)$, $t>1/2$ and $R\geq R_0$ 
there exists $C'>0$, such that for all 
$\lambda\in I'$ and $\psi\in L^2_{\beta'+t}$
\begin{align}
\label{eq:gamma2}
 \bigl\|\gamma_{i}R(\lambda\pm\i
 0)\psi\bigr\|_{L^2_{\beta'-t}}&\leq {C'}\norm{\psi}_{L^2_{\beta'+t}};\quad i=1,\dots,d.
\end{align}

\end{enumerate}
 \end{subequations}
\end{thm}

\begin{remarks}\label{rem:22110219}
\begin{enumerate}[1)]

 \item 
These bounds are `strong' in the sense that $\beta_c>1$,
while Isozaki \cite{Is} obtained similar bounds for classical $C^3$
long-range potentials with $\beta_c<1$. 
Such strong versions first appeared in \cite{HS} for 
classical $C^\infty$ potentials.

\item \label{item:bound3}
Our proof is elementary, 
relying only on \cs and the product rule for differentiation, 
and it admits considerably weaker assumptions than in \cite{Is,HS}.
In fact, it applies even to $C^3$ potentials by redefining appropriately
$\gamma_\|$ 
and $\beta_c>1$, see Remark~\ref{rem:230112}. 

\item 
See Subsection~\ref{subsec:200920} for a motivation from Classical Mechanics. 
The third or higher order derivatives of $V$ appear as purely `quantum effects', 
and there is no classical interpretation of their appearances. 
Note also that in our setting these derivatives might
 not have classical decay. 

\item We may let $I=\R_+$ by letting $C,C'>0$ from
 \eqref{eq:gamma1a}--\eqref{eq:gamma2} be dependent on $\lambda\in
 I$, cf.\ Remark \ref{remark:lambda_high-order-deriv}
 \ref{item:lambda_high-order-deriv1}. 

\item\label{item:2342619}
The potential here is more smooth than in the previous subsection. 
We will apply Theorem~\ref{thm:proof-strong-bound} 
 to a regularized potential appearing as a technical tool, see Lemma~\ref{lem:230111}.

\end{enumerate}
\end{remarks}

The rest of the paper is organized as follows. 
In Section~\ref{sec:constr-solut-eikon} we prove Theorem~\ref{thm:main result2}
following the scheme of \cite{CS}. 
Section~\ref{sec:proof-strong-bound four} is 
devoted to the proof of Theorem~\ref{thm:proof-strong-bound},
for which a commutator-type argument plays a central role. 
After these preliminaries, 
we prove Theorems~\ref{thm:comp-gener-four}, \ref{thm:char-gener-eigenf-1} and \ref{thm:221207} and Corollary~\ref{cor:230623}
in Section~\ref{sec:Generalized Fourier transform},
where we see that the strong radiation condition bounds provide simple and intuitive proofs.
Finally in Section~\ref{sec:Time-dependent theory} we prove Theorems~\ref{thm:2307032}, \ref{thm:230703} and \ref{thm:time-depend-theory2}.

\section{Eikonal equation}\label{sec:constr-solut-eikon}

In this section we prove Theorem~\ref{thm:main result2}. 
We will follow the framework of Cruz--Skibsted \cite{CS},
and solve \eqref{eq:2212116} in a somewhat abstract manner. 
The equation we investigate here is of the form
\begin{align}\label{eikeqFinal}
(\nabla \Phi)\cdot G^{-1}(\nabla \Phi)=1
\end{align} 
with $G$ being a given ($d\times d$)-matrix-valued function on $\mathbb R^d$,
which is assumed to be sufficiently close to the identity matrix $I_d$. 
Note that the \textit{potential} eikonal equation \eqref{eq:2212116} 
is always translated into the \textit{geometric} eikonal equation \eqref{eikeqFinal}
through the change of variables 
\begin{align}
 \Phi=(2\lambda)^{-1/2}S,\quad 
 G=\bigl(1-\lambda^{-1}\chi_R V\bigr)I_d.
\label{eq:230212}
\end{align} 
Thus the arguments of this section on \eqref{eikeqFinal} will
readily apply to an $S$ solving \eqref{eq:2212116}. 

The paper \cite{CS} adopts a variational method, 
and we recall the precise setting in Subsection~\ref{subsec:results-from-citecs}, 
quoting some results from there. 
The existence and smoothness of a solution to \eqref{eikeqFinal} obtained in \cite{CS} 
yield the corresponding assertions \eqref{item:22120919a}--\eqref{item:22120919c} of Theorem~\ref{thm:main result2}.
Then in Subsection~\ref{subsec:Improved bounds for a potential, C^2} we discuss 
the remaining problems, i.e. the smoothness and the uniform bounds \eqref{eq:errestza2}, 
under our more restrictive assumption of our paper. 
This completes the proof of Theorem~\ref{thm:main result2}. 
At the end of the section we present Corollary~\ref{cor:epsSmall}, 
a modification of \eqref{eq:errestza2}, which will be necessary in 
Sections~\ref{sec:Generalized Fourier transform} and \ref{sec:Time-dependent theory}.


\subsection{Terminologies and results from \cite{CS}}\label{subsec:results-from-citecs}

\subsubsection{Class of Riemannian metrics}

In this subsection we discuss the following class $\mathcal M_d^l$ of functions with values in square matrices
of order $d$,
or of Riemannian metrics on $\mathbb R^d$.
Denote the set of all the real symmetric matrices of order $d$ by $\mathcal S_d(\R)$.

\begin{defn}
For any $l\in\mathbb N_0$ we set 
\[
\widetilde{\mathcal M}^l_d
=\bigl\{G=(g_{ij})_{i,j}\colon \R^d\rightarrow \mathcal S_d(\R)\,\big|\,
G\text{ is of class }C^l,\ \|G\|_l<\infty\bigr\}\]
along with
 \begin{equation}
 \label{condoz} 
 \begin{split}
\|G\|_l
=\sup\bigl\{
&
\langle x\rangle^{|\alpha|}\left| \partial^\alpha g_{ij}(x)\right|
\,\big|\,
x\in\R^d,\ 
|\alpha|=0,\dots,l,\ 
i,j=1,\dots,d\bigr\}. 
\end{split}
 \end{equation}
 In addition, we set 
\begin{equation}
\label{ellcond}
\mathcal M^l_d
=\bigl\{G\in \widetilde{\mathcal M}^l_d \mid 
\exists a,b>0\text{ s.t.\ }\forall x\in\mathbb R^d:\
aI_d\le G(x)\le bI_d\bigr\}.
 \end{equation} 
\end{defn}

\begin{remark}
 Obviously
 $\mathcal M^l_d\subseteq\widetilde{\mathcal M}^l_d$ is
 open with respect to $\|\cdot\|_l$. Let $I$ be any closed
 interval in $\R_+$ (as in Theorem~\ref{thm:main result2}). 
If we then let $R>0$ be large enough,
 $G$ from \eqref{eq:230212} can be arbitrarily close to $I_d$ in the
 class $\mathcal M^2_d$ uniformly in $\lambda\in I$. This is due to
 the extra order of decay $\sigma\in (0,1)$ in
 \eqref{eq:cond22bb}. This uniformity will be vital for our proof of Theorem~\ref{thm:main result2}.
\end{remark}

\subsubsection{Energy functional and geodesics}

Next we define a \textit{geodesic} for $G\in\mathcal M^1_d$
in a variational manner. 
For any $p\in (1,\infty)$ we introduce the Banach space 
\begin{equation} 
\label{eq:54}
{
X^p=\big(W_0^{1,p}((0,1);\mathbb R)\bigr)^d,
}
\quad 
\|\kappa\|_p=\Bigg(\sum_{i=1}^d\int_0^1|\dot \kappa_i(t)|^p\,\mathrm dt\Biggr)^{1/p}
\ \ \text{for }\kappa\in X^p,
\end{equation}
although for the moment we shall only use the space {$X:=X^2$}, 
which is a Hilbert space equipped with the inner product
\begin{equation*}
 \langle \mu,\nu\rangle_X= \langle \mu,\nu\rangle=\langle I_d;\mu,\nu\rangle
=
{
\int_0^1\, \dot \mu(t)\cdot
\dot \nu(t)\,\d t;
}
\quad \mu,\nu\in X.
\end{equation*}

For any $\kappa\in X$ and $x\in\mathbb R^d$ 
define a path $\kappa_x\in (H^1(0,1))^d$ from $0$ to $x$ as 
\[\kappa_x(t)=tx+\kappa(t)\ \ \text{for } t\in [0,1].\]
Now, given $G\in \mathcal M^1_d$, we consider the \textit{energy functional}
\begin{equation}
\label{ener} 
\mathcal E\colon \R^d\times X\rightarrow \R,\quad 
\mathcal E(x,\kappa)=\int_0^1 \dot\kappa_x(t)\cdot G(\kappa_x(t)) \dot\kappa_x(t)\,\mathrm dt. 
\end{equation}

\begin{defn}\label{def:2206022356}
Let $x\in \mathbb R^d$, and denote the space of the paths $\kappa_x$ by 
\begin{equation}
\label{paths} 
X_x=
\bigl\{\kappa_x\in (H^1(0,1))^d\,\big|\, \kappa\in X\bigr\}
.
\end{equation}
We call $\gamma_x\in X_x$ a \emph{geodesic for $G\in\mathcal M^1_d$ from $0$
to $x$ in unit time} if the associated $\gamma\in X$ satisfies 
\begin{align}
(\partial_\kappa \mathcal E)(x,\gamma)=0\in X'
.
\label{eq:22060223}
\end{align}
\end{defn}
\begin{remarks} 
\begin{enumerate}[1)]
\item 
The duality pairing of 
$\partial_\kappa \mathcal E(x,\kappa)\in X'$ 
and $\mu\in X$ is directly computed as 
 \begin{align}\label{eq:Deriv}
\bigl\langle \partial_\kappa \mathcal E(x,\kappa),\mu\bigr\rangle
 =\int_0^1
\bigl\{2\dot \mu \cdot G(\kappa_x) \dot\kappa_x
+\dot\kappa_x \cdot(\mu_i\partial_i G)(\kappa_x)\dot\kappa_x
\bigr\}\,\d t.
\end{align}
\item 
Classically, a geodesic for $G=(g_{ij})_{i,j}$ 
is defined as a smooth curve $y=y(t)$ in $\mathbb R^d$ solving 
\begin{align}
\ddot y_i 
+\tfrac12g^{ij}\bigl(\partial_kg_{jl}+\partial_lg_{jk}
-\partial_jg_{kl}\bigr)\dot y_k\dot y_l
=0\ \ \text{for }i=1,\dots,d,
\label{eq:3b}
\end{align}
where $G^{-1}=(g^{ij})_{i,j}$, or equivalently 
\begin{equation}
\label{eq:3a}
G \ddot y
=
-((\dot y \cdot\nabla) G) \dot y
+\tfrac12\dot y\cdot (\nabla_\bullet G)\dot y
.
\end{equation}
Clearly \eqref{eq:22060223} is a weak form of \eqref{eq:3a} according to the expression \eqref{eq:Deriv}.
A geodesic in the sense of Definition~\ref{def:2206022356} 
is equivalent to an $H^1$ curve solving the $H^{-1}$ equation \eqref{eq:3a}.
Note that due to the Sobolev embedding theorem for $y\in (H^1(0,1))^d$ 
the right-hand side of \eqref{eq:3a} 
belongs to $(L^1(0,1))^d\subseteq (H^{-1}(0,1))^d$. 

\item
We shall use the notation $\gamma,\gamma_x$ exclusively when
discussing geodesics,
and $\kappa,\kappa_x$ or $\mu,\mu_x$ for general paths.
There would be no confusion since, as long as $G$ is close to $I_d$, a geodesic is unique, 
see Theorem~\ref{thm:main result}. 
\end{enumerate}
\end{remarks}

For any $G\in\mathcal M^2_d$ and $x\in\mathbb R^d$
there always exists at least one geodesic from $0$ to $x$, given as follows. 

\begin{lemma}\label{lemma:existence}
Let $G\in\vM^2_d$ with $a,b>0$ as in \eqref{ellcond}, and let $x\in\mathbb R^d$.
Then there exists a minimizer $\gamma\in X$ 
of $\mathcal E(x,\cdot)$, i.e.,
\[\mathcal E(x,\gamma)=\inf\{\mathcal E(x,\kappa)\,|\,\kappa\in X\},\]
and hence the associated $\gamma_x$ is a geodesic for $G$. 
Moreover, for any $t\in[0,1]$
\begin{subequations}
{
\begin{equation}
 \label{eq:estminb}
\mathcal E(x,\gamma)=\dot\gamma_x(t)\cdot G(\gamma_x(t))\dot\gamma_x(t)
\end{equation}
and}
\begin{equation}
 \label{eq:estmin}
 \tfrac{a}{b}|x|^2\leq |\dot\gamma_x(t)|^2 \leq\tfrac{b}{a} |x|^2,
\quad
\tfrac{a}{b}|tx|^2\leq |\gamma_x(t)|^2\leq\tfrac{b}{a} |tx|^2.
\end{equation} 
\end{subequations}
\end{lemma}
\begin{proof}
It is a part of \cite[Lemma~1]{CS}. We omit the proof.
\end{proof}

In general a geodesic is not necessarily associated with the minimizer of 
$\mathcal E(x,\cdot)$. 
However it is the case when $G$ is sufficiently close to $I_d$. 

 \begin{thm}\label{thm:main result} 
There exists a neighbourhood $U\subseteq \vM^2_d$ of $I_d$ such that
the following holds.
\begin{enumerate}

\item\label{item: Condition 1}
For any $G\in U$ and $x\in \R^d$ there exists a unique geodesic $\gamma_x\in X_x$.

\item\label{item: Condition 2}
There exists $c>0$ such that for any $G\in U$, $x\in \R^d$ and $\kappa\in X$
\begin{equation*}
\bigl\langle (\partial_\kappa^2 \mathcal E)(x,\gamma);\kappa,\kappa\bigr\rangle
\geq c\|\kappa\|_2^2.
\end{equation*}
\end{enumerate}
Moreover, for any $G\in U\cap \mathcal M^l_d$ with $l\ge 2$ the map $\R^d\ni x\mapsto \gamma\in X$ 
from \eqref{item: Condition 1} is of class $C^{l-1}$.
\end{thm}

\begin{remark}\label{remark:results-from-citecs-1} 
The second derivative of $\mathcal E(x,\cdot)$ reads for any $\kappa,\mu,\nu\in X$ as 
\begin{align}\label{hessianGb}
 \begin{split}
\bigl\langle \partial_\kappa^2 \mathcal E(x,\kappa); \mu,\nu \bigr\rangle
& =\int_0^1
\Bigl(
2\dot \mu\cdot G(\kappa_x)\dot \nu
+2 \dot \mu \cdot (\nu_i\partial_i G)(\kappa_x) \dot\kappa_x
\\&\phantom{{}={}}\qquad{}
+2\dot \nu\cdot(\mu_i\partial_i G)(\kappa_x)\dot\kappa_x 
+ \dot \kappa_x \cdot(\mu_i\nu_j\partial_i\partial_j G)(\kappa_x)\dot \kappa_x\Bigr)\,\d t.
 \end{split}
\end{align}
\end{remark}

\begin{proof}
The assertions \eqref{item: Condition 1} and \eqref{item: Condition 2} 
follows from \cite[Theorem~1~(a)]{CS}, if we allow $c>0$ to be dependent on $G\in U$. 
However, if we let $U$ be small enough, 
we can choose $c$ uniformly in $G\in U$ due to \cite[Lemma~6]{CS}. 
The last assertion follows from \cite[Proposition~2 1)]{CS}. 
Hence we are done. 
\end{proof}

\subsubsection{Solution to the eikonal equation}

We are ready to construct and study a specific solution to the eikonal equation \eqref{eikeqFinal}. 

 \begin{thm} \label{thm:results-from-citecs-1} 
Let $U\subseteq\mathcal M^2_d$ be given as in Theorem~\ref{thm:main
 result}, $l\ge 2$ 
and define for each $G\in U\cap \mathcal M^l_d$ 
 \begin{equation}
 \label{maxsol} 
\Phi\colon \R^d\to [0,\infty)
 ,\quad 
\Phi(x)=\mathcal E(x,\gamma)^{1/2}
=\bigl(\inf\{\mathcal E(x,\kappa)\,|\, \kappa\in X\}\bigr)^{1/2} 
.
 \end{equation}
Then for any $x\in \R^d$ the identity 
 \begin{align}
& \nabla (\Phi(x)^2)=2G(x)\dot\gamma_x(1)
\label{eq:10}
\end{align} 
holds. 
Moreover, $\Phi$ is a $C^l$ solution to \eqref{eikeqFinal} 
on $\mathbb R^d\setminus\{0\}$.
\end{thm}

\begin{proof}
The assertion is due to \cite[(36)]{CS} and \cite[Proposition~2.2)]{CS}, 
and we omit the proof. 
\end{proof}

Using the above results we can prove a part of Theorem~\ref{thm:main result2}. 
The proofs of the $(\lambda,x)$-smoothness and (\ref{item:22120919d})
will given in Subsection \ref{subsec:Improved bounds for a potential, C^2}.

\begin{proof}[Proof of a part of Theorem~\ref{thm:main result2}]
Suppose Condition~\ref{cond:220525} for some 
$l\geq 2$ and let $I$ be any given closed interval in $\R_+$ (as in the theorem).

\smallskip
\noindent
\textit{(\ref{item:22120919a})}\ 
Here we let 
\[G(\lambda,x)=\bigl(1-\lambda^{-1}\chi_R(x)V(x)\bigr)I_d\] 
as in \eqref{eq:230212}. 
Let $U\subseteq\mathcal M^2_d$ be given as in Theorem~\ref{thm:main result},
and take large $R>0$,
so that $G(\lambda,\cdot)\in U\cap \mathcal M^l_d$ for all $\lambda\in I$.
Then for such $G(\lambda,\cdot)$ we can find a solution 
$\Phi(\lambda,\cdot)\in C^l(\mathbb R^d\setminus\{0\})$ to \eqref{eikeqFinal} 
by Theorem~\ref{thm:results-from-citecs-1}.
Thus, if we set 
\begin{align}
S(\lambda,\cdot)=(2\lambda)^{1/2}\Phi(\lambda,\cdot)\in C^l((\mathbb R^d\setminus\{0\})
\label{eq:230323}
\end{align} 
as in \eqref{eq:230212}, it obviously solves \eqref{eq:2212116}. 
We actually have joint smoothness $S\in C^l(I\times(\mathbb R^d\setminus\{0\}))$,
but we will verify it later.

\smallskip
\noindent
\textit{(\ref{item:22120919})}\ 
By the definition \eqref{maxsol} and \eqref{eq:estminb} it follows
 that 
\[
\Phi(\lambda,x)=\int_0^1\mathcal E(x,\gamma)^{1/2}\,\mathrm dt
=\int_0^1\sqrt{\dot\gamma_x(t)\cdot G(\gamma_x(t))\dot\gamma_x(t)}\,\mathrm dt.
\]
Hence $\Phi(\lambda,x)$ is the distance with
 respect to the Riemannian metric $(2\lambda)^{-1}g$ from the
 origin to $x$ 
along
$\gamma_x$. Using again \eqref{eq:estminb} and the fact that $\Phi$
solves \eqref{eikeqFinal} it follows that indeed
$\Phi(\lambda,x)$ is the geodesic distance from the origin to $x$.
Obviously then $S(\lambda,\cdot)$ is the geodesic distance from
the origin with respect to $g$. 

\smallskip
\noindent
\textit{(\ref{item:22120919c})}\ 
Set for any $x\in \mathbb R^d\setminus\{0\}$
\[s(\lambda,x)=|x|^{-1}\Phi(\lambda,x)-1=(2\lambda)^{-1/2}|x|^{-1}S(\lambda,x)-1.\]
Then $s(\lambda,\cdot)\in C^l(\mathbb R^d\setminus\{0\})$, 
and it satisfies the asserted identity. 
Furthermore it extends 
smoothly to the origin by letting $s(\lambda,0)=0$. 
In fact we have $G(\lambda,x)=I$ for $|x|\le R$, and 
this implies $\Phi(\lambda,x)=|x|$ or $s(\lambda,x)=0$ there. 
The joint smoothness of $s$ follows from that of $S$
 (in turn to be given in Subsection \ref{subsec:Improved bounds for a potential, C^2}). 
\end{proof}

\subsection{Improved bounds}
\label{subsec:Improved bounds for a potential, C^2}

In this subsection we prove the remaining assertions of Theorem~\ref{thm:main result2}.
As for (\ref{eq:errestza2}), note that the similar bounds from \cite[Theorem~1 (b)]{CS} do not suffice, 
so we have to argue properly with the more restrictive assumptions of
the present paper. 
The following proposition is a key for our proof. 
See \eqref{eq:54} for the definitions of $X^p$ and
$\|\cdot\|_p$. The appearing curve $\gamma\in X$ comes
from using Theorem~\ref{thm:results-from-citecs-1} with the considered $G$.

\begin{proposition}\label{prop:main result2bbf} 
Suppose Condition~\ref{cond:220525} for some 
$l\geq 2$ and let $I$ be any given closed
 interval in $\R_+$ (as in Theorem~\ref{thm:main result2}).
Consider the matrix $G=\bigl(1-\lambda^{-1}\chi_RV\bigr)I_d$ with
$\lambda\in I$ and $R>0$. 

For any $\delta\in (0,1)$ there exists $C>0$ such that 
uniformly in all sufficiently large $R>0$, $k+|\alpha|=1,\dots,l-1$, 
$p\in [1+\delta,1+\delta^{-1}]$, 
$\sigma'\in(0,\sigma]$ with $\sigma'p\le 1-\delta$ and 
$(\lambda,x)\in I\times \mathbb R^d$, the corresponding $\gamma\in X$
obeys 
\begin{align}
 \label{eq:722}
 \begin{split}
\|\partial_\lambda^k\partial_x^\alpha\gamma\|_{p}\leq C
\lambda^{-1-k}
 |x|^{2-m'(k+\abs{\alpha}+1)+k},
 \end{split}
\end{align} 
where $m'(k)=m(k)-\sigma+\sigma'$. 
\end{proposition}
\begin{remarks}\label{rem:230309}
\begin{enumerate}[1)]
\item\label{item:230313}
Note that $\gamma=0$ for $|x|\le R$.
Thus $|x|$ on the right-hand side of \eqref{eq:722} can be replaced by $\langle x\rangle$.
\item
Obviously the strongest $x$-decay is achived choosing
 $\sigma'=\sigma$ in \eqref{eq:722}. This will be used in the proof of \eqref{eq:errestza2} 
 simply by using the estimate for any $p\in (1,\infty)$ with $\sigma p<1$. 
The (local) uniformity in $p$ and $\sigma'$ in the assertion is technically necessary 
for the proof of Proposition~\ref{prop:main result2bbf} itself, 
which relies on induction in $l$. However our proof does not extend to
the assertion of having uniformity in $p\in (1,\infty)$ rather than
the stated condition $p\in [1+\delta,1+\delta^{-1}]$. 
\end{enumerate}
\end{remarks}

In the proof of Proposition~\ref{prop:main result2bbf}
we will repeatedly use the following \textit{Hardy} 
and \textit{generalized H\"older inequalities}.

\begin{lemma}\label{lemma:55}
For any $p\in(1,\infty)$ and $\kappa \in X^p$ one has 
 \begin{align}
 \label{eq:55}
 \Bigg(\sum_{i=1}^d\int_0^1|t^{-1}\kappa_i(t)|^p\,\mathrm dt\Biggr)^{1/p}&\leq 
 \tfrac{p}{p-1}
 \|\kappa \|_{p}.
 \end{align} 
\end{lemma}
\begin{proof}
{
The assertion follows from the one for $d=1$, which is the well-known Hardy inequality. 
We omit the details. 
}
\end{proof}

\begin{lemma}\label{lem:Holder}
Let $p_1,\dots,p_n\in [1,\infty]$, $n\in\mathbb N$, satisfy 
$p_1^{-1}+\dots+ p_n^{-1}= 1$. 
Then for any $f_i\in L^{p_i}(0,1)$, $i=1,\dots,n$, one has 
 \begin{equation}
 \label{eq:Holder}
 \int_0^1|f_1(t)|\cdots |f_n(t)|\,\d t 
\leq \|f_1\|_{L^{p_1}}\cdots \|f_n\|_{L^{p_n}}.
 \end{equation}
\end{lemma}
\begin{proof}
The assertion follows easily by repeated application of the familiar H\"older's inequality.
We omit the details.
\end{proof}

The below lemma will be useful in a reduction procedure in the proof of Proposition~\ref{prop:main result2bbf}, 
see Steps~I, III and IV there. 
\begin{lemma}\label{lemma:230221}
For any $\delta\in (0,1)$ 
take a sufficiently small neighbourhood $U\subseteq\mathcal M_d^2$ of $I_d$.
Then it follows that uniformly in 
$p,q\in [1+\delta,1+\delta^{-1}]$ with $p^{-1}+q^{-1}=1$, $G\in U$, $x\in\mathbb R^d$, 
$\mu\in X\cap X^p$ and $\nu\in X\cap X^q$, the
corresponding $\gamma\in X$ from Theorem \ref{thm:main result} obeys 
\begin{equation*}
\bigl|\bigl\langle (\partial^2_\kappa \mathcal E)(x,\gamma)-2I_d;\mu,\nu\bigr\rangle\bigr|
\leq \tfrac12\|\mu\|_{p}\|\nu\|_{q}.
 \end{equation*} 
\end{lemma}
\begin{proof}
Take any $\delta\in (0,1)$. 
We first let $U\subseteq\mathcal M^2_d$ be small enough that 
not only Theorem~\ref{thm:main result} is available,
but also we can find $a,b>0$ as in \eqref{ellcond} uniformly in $G\in U$. 
Then by \eqref{hessianGb}, 
\eqref{eq:estmin}, H\"older's inequality and Lemma~\ref{lemma:55}
we can bound 
 \begin{align*}
\bigl|\bigl\langle (\partial^2_\kappa \mathcal E)(x,\gamma)-2I_d;\mu,\nu\bigr\rangle\bigr|
&\leq 
\int_0^1
\Bigl|
2\dot \mu\cdot (G-I_d)\dot \nu
+2 \dot \mu \cdot (\nu_i\partial_i G) \dot\gamma_x
\\&\qquad\quad
+2\dot \nu\cdot(\mu_i\partial_i G)\dot\gamma_x 
+ \dot \gamma_x \cdot(\mu_i\nu_j\partial_i\partial_j G)\dot \gamma_x\Bigr|\,\d t
\\&
\le 
C_1\|G-I_d\|_2
\int_0^1
\Bigl(
|\dot \mu||\dot \nu|
+|\dot \mu||\nu|\langle tx\rangle^{-1} |x|
\\&\qquad\qquad\qquad\qquad\ 
+|\dot \nu||\mu|\langle tx\rangle^{-1} |x|
+|\mu||\nu|\langle tx\rangle^{-2}|x|^2\Bigr)\,\d t
\\&
\le 
C_2\|G-I_d\|_2\|\mu\|_p\|\nu\|_q
, 
\end{align*}
where $C_1,C_2>0$ only depend on the constants $\delta,a,b$. 
Thus, possibly by letting $U$ be smaller, the assertion follows. 
\end{proof}

{With} these preparations we can now prove Proposition~\ref{prop:main result2bbf}.

\begin{proof}[Proof of Proposition~\ref{prop:main result2bbf}]
The proof proceeds by induction in $l$ with the
 interval $I\subset \R_+$ being fixed. Note that the assertion of the
 proposition amounts to the statement that for each
 $l\geq 2$ the bounds \eqref{eq:722} hold uniformly in various
 parameters. This is a statement amenable to 
 induction.

Fix any 
 $l\geq 2$ and $I\subset \R_+$ as in the proposition.
Fix also any $\delta\in(0,1)$. Then 
the conclusion of Lemma~\ref{lemma:230221} is available for 
$G=(1-\lambda^{-1}\chi_RV)I$ uniformly in $\lambda\in I$ and all
sufficiently large $R>0$.
In particular it follows 
from the classical implicit function theorem (see e.g.\
\cite[Theorem~15.1]{De} or \cite[Theorem~C.7]{Ir}) applied to the equation 
\eqref{eq:22060223} with parameter $(\lambda,x)$, that 
the mapping
\begin{align}
I\times\mathbb R^d\to X,\ \ (\lambda,x)\mapsto \gamma,
\label{eq:24050920}
\end{align}
is of class $C^{l-1}$, cf. Theorem \ref{thm:main result}.

\smallskip
\noindent
\textit{Step I.}\ 
We first let $l=2$,
and we start with the case $k=0$ and $|\alpha|=1$. 
It suffices to show that 
we have 
uniformly in $i=1,\dots,d$, $p,q\in [1+\delta,1+\delta^{-1}]$ with $p^{-1}+q^{-1}=1$, 
$\sigma'\in(0,\sigma]$ with $\sigma'p\le 1-\delta$, $(\lambda,x)\in I\times \mathbb R^d$
and $\mu\in X\cap X^q$ 
\begin{equation}
 \label{eq:712}
 \bigl|\bigl\langle
(\partial^2_\kappa \mathcal E)(\lambda,x,\gamma);\partial_i\gamma,\mu\bigr\rangle\bigr|
\leq C_1\lambda^{-1}|x|^{-\sigma'}\|\mu\|_{q}.
 \end{equation} 
In fact if \eqref{eq:712} holds true, 
 we obtain in
combination with Lemma~\ref{lemma:230221} that 
\begin{align}
\begin{split}
|\langle \partial_i\gamma,\mu\rangle|
&\le 
\tfrac12
\bigl|\bigl\langle 2I_d-(\partial^2_\kappa \mathcal E)(\lambda,x,\gamma);\partial_i\gamma,\mu\bigr\rangle\bigr|
+
\tfrac12
\bigl|\bigl\langle (\partial^2_\kappa \mathcal E)(\lambda,x,\gamma);\partial_i\gamma,\mu\bigr\rangle\bigr|
\\&
\le 
\tfrac14\|\partial_i\gamma\|_{p}\|\mu\|_{q}
+
\tfrac{C_1}2\lambda^{-1}|x|^{-\sigma'}\|\mu\|_{q}.
\end{split}
\label{eq:230305}
\end{align}
Note that any $\mu\in X\cap X^q$ takes the form
 $\mu(t)=\int_0^t\parb{g(s)-\int_0^1g(\tau)\d\tau}\d s$ with $g\in
L^2((0,1);\R)^d\cap L^q((0,1);\R)^d$, and vice versa. Using arbitrary $g$ in
this class, 
 the standard {$L^p$-$L^q$ duality} argument and density of
 $L^2\cap L^q$ in $ L^q$ allow us to compute
 \begin{align*}
 \|\partial_i\dot\gamma\|_{L^p}&=\sup_{\|g\|_{L^q}\leq 1}\,
|\inp{\partial_i\dot\gamma,g}_{L^2}|=\sup_{\|g\|_{L^q}\leq 1}\,
|\inp{\partial_i\gamma,\mu_g}|;\\
&\mu_g(t)=\int_0^t\parbb{g(s)-\int_0^1g(\tau)\d\tau}\d s,\quad t\in (0,1). 
 \end{align*} Noting that
 \begin{equation*}
 \|\mu_g\|_{q}=\|\dot\mu_g\|_{L^q}\leq 2 \|g\|_{L^q},
 \end{equation*}
it now follows from \eqref{eq:230305} that
\[
 \|\partial_i\gamma\|_{p}
\le 
\tfrac12\|\partial_i\gamma\|_{p}
+
C_1\lambda^{-1}|x|^{-\sigma'},
\]
yielding the assertion \eqref{eq:722} for $k=0$ and $|\alpha|=1$ with $C=C_2=2C_1$.

\smallskip
\noindent
\textit{Step II.}\ 
Here we prove \eqref{eq:712}. For that 
we first claim that for any $\mu\in X$ 
\begin{align}\label{eq:Deriv1}
 \begin{split}
\bigl\langle (\partial^2_\kappa \mathcal E)(\lambda,x,\gamma);\partial_i\gamma,\mu\bigr\rangle
 ={}&\int_0^1
\Bigl(
2\dot \mu \cdot (I_d-G(\lambda,\gamma_x)) e_i
-2t\dot \mu \cdot (\partial_iG)(\lambda,\gamma_x) \dot\gamma_x
\\&\qquad{}
-2\dot\gamma_x \cdot(\mu_j\partial_j G)(\lambda,\gamma_x)e_i
\\&\qquad{}
-t\dot\gamma_x \cdot(\mu_j\partial_i \partial_jG)(\lambda,\gamma_x)\dot\gamma_x
\Bigr)\,\d t.
 \end{split} 
\end{align}
In fact, if we differentiate \eqref{eq:22060223} in $x_i$, it follows that 
for any $\mu\in X$
\begin{align}
\bigl\langle(\partial_i\partial_\kappa \mathcal E)(\lambda,x,\gamma),\mu\bigr\rangle
+\bigl\langle (\partial^2_\kappa \mathcal E)(\lambda,x,\gamma);\partial_i\gamma,\mu\bigr\rangle
=0.
\label{eq:230306}
\end{align}
On the other hand, also differentiating \eqref{eq:Deriv} in $x_i$, 
we have for any $\kappa, \mu\in X$
 \begin{align*}
\bigl\langle \partial_i\partial_\kappa \mathcal E(\lambda,x,\kappa),\mu\bigr\rangle
 =\int_0^1
\Bigl(&
2\dot \mu \cdot G(\lambda,\kappa_x) e_i
+2t\dot \mu \cdot (\partial_iG)(\lambda,\kappa_x) \dot\kappa_x
\\&{}
+2\dot\kappa_x \cdot(\mu_j\partial_j G)(\lambda,\kappa_x)e_i
+t\dot\kappa_x \cdot(\mu_j\partial_i \partial_jG)(\lambda,\kappa_x)\dot\kappa_x
\Bigr)\,\d t.
\end{align*}
We obtain \eqref{eq:Deriv1} by 
combining the above identities and using that $\mu(0)=\mu(1)=0$. 

Next we estimate the right-hand side of \eqref{eq:Deriv1}. 
By using \eqref{eq:estmin}, H\"older's inequality and the Hardy inequality \eqref{eq:55}
we have uniformly in the relevant parameters
\begin{align*}
\bigl|\bigl\langle (\partial^2_\kappa \mathcal E)(\lambda,x,\gamma);
\partial_i\gamma,\mu\bigr\rangle\bigr|
&\le 
C_3\lambda^{-1}\int_0^1
\Bigl(
|\dot \mu| |tx|^{-\sigma'}
+t|\dot \mu| |tx|^{-1-\sigma'}|x|
\\&\phantom{{}={}}\qquad\qquad\quad{}
+|\mu||tx|^{-1-\sigma'}|x| 
+t|\mu||tx|^{-2-\sigma'}|x|^2
\Bigr)\,\d t
\\&\le 
C_4\lambda^{-1}|x|^{-\sigma'}\int_0^1
\bigl(
|\dot \mu| t^{-\sigma'}
+|\mu|t^{-1-\sigma'} 
\bigr)\,\d t
\\&
\le 
C_5\lambda^{-1}|x|^{-\sigma'}\|\mu\|_q.
\end{align*}
We have shown \eqref{eq:712} and hence indeed \eqref{eq:722} for $k=0$ and $|\alpha|=1$.

\smallskip
\noindent
\textit{Step III.}\ 
Next we prove the assertion for $k=1$ and $|\alpha|=0$
along the lines of Steps~I and II. 
In fact, we can first reduce it to proving 
\begin{equation}
 \label{eq:712b}
 \bigl|\bigl\langle
(\partial^2_\kappa \mathcal E)(\lambda,x,\gamma);\partial_\lambda\gamma,\mu\bigr\rangle\bigr|
\leq C_6\lambda^{-2}|x|^{1-\sigma'}\|\mu\|_{q}
 \end{equation} 
uniformly in the relevant parameters. 
The reasoning is the same as in Step I with Lemma~\ref{lemma:230221}, 
and we omit it. 
Then in order to prove \eqref{eq:712b} we deduce the
 following expression valid for any $\mu\in X$: 
\begin{align}\label{eq:Deriv1b}
 \begin{split}
\bigl\langle (\partial^2_\kappa \mathcal E)(\lambda,x,\gamma);\partial_\lambda\gamma,\mu\bigr\rangle
 ={}&\int_0^1
\bigl(
-2\dot \mu \cdot (\partial_\lambda G)(\lambda,\gamma_x) \dot\gamma_x
\\&\qquad{}
-\dot\gamma_x \cdot(\mu_j\partial_j\partial_\lambda G)(\lambda,\gamma_x)\dot\gamma_x
\bigr)\,\d t.
 \end{split} 
\end{align}
Again, this formula follows in a manner similar to the first part of Step~II
(similar, but actually slightly simpler),
and we omit the proof. 
Then we bound the right-hand side of \eqref{eq:Deriv1b} 
by using \eqref{eq:estmin}, H\"older's inequality and the Hardy inequality \eqref{eq:55}
\begin{align*}
\bigl|\bigl\langle (\partial^2_\kappa \mathcal E)(\lambda,x,\gamma);
\partial_\lambda\gamma,\mu\bigr\rangle\bigr|
&\le 
C_7\lambda^{-2}\int_0^1
\bigl(
|\dot \mu| |tx|^{-\sigma'}|x|
+|\mu||tx|^{-1-\sigma'}|x|^2
\bigr)\,\d t
\\&\le 
C_8\lambda^{-2}|x|^{1-\sigma'}\int_0^1
\bigl(
|\dot \mu| t^{-\sigma'}
+|\mu|t^{-1-\sigma'} 
\bigr)\,\d t
\\&
\le 
C_9\lambda^{-2}|x|^{1-\sigma'}\|\mu\|_q.
\end{align*}
This amounts to \eqref{eq:712b}, and we have shown \eqref{eq:722} for $k=1$ and $|\alpha|=0$.

\smallskip
\noindent
\textit{Step IV.}\ 
From here to the end of the proof we let $l\ge 3$. 
It suffices to discuss only the case $k+|\alpha|=l-1$
since the cases $k+|\alpha|=1,\dots, l-2$ follow by the induction hypothesis. 
In the following let for short 
\[
\beta=(k,\alpha)\in\mathbb N_0^{d+1},\quad 
|\beta|=k+|\alpha|=l-1,\quad 
\partial^\beta=\partial_{\lambda,x}^\beta=\partial_\lambda^k\partial_x^\alpha.
\]
Here we only note that the proof is reduced to verifying 
\begin{equation}
 \label{eq:10e}
 \bigl|\bigl\langle
(\partial^2_\kappa \mathcal E)(\lambda,x,\gamma);\partial^\beta\gamma,\mu\bigr\rangle\bigr|
\leq C_{1}\lambda^{-1-k}|x|^{2-m'(l)+k}\|\mu\|_{q}
\end{equation} 
uniformly in $\mu\in X\cap X^q$ and the relevant parameters. 
This is indeed a valid reduction, which may be seen by mimicking
Step I. Again we omit the details.

\smallskip
\noindent
\textit{Step V.}\ 
To show \eqref{eq:10e} we proceed in parallel to Steps II and III,
 differentiating \eqref{eq:22060223}. Indeed repeated differentiation of \eqref{eq:22060223}
in $(\lambda,x)$ yields the Fa\`a di Bruno formula 
\begin{equation}\label{eq:77}
\begin{split}
&\bigl\langle(\partial^2_\kappa \mathcal E)(\lambda,x,\gamma);\partial^\beta\gamma,\mu\bigr\rangle
\\& =\sum_{n,\beta^0,\dots,\beta^n}
 C_*
\bigl\langle(\partial^{\beta^0}\partial^{n+1}_\kappa\mathcal E)(\lambda,x,\gamma);
\partial^{\beta^1}\gamma,\dots, \partial^{\beta^n}\gamma, \mu\bigr\rangle,
\end{split}
\end{equation} 
where the indices $n\in \mathbb N_0$, $\beta^i=(k^i,\alpha^i)\in \mathbb N_0\times\mathbb N_0^d$, 
$i=0,\dots,n$, run over 
\begin{align}
\begin{split}
&n=0,\dots,l-1,\quad \beta^0+\dots+\beta^n=\beta
,\quad 
1\le |\beta^i|\le l-2\ \ \text{for }i=1,\dots,n. 
\end{split}
\label{eq:2303090}
\end{align}
Here and below combinatorial constants are denoted simply by $C_*\in\mathbb \N$
without distinction.
 
Let us further write down a detailed expression of the summand of \eqref{eq:77}. 
Directly differentiating the definition \eqref{ener}, 
we can compute $\partial^{n+1}_\kappa \mathcal E(\lambda,x,\kappa)$ as 
\begin{align*}
\bigl\langle \partial^{n+1}_\kappa \mathcal E(\lambda,x,\kappa);\mu^1,\dots,\mu^{n+1}\bigr\rangle 
& =
\sum_{|\eta|=n+1}C_*
\int_0^1(\partial_x^\eta g_{ab})(\dot\kappa_x)_a (\dot\kappa_x)_b
\mu^1_{c_1}\cdots\mu^{n+1}_{c_{n+1}}\,\d t
\\&\phantom{{}={}}{}
 +\sum_{|\eta|=n}C_*\int_0^1(\partial_x^\eta g_{ab}) (\dot\kappa_x)_b\mu^1_{c_1}\cdots
\dot\mu^i_a\cdots\mu^{n+1}_{c_{n+1}}\,\d t
\\&\phantom{{}={}}{}
 +\sum_{|\eta|=n-1}C_*\int_0^1(\partial_x^\eta g_{ab}) 
\mu^1_{c_1}\cdots\dot\mu^i_a\cdots\dot\mu^j_b\cdots\mu^{n+1}_{c_{n+1}}\,\d t
 .
\end{align*}
Here and
 henceforth $\dot\mu^i_a$ and $\dot\mu^j_b$ replace the corresponding factors, 
and for short we omit appropriate summations in 
 $a,b,c_1,\dots,c_{n+1},i,j$ as well as a
 specification of 
 the combinatorial constants $C_*\in\N$. Also we abbreviate $(\partial_x^\eta g_{ab})(\kappa_x)=(\partial_x^\eta g_{ab})$. 
We further proceed, applying $\partial^{\beta^0}$, as 
\begin{align*}
&\bigl\langle \partial^{\beta^0}\partial^{n+1}_\kappa \mathcal E(\lambda,x,\kappa);
\mu^1,\dots,\mu^{n+1}\bigr\rangle 
\\& 
=
\sum_{\genfrac{}{}{0pt}{}{|\eta|=n+1,}{e_a+e_b\le \alpha^0}}C_*
\int_0^1t^{|\alpha^0|-2}\bigl(\partial^{(k^0,\alpha^0-e_a-e_b+\eta)} g_{ab}\bigr)
\mu^1_{c_1}\cdots\mu^{n+1}_{c_{n+1}}\,\d t
\\&\phantom{{}={}}{}
 +\sum_{\genfrac{}{}{0pt}{}{|\eta|=n+1,}{e_b\le \alpha^0}}C_*
\int_0^1t^{|\alpha^0|-1}\bigl(\partial^{(k^0,\alpha^0-e_b+\eta)} g_{ab}\bigr)(\dot\kappa_x)_a 
\mu^1_{c_1}\cdots\mu^{n+1}_{c_{n+1}}\,\d t
\\&\phantom{{}={}}{}
 +\sum_{|\eta|=n+1}C_*
\int_0^1t^{|\alpha^0|}\bigl(\partial^{(k^0,\alpha^0+\eta)} g_{ab}\bigr)(\dot\kappa_x)_a (\dot\kappa_x)_b
\mu^1_{c_1}\cdots\mu^{n+1}_{c_{n+1}}\,\d t
\\&\phantom{{}={}}{}
 +\sum_{\genfrac{}{}{0pt}{}{|\eta|=n,}{e_b\le\alpha^0}}C_*
\int_0^1t^{|\alpha^0|-1}\bigl(\partial^{(k^0,\alpha^0-e_b+\eta)} g_{ab}\bigr) \mu^1_{c_1}\cdots\dot\mu^i_a\cdots\mu^{n+1}_{c_{n+1}}\,\d t
\\&\phantom{{}={}}{}
 +\sum_{|\eta|=n}C_*\int_0^1t^{|\alpha^0|}\bigl(\partial^{(k^0,\alpha^0+\eta)} g_{ab}\bigr) (\dot\kappa_x)_b\mu^1_{c_1}\cdots\dot\mu^i_a\cdots\mu^{n+1}_{c_{n+1}}\,\d t
\\&\phantom{{}={}}{}
 +\sum_{|\eta|=n-1}C_*\int_0^1t^{|\alpha^0|}\bigl(\partial^{(k^0,\alpha^0+\eta)} g_{ab}\bigr) \mu^1_{c_1}\cdots\dot\mu^i_a\cdots\dot\mu^j_b\cdots\mu^{n+1}_{c_{n+1}}\,\d t
 .
\end{align*}
We substitute $\kappa=\gamma, 
\mu^1=\partial^{\beta^1}\gamma,\dots,\mu^n=\partial^{\beta^n}\gamma$
and $\mu^{n+1}=\mu$
in the above formula, and denote the corresponding
six types of integrals on the
right-hand side simply as $I_1,\dots,I_6$, respectively.
Then it suffices to estimate each such integral in agreement with \eqref{eq:10e}. 

\smallskip
\noindent
\textit{Step VI.}\ 
Finally we bound the above $I_1,\dots,I_6$. 
They are treated similarly 
by using \eqref{eq:cond22bb}, \eqref{eq:estmin}, 
the generalized H\"older inequality \eqref{eq:Holder}, 
the Hardy inequality \eqref{eq:55} and the induction hypothesis. 
We may let $|x|\ge R$ below, see Remark~\ref{rem:230309} \ref{item:230313}.

We carefully bound $I_1$,
while we record only key steps for $I_2,\dots,I_6$ (they are treated similarly). 
Use \eqref{eq:cond22bb}, \eqref{eq:estmin} and $|\eta|=n+1$, to bound it for $|x|\ge R$ as 
\begin{align}
\begin{split}
|I_1|
&
\le 
C_{2}\lambda^{-1-k^0}
\int_0^1
t^{|\alpha^0|-2}|tx|^{-m'(|\alpha^0|+|\eta|-2)}
\bigl|\partial^{\beta^1}\gamma_{c_1}\bigr|
\cdots
\bigl|\partial^{\beta^n}\gamma_{c_n}\bigr|\bigl|\mu_{c_{n+1}}\bigr|\,\d t
\\&
\le 
C_{3}\lambda^{-1-k^0}|x|^{-m'(|\alpha^0|+n-1)}
\\&\phantom{{}={}}{}
\cdot\int_0^1
t^{-\sigma'}
\bigl|t^{-1}\partial^{\beta^1}\gamma_{c_1}\bigr|
\cdots
\bigl|t^{-1}\partial^{\beta^n}\gamma_{c_n}\bigr|
\bigl|t^{-1}\mu_{c_{n+1}}\bigr|\,\d t
.
\end{split}
\label{eq:23030916}
\end{align}
We take large $p_0,\dots,p_n\in (1,\infty)$ and small
 $\sigma'_1,\dots,\sigma'_n\in (0,\sigma]$ such that 
 \begin{equation*}
 p_0^{-1}+\dots+p_{n+1}^{-1}=1-q^{-1}=p^{-1},\quad \sigma'p_0<1,\quad 
\sigma'_1p_1<1,\ \dots,\ \sigma'_{n+1}p_{n+1}<1,
 \end{equation*}
and apply the generalized H\"older inequality \eqref{eq:Holder} and the Hardy inequality \eqref{eq:55}
to \eqref{eq:23030916}. Note that here and henceforth $\delta\in
 (0,1)$, $p\in [1+\delta,1+\delta^{-1}]$ and 
$\sigma'\in(0,\sigma]$ with $\sigma'p\le 1-\delta$ are given parameters
as in Proposition \ref{prop:main result2bbf}. 
Then we use the induction hypothesis with a
 sufficiently smaller $\delta\in (0,1)$ (to make sure that $p_i\in
 [1+\delta,1+\delta^{-1}]$ and $\sigma'p_i\le 1-\delta$) 
and with $m'_i(k)=m(k)-\sigma+\sigma_i'$.

Noting also \eqref{eq:2303090} and $|\beta^0|-k^0+n+1=|\alpha^0|+n+1\ge 2$
by $e_a+e_b\le \alpha^0$, 
it follows that for $|x|\ge R$
\begin{align}
\begin{split}
|I_1|
&\le 
C_{4}
\lambda^{-1-k^0}|x|^{2-m'(|\alpha^0|+n+1)}
\bigl\|\partial^{\beta^1}\gamma\bigr\|_{p_1}
\cdots
\bigl\|\partial^{\beta^n}\gamma\bigr\|_{p_n}\|\mu\|_{q}
\\&
\le C_{5}
\lambda^{-1-k}|x|^{2-m'(|\beta^0|-k^0+n+1)+(2-m_1'(|\beta^1|+1)+k^1)+\dots+(2-m_n'(|\beta^n|+1)+k^n)}\|\mu\|_{q}
\\&
\le C_{6}
\lambda^{-1-k}|x|^{2-m'(|\beta|-k^0+1)+k-k^0}\|\mu\|_{q}
\\&
\le C_{7}
\lambda^{-1-k}|x|^{2-m'(l)+k}\|\mu\|_{q}
.
\end{split}
\label{eq:230327}
\end{align}
This agrees with \eqref{eq:10e}, as wanted.
 
As for $I_2,\dots,I_6$, we aim at deducing a bound similar to the first line of \eqref{eq:230327},
since then the remaining arguments are essentially the
same. 
More precisely, for the cases where $|\beta^0|-k^0+n+1\ge 2$ is valid, 
this condition and the induction hypothesis suffice for the
final conclusion of \eqref{eq:230327}. If $ |\beta^0|-k^0+n+1= 1$ we
can proceed similarly to reach from the first line of
\eqref{eq:230327} to the final line of the estimation.

We can bound $I_2$ similarly as
\begin{align*}
|I_2|
&\le 
C_{8}\lambda^{-1-k^0}\int_0^1t^{|\alpha^0|-1}
|tx|^{-m'(|\alpha^0|+|\eta|-1)}|x| 
\bigl|\partial^{\beta^1}\gamma_{c_1}\bigr|\cdots\bigl|\partial^{\beta^n}\gamma_{c_n}\bigr|\bigl|\mu_{c_{n+1}}\bigr|\,\d t
\\&
\le 
C_{9}\lambda^{-1-k^0}
|x|^{1-m'(|\alpha^0|+n)}\int_0^1
t^{-\sigma'}
\bigl|t^{-1}\partial^{\beta^1}\gamma_{c_1}\bigr|\cdots\bigl|t^{-1}\partial^{\beta^n}\gamma_{c_n}\bigr|\bigl|t^{-1}\mu_{c_{n+1}}\bigr|\,\d t
\\&
\le C_{10}
\lambda^{-1-k^0}|x|^{2-m'(|\alpha^0|+n+1)}
\bigl\|\partial^{\beta^1}\gamma\bigr\|_{p_1}
\cdots
\bigl\|\partial^{\beta^n}\gamma\bigr\|_{p_n}\|\mu\|_{q}, 
\end{align*} 
and from there we proceed as in \eqref{eq:230327} (as explained above).
As for $I_3$, we bound it as 
\begin{align*}
|I_3|
&\le 
C_{11}\lambda^{-1-k^0}\int_0^1t^{|\alpha^0|}
|tx|^{-m'(|\alpha^0|+|\eta|)}|x|^2 
\bigl|\partial^{\beta^1}\gamma_{c_1}\bigr|\cdots\bigl|\partial^{\beta^n}\gamma_{c_n}\bigr|\bigl|\mu_{c_{n+1}}\bigr|\,\d t
\\&
\le 
C_{12}\lambda^{-1-k^0}
|x|^{2-m'(|\alpha^0|+n+1)}\int_0^1
|t|^{-\sigma'}
\bigl|t^{-1}\partial^{\beta^1}\gamma_{c_1}\bigr|\cdots\bigl|t^{-1}\partial^{\beta^n}\gamma_{c_n}\bigr|\bigl|t^{-1}\mu_{c_{n+1}}\bigr|\,\d t
\\&
\le C_{13}
\lambda^{-1-k^0}|x|^{2-m'(|\alpha^0|+n+1)}
\bigl\|\partial^{\beta^1}\gamma\bigr\|_{p_1}
\cdots
\bigl\|\partial^{\beta^n}\gamma\bigr\|_{p_n}\|\mu\|_{q}. 
\end{align*}
From $I_4$ there appears a factor which is either directly bounded by
$\|\partial^{\beta^*}\gamma\|_{p_*}$ 
without the Hardy bound, or alternatively a factor 
 bounded by the $L^q$-norm of $\dot\mu$. Hence typically $I_4$ and $I_5$ are bounded as 
\begin{align*}
|I_4|
&\le 
C_{14}\lambda^{-1-k^0}\int_0^1t^{|\alpha^0|-1}
|tx|^{-m'(|\alpha^0|+|\eta|-1)} 
\bigl|\partial^{\beta^1}\gamma_{c_1}\bigr|\cdots\bigl|\partial^{\beta^i}\dot\gamma_a\bigr|\cdots\bigl|\partial^{\beta^n}\gamma_{c_n}\bigr|\bigl|\mu_{c_{n+1}}\bigr|\,\d t
\\&\le 
C_{15}\lambda^{-1-k^0}|x|^{-m'(|\alpha^0|+n-1)}
\\&\phantom{{}={}}{}
\cdot\int_0^1
t^{-\sigma'} 
\bigl|t^{-1}\partial^{\beta^1}\gamma_{c_1}\bigr|
\cdots\bigl|\partial^{\beta^i}\dot\gamma_a\bigr|\cdots\bigl|t^{-1}\partial^{\beta^n}\gamma_{c_n}\bigr|\bigl|t^{-1}\mu_{c_{n+1}}\bigr|\,\d t
\\&\le 
C_{16}\lambda^{-1-k^0}|x|^{2-m'(|\alpha^0|+n+1)}
\bigl\|\partial^{\beta^1}\gamma\bigr\|_{p_1}
\cdots
\bigl\|\partial^{\beta^n}\gamma\bigr\|_{p_n}\|\mu\|_{q}
, 
\end{align*}
and 
\begin{align*}
|I_5|
&\le 
C_{17}\lambda^{-1-k_0}\int_0^1t^{|\alpha^0|}
|tx|^{-m'(|\alpha^0|+|\eta|)} |x|
\bigl|\partial^{\beta^1}\gamma_{c_1}\bigr|\cdots\bigl|\partial^{\beta^i}\dot\gamma_a\bigr|\cdots\bigl|\partial^{\beta^n}\gamma_{c_n}\bigr|\bigl|\mu_{c_{n+1}}\bigr|\,\d t
\\&\le 
C_{18}\lambda^{-1-k_0}|x|^{1-m'(|\alpha^0|+n)}
\\&\phantom{{}={}}{}
\cdot
\int_0^1
t^{-\sigma'} 
\bigl|t^{-1}\partial^{\beta^1}\gamma_{c_1}\bigr|\cdots\bigl|\partial^{\beta^i}\dot\gamma_a\bigr|\cdots\bigl|t^{-1}\partial^{\beta^n}\gamma_{c_n}\bigr|\bigl|t^{-1}\mu_{c_{n+1}}\bigr|\,\d t
\\&\le 
C_{19}\lambda^{-1-k^0}|x|^{2-m'(|\alpha^0|+n+1)}
\bigl\|\partial^{\beta^1}\gamma\bigr\|_{p_1}
\cdots
\bigl\|\partial^{\beta^n}\gamma\bigr\|_{p_n}\|\mu\|_{q}
.
\end{align*}
Finally we can typically bound $I_6$ as 
\begin{align*}
|I_6|
&\le 
C_{20}\lambda^{-1-k_0}
\\&\phantom{{}={}}{}
\cdot\int_0^1t^{|\alpha^0|}
|tx|^{-m'(|\alpha^0|+|\eta|)}
\bigl|\partial^{\beta^1}\gamma_{c_1}\bigr|\cdots
\bigl|\partial^{\beta^i}\dot\gamma_a\bigr|\cdots\bigl|\partial^{\beta^j}\dot\gamma_b\bigr|
\cdots\bigl|\partial^{\beta^n}\gamma_{c_1}\bigr|\bigl|\mu_{c_{n+1}}\bigr|\,\d t
\\&\le 
C_{21}\lambda^{-1-k_0}|x|^{-m'(|\alpha^0|+n-1)}
\\&\phantom{{}={}}{}
\cdot\int_0^1
t^{-\sigma'} 
\bigl|t^{-1}\partial^{\beta^1}\gamma_{c_1}\bigr|\cdots
\bigl|\partial^{\beta^i}\dot\gamma_a\bigr|\cdots\bigl|\partial^{\beta^j}\dot\gamma_b\bigr|
\cdots\bigl|t^{-1}\partial^{\beta^n}\gamma_{c_n}\bigr|\bigl|t^{-1}\mu_{c_{n+1}}\bigr|\,\d t
\\&\le 
C_{22}\lambda^{-1-k^0}|x|^{2-m'(|\alpha^0|+n+1)}
\bigl\|\partial^{\beta^1}\gamma\bigr\|_{p_1}
\cdots
\bigl\|\partial^{\beta^n}\gamma\bigr\|_{p_n}\|\mu\|_{q}
 .
\end{align*}
Therefore we are done.
\end{proof}

We complete the proof of Theorem~\ref{thm:main result2} as follows.

\begin{proof}[{Completion of the proof of Theorem~\ref{thm:main result2}}]
Fix a closed interval $I\subset \R_+$ and $p\in (1,\infty)$ with $\sigma p<1$, 
and let $R>0$ be large as in Proposition~\ref{prop:main result2bbf}. 
For any $k+|\alpha|\le l$ it suffices {to argue for $|x|> R$ since $s$} vanishes for $|x|\le R$. 
All the estimates below are tacitly understood to be uniform in $\lambda\in I$.

\smallskip
 \noindent
 \textit{Step I.}\ 
 We first show that for any $k+|\alpha|\le l-1$ and $|x|> R$
\begin{equation}
 \label{eq:722b}
 \abs{\partial_\lambda^k\partial_x^\alpha\dot\gamma(1)}\leq C_1\lambda^{-1-k}
 |x|^{2-m(k+\abs{\alpha}+1)+k}.
 \end{equation}
We may let $k+|\alpha|\ge 1$ since the case $k+|\alpha|=0$ 
follows from \eqref{eq:722b} with $k=0$ and $|\alpha|=1$ and integration. 
Now let us use {a representation} 
\begin{align}
(\partial_\lambda^k\partial_x^\alpha\dot\gamma)(1)&
=2\int^1_{1/2}
\biggl((\partial_\lambda^k\partial_x^\alpha\dot\gamma)(t)
+\int^1_t(\partial_\lambda^k\partial_x^\alpha\ddot\gamma)(\tau)
 \, \d \tau
\biggr)
\,\d t
\label{eq:9}
.
\end{align}
This is due to the fundamental theorem of calculus, 
but in fact we have to check integrability of the integrands. 

By Proposition~\ref{prop:main result2bbf} and H\"older's inequality it follows that 
\begin{equation}
 \label{eq:722ba2}
 \norm{\partial_\lambda^k\partial_x^\alpha\dot\gamma}_{L^1(1/2,1)^d}
\leq C_2\lambda^{-1-k} |x|^{2-m(k+|\alpha|+1)+k},
\end{equation}
where the norm on the left-hand side denotes the $L^1$-norm 
of a function on the interval $(1/2,1)$ with values in $\mathbb R^d$. 

On the other hand, as for
$\partial_\lambda^k\partial_x^\alpha\ddot\gamma$, we compute it using
the expression 
\begin{align}
\label{eq:8}
\ddot \gamma_n 
=-
\tfrac12g^{nl}\bigl(\partial_ig_{jl}+\partial_jg_{il}-\partial_lg_{ij}
\bigr)(\dot \gamma_x)_i(\dot \gamma_x)_j
=:-\Gamma^n_{ij}(\dot \gamma_x)_i(\dot \gamma_x)_j
\end{align}
following from \eqref{eq:3b}.
By \eqref{eq:8}, the product rule and the chain rule of
differentiation 
we can write 
\begin{align}
\partial_\lambda^k\partial_x^\alpha\ddot\gamma_n
&=
\sum
C_*
\bigl(\partial_\lambda^{k^0}\partial_x^\eta\Gamma^n_{ij}\bigr)
\Biggl(\prod_{a=1}^d\prod_{b=1}^{\eta_a}\bigl(\partial_\lambda^{k^{ab}}\partial_x^{\alpha^{ab}}\gamma_x\bigr)_{a}\Biggr)
\bigl(\partial_\lambda^{k^1}\partial_x^{\alpha^1}\dot\gamma_x\bigr)_i
\bigl(\partial_\lambda^{k^2}\partial_x^{\alpha^2}\dot\gamma_x\bigr)_j
,
\label{eq:240720}
\end{align}
where the indices $k^0,k^{ab},k^1,k^2\in \mathbb N_0$ and 
$\eta,\alpha^{ab},\alpha^1,\alpha^2\in\mathbb N_0^d$ run over 
\begin{align}
\begin{split}
&k^{ab}+|\alpha^{ab}|\ge 1,\quad 
k^0+\sum_{a=1}^d\sum_{b=1}^{\eta_a}k^{ab}+k^1+k^2=k
,
\\&
|\eta|=0,\dots,k+|\alpha|,\quad 
\sum_{a=1}^d\sum_{b=1}^{\eta_a}\alpha^{ab}+\alpha^1+\alpha^2
=\alpha
.
\end{split}
\label{eq:23032822}
\end{align}
Note we read
$\prod_{b=1}^{\eta_a}\bigl(\partial_\lambda^{k^{ab}}\partial_x^{\alpha^{ab}}\gamma_x\bigr)_{a}=1$
if $\eta_a=0$. 

We recall that $\gamma_x(\tau)=\tau x+\gamma(\tau)$, and since
$\tau\in [1/2,1]$, we do not have to worry about inverse powers of
$\tau$ {when estimating $\partial_\lambda^{k^0}\partial_x^\eta\Gamma^n_{ij}$ by Lemma~\ref{lemma:existence}} 
(as we did in the proof of Proposition \ref{prop:main
 result2bbf}). On the other hand when taking $x$-derivatives of $\gamma_x$,
there are still contributions from differentiating the first term $\tau x$, that
in addition to the derivatives of $\gamma$ need consideration.
The latter are treated by Proposition \ref{prop:main
 result2bbf}. 
We prefer the following uniform treatment,
\begin{subequations}
\begin{align}\label{eq:esmany}
 \begin{split}
 \norm { \parb{\partial_\lambda^{k^{ab}}\partial_x^{\alpha^{ab}}\gamma_x}_a}_{L^p(1/2,1)}&\leq
 C_p
 \lambda^{-k^{ab}}|x|^{-\frac{\rho+1}2(k^{ab}+|\alpha^{ab}|-1)+k^{ab}},\\
\norm
 { \parb{\partial_\lambda^{k^1}\partial_x^{\alpha^1}\dot\gamma_x}_i}_{L^p(1/2,1)}
 &\leq
 C_p
 \lambda^{-k^{1}}|x|^{1-\frac{\rho+1}2(k^{1}+|\alpha^{1}|)+k^{1}},\\
\norm
 { \parb{\partial_\lambda^{k^2}\partial_x^{\alpha^2}\dot\gamma_x}_j}_{L^p(1/2,1)}
&\leq
 C_p \lambda^{-k^{2}}|x|^{1-\frac{\rho+1}2(k^{2}+|\alpha^{2}|)+k^{2}}. 
 \end{split}
\end{align} 
 For any term with $\eta\neq 0$ in {the expansion \eqref{eq:240720}}, say denoted by
$T$, we can show an estimate agreeing with \eqref{eq:722ba2}. In fact 
 by using \eqref{eq:estmin}, \eqref{eq:23032822}, \eqref{eq:esmany}
 with $p=\abs{\eta}+2$, the generalized H\"older
inequality \eqref{eq:Holder} (with this $p$ applied to each of the
$p$ factors), the Hardy inequality \eqref{eq:55} and
Proposition~\ref{prop:main result2bbf} we can estimate
\begin{align*}
\norm {T}_{L^1(1/2,1)}
&\le 
C_3\sum
\lambda^{-1-k^0}|x|^{-m(|\eta|+1)}
\\&\phantom{{}={}}{}
\cdot
\Biggl(\prod_{a=1}^d\prod_{b=1}^{\eta_a}\,\lambda^{-k^{ab}}|x|^{-\frac{\rho+1}2(k^{ab}+|\alpha^{ab}|-1)+k^{ab}}\Biggr)
\\&\phantom{{}={}}{}
\cdot
\lambda^{-k^{1}}|x|^{1-\frac{\rho+1}2(k^{1}+|\alpha^{1}|)+k^{1}}\,
 \lambda^{-k^{2}}|x|^{1-\frac{\rho+1}2(k^{2}+|\alpha^{2}|)+k^{2}}\\
&\le 
C_4\lambda^{-1-k}|x|^{2-m(k+|\alpha|+1)+k}.
\end{align*}
{For any term with $\eta=0$ in the expansion \eqref{eq:240720}
we use the following estimates.}
For $k^{1}+|\alpha^{1}|\geq 1$ the second  bound of
\eqref{eq:esmany} can be
replaced by a sharper bound (similar to the first bound of
\eqref{eq:esmany}),
\begin{equation}\label{eq:except}
 \norm
 { \parb{\partial_\lambda^{k^1}\partial_x^{\alpha^1}\dot\gamma_x}_i}_{L^p(1/2,1)}
 \leq C_p'
 \lambda^{-k^{1}}|x|^{-\frac{\rho+1}2(k^{1}+|\alpha^{1}|-1)+k^{1}},\quad
 k^{1}+|\alpha^{1}|\geq 1,
\end{equation} and similarly for the  third bound for 
$k^{2}+|\alpha^{2}|\geq 1$. 
\end{subequations}
Thanks to these sharper bounds   we  can then argue similarly for the case $\eta= 0$, obtaining an estimate
 for any such term also agreeing with \eqref{eq:722ba2}. This is
straightforward, and we omit the details. Hence we conclude that
\begin{equation}
\|\partial_\lambda^k\partial_x^\alpha\ddot\gamma\|_{L^1(1/2,1)^d}\leq C_5\lambda^{-1-k}|x|^{2-m(k+|\alpha|+1)+k}.
\label{eq:230311}
\end{equation}

 Clearly \eqref{eq:9} and the bounds \eqref{eq:722ba2} and
 \eqref{eq:230311} yield \eqref{eq:722b}. 

\begin{subequations}
\smallskip
 \noindent
 \textit{Step II.}\ 
Now to check the smoothness of $S$ and
 $s$ in $(\lambda,x)$ up to order $l$, it suffices for partial
 derivatives with $\alpha\neq
 0$ to
 prove the continuity of derivatives 
of $\dot\gamma_x(1)$ (or of $\dot\gamma(1)$) up to order $l-1$. Here
we use \eqref{eq:230323} and \eqref{eq:10}. However this is already
implicitly done in our treatment of 
\eqref{eq:9} in Step I, which is based on the $C^{l-1}$-smoothness
of \eqref{eq:24050920} from Proposition \ref{prop:main
 result2bbf}. For derivatives with $\alpha =0$, i.e. containing only $\lambda$-derivatives we differentiate \eqref{maxsol} in $\lambda$, leading (thanks to the stationarity condition 
$(\partial_\kappa \mathcal E)(\lambda,x,\gamma) =0$) to
\begin{equation*}
 2|x|\Phi(\partial_\lambda s)
=(\partial_ \lambda \mathcal E)(\lambda,x,\gamma)
=\lambda^{-2}\int_0^1\chi_RV\abs{\dot \gamma_x}^2 \,\d t,
\end{equation*} from which we deduce the formula
 \begin{equation}\label{eq:derl2l}
 \partial_\lambda s
=2^{-1}|x|^{-2}(1+s)^{-1}\lambda^{-2}\int_0^1\chi_RV\abs{\dot \gamma_x}^2 \,\d t. 
 \end{equation} Inductively we can differentiate
 \eqref{eq:derl2lb} up
 to 
 $l-1$ times in $\lambda$, manifestly yielding continuous expressions for
 $\partial^k_\lambda s$, $k=1, \dots, l$.

\smallskip
 \noindent
 \textit{Step III.}\ 
We are left with \eqref{eq:errestza2}. 
Integrating the expression \eqref{eq:10} using \eqref{eq:230212} and \eqref{eq:722b} for $k+\abs{\alpha}=0$,
we obtain 
\begin{align*}
 S(\lambda,x)^2=(2\lambda)\bigl(x^2+\vO\bigl(\lambda^{-1}|x|^{2-\sigma}\bigr)\bigr),
 \end{align*} 
 and hence \eqref{eq:errestza2} for $k+|\alpha|=0$. 
Next by \eqref{eq:10}, \eqref{eq:230212} and Theorem~\ref{thm:main result2} 
\eqref{item:22120919c} we can write
 \begin{align}
\nabla_x s
=|x|^{-2}(1+s)^{-1}\Bigl[
\dot\gamma(1)
-\lambda^{-1}x\chi_RV
-\lambda^{-1}\chi_RV\dot\gamma(1)
-xs(2+s)\Bigr]
. \label{eq:derl2lb}
 \end{align}
Using this representation and \eqref{eq:722b} we can inductively obtain
\eqref{eq:errestza2} for $k=0$ and all $|\alpha|\le l$. The
straightforward details are omitted. Along the same line any general
mixed derivative can be computed and estimated from
\eqref{eq:derl2lb}. Hence 
we conclude \eqref{eq:errestza2} for all
$k+|\alpha|\le l$, assuming $\alpha\neq 0$.
\end{subequations}

To treat pure $\lambda$-derivatives we use \eqref{eq:derl2l}, leading to 
 \begin{equation*}
 \abs{\partial_\lambda s}\leq C_6\lambda^{-2}\abs{x}^{-2}\int_0^1\abs{tx}^{-\sigma}\,\abs{x}^2\,\d t
=C_7 \lambda^{-2}\abs{x}^{-\sigma},
 \end{equation*} 
 showing 
\eqref{eq:errestza2} for $k=1$ and $\alpha=0$.
We can inductively show \eqref{eq:errestza2} for $|\alpha|=0$ and
all $k\le l$ by repeated $\lambda$-differentiation of 
\eqref{eq:derl2l} and check of the
resulting expressions. This is 
 more complicated than for mixed derivatives, so let us give some
 details of the proof.
 We compute $\partial^k_\lambda \int_0^1\chi_RV\abs{\dot
 \gamma_x}^2 \,\d t$ by differentiating inside the integral using the
product rule. Since $(\chi_RV)(\gamma_x(t))$ is a composition we need
the chain rule to compute $\lambda$-derivatives of this factor, more
or less as we did in Step I. This leads to multiple factors of
$\partial^k_\lambda\gamma_x(t)/t$ for which we have good
$L^p$-bounds. So the main thing is to bound the derivatives
$(\partial_x^\eta(\chi_RV))(\gamma_x(t))$ of the external factor, and
for that we mimic Step VI of the proof of Proposition \ref{prop:main
 result2bbf} estimating as
\begin{equation*}
 \abs{\partial_x^\eta(\chi_RV)}\leq C_8|tx|^{-m(|\eta|)}.
\end{equation*} In parallel to the estimation of $\partial_\lambda s$ this yields the factor
$|x|^{-m(|\eta|)}$ as well as the factor
$t^{-m(|\eta|)}=t^{-\sigma}t^{\sigma-m(|\eta|)}$. Choosing a small
enough $p>1$ the familiar generalized H\"older
inequality yields a bound in terms for the $L^1$-norm of $t^{-\sigma
 p}$ and (after a redistribution of powers of $t^{-1}$ as in Step VI)
products of various $X^{p_j}$-norms for which Proposition \ref{prop:main
 result2bbf} applies. We omit the book-keeping details.

Hence we have proven \eqref{eq:errestza2} for all $k+|\alpha|\le l$.
\end{proof}

In Sections~\ref{sec:Generalized Fourier transform} and \ref{sec:Time-dependent theory} 
we will actually use the following modification of \eqref{eq:errestza2}. 
 
 \begin{corollary}\label{cor:epsSmall}
For any closed interval $I\subset \R_+$ let $s_R\in C^l(I\times\mathbb R^d)$ be given 
as in Theorem~\ref{thm:main result2} for $R\ge R_0$, where $R_0>0$ is
taken sufficiently large.
Then there exists $C>0$ such that 
for any $R\ge R_0$, $\sigma'\in (0,\sigma]$, $k+|\alpha|\le l$ and $(\lambda,x)\in I\times\mathbb R^d$
 \begin{align*}
\bigl|\partial^k_\lambda\partial^\alpha_x s(\lambda,x)\bigr|
\le 
C R^{-\sigma+\sigma'} \lambda^{-1-k}\langle x\rangle^{-m'(k+|\alpha|)+k}
.
\end{align*}
where $m'(k)$ is given in Proposition~\ref{prop:main result2bbf}.
\end{corollary}
\begin{proof} 
The assertion follows from the arguments of the subsection 
with all the estimates involving 
\[|\partial^\alpha\chi_RV|\le C_1 \langle x\rangle^{-m(|\alpha|)}\] 
replaced by 
\begin{align*}
|\partial^\alpha\chi_RV|\le C_2 R^{-\sigma+\sigma'}\langle x\rangle^{-m'(|\alpha|)}
.
\end{align*}
It is straightforward, and we omit repeating the arguments. 
\end{proof}

\section{Strong radiation condition bounds}\label{sec:proof-strong-bound four}

In this section we prove Theorem~\ref{thm:proof-strong-bound}. 
Our main tool is a commutator-type argument.
We adopt a second order differential operator as our conjugate operator, 
whereas the standard Mourre theory employs a first order one. 
This might appear rather peculiar, 
since the resulting operator would be of third order, which usually cannot have a sign.
However, we repeatedly use a certain \textit{increment or decrement identity} to make it of even order
with a definite sign. 

We first introduce notation of the section in Subsection~\ref{subsec:2302052}, 
and then discuss Classical Mechanics interpretations in Subsection~\ref{subsec:200920}. 
Motivated by the classical picture, 
the main propositions of the section will be presented in Subsection~\ref{subsec:22091822}. 
After some preliminaries in Subsection~\ref{subsec:22091823b} these propositions will be proved in 
Subsection~\ref{subsec:22091823}. 
Finally in Subsection~\ref{subsec:22091825} 
we will complete the proof of Theorem~\ref{thm:proof-strong-bound}.

\subsection{Notation}\label{subsec:2302052}

Throughout the section we assume Condition~\ref{cond:220525} with $l= 4$ and $q\equiv 0$.
Let $I\subset \R_+$ be a closed interval, 
and let $S\in C^4(I\times(\mathbb R^d\setminus\{0\}))$ be the function
from
Theorem~\ref{thm:main result2} given for $R\geq R_0>0$. 

According to Theorem~\ref{thm:main result2}, $S(\lambda,\cdot)$ is 
a geodesic distance from the origin.
It is convenient to normalize this function and
 consider $\Phi(\lambda,\cdot)=S(\lambda,\cdot)/\sqrt{2\lambda}$,
 exactly as done in
 Section \ref{sec:constr-solut-eikon}, but even then the singularity
at the origin might cause problems. Consequently we regularize $\Phi(\lambda,\cdot)$
as follows.

\begin{lemma}\label{lem:230114} 
For all $R\geq R_0$
there exists $f\in C^4(I\times\mathbb R^d)$ 
such that (with the dependence on $\lambda\in I$ and $R\geq R_0$
being suppressed): 
\begin{equation}
\text For \,\, any \,\,\lambda\in I\mand |x|>R:\quad f(x)=(2\lambda)^{-1/2}S(\lambda,x).
\label{eq:2210151242}
\end{equation}
Moreover (with all constants below
being independent of $\lambda$ and $R$):
 \begin{subequations}
 \begin{enumerate}
\item
There exist $c,C>0$ such that for any $(\lambda,x)\in I\times\mathbb R^d$
\begin{equation}
c\langle x\rangle \le f(x)\le C\langle x\rangle;
\label{eq:22061920aaa}
\end{equation}
\item
There exists $C'>0$ such that 
for any $|\alpha|\le 4$, $(\lambda,x)\in I\times\mathbb R^d$
\begin{align}
\label{eq:3final2p}
\left| \partial^\alpha f(x)\right|
&\le 
C' f(x)^{1-\min\{|\alpha|,m(|\alpha|)\}}
.
\end{align}
\item
There exists $C''>0$ 
such that for any $(\lambda,x)\in I\times\mathbb R^d$ 
\begin{align}\label{eq:5b}
\bigl(1-C''f^{-\sigma}\bigr)I_d
\le
(\nabla f)\otimes(\nabla f)+ f(\nabla^2f)
\le 
\bigl(1+C''f^{-\sigma}\bigr)I_d
.
\end{align} 
\end{enumerate}
 \end{subequations}
\end{lemma}
\begin{proof}
Such modification is clearly possible due to Theorem~\ref{thm:main result2}. 
We omit the details.
\end{proof}

\begin{remark}\label{rem:230331}
In Section~\ref{sec:Generalized Fourier transform} we will use the
slight modification of \eqref{eq:5b} given with $C''f^{-\sigma}$
replaced by $C'''R^{-\sigma'+\sigma}f^{-\sigma'}$, $\sigma'\in (0,\sigma)$.
\end{remark}

\subsection{Classical Mechanics}\label{subsec:200920}

Here we present a stationary bound holding along 
a classical scattering orbit in the phase space $T^*\mathbb R^d\cong \mathbb R^{2d}$. 
The arguments of this subsection are not necessary for our purpose, 
but they serve as important motivation for our proof of 
Theorem~\ref{thm:proof-strong-bound}.
Moreover the proof here will be directly `lifted' to Quantum Mechanics, 
 apart from the fact that quantum observables of course do not generally commute.

\subsubsection{Free Hamiltonian}
Let us start with the trivial case with $V\equiv 0$ for simplicity.
Hence we consider the free classical Hamiltonian 
\[H^{\mathrm{cl}}_0(x,\xi)=\tfrac12\xi^2\ \ \text{for }(x,\xi)\in \mathbb R^{2d},\]
and the associated Hamilton equations 
\[\dot x=\xi,\quad \dot\xi=0.\]
Then for any initial data $(y,\eta)\in \mathbb R^{2d}$ 
we have an explicit classical orbit
\begin{align*}
x(t)=\eta t+y,\quad
\xi(t)=\eta.
\end{align*}
Assuming the orbit has a positive energy 
\begin{align*}
\lambda=H^{\mathrm{cl}}_0(x(t),\xi(t))=\tfrac12\eta^2>0,
\end{align*}
which is fixed, we discuss the asymptotic relation between observables along the orbit. 
Obviously it follows that (forward in time) 
\begin{align}
 \label{eq:clasLarge}
\xi=\sqrt{2\lambda}|x|^{-1}x+\mathcal O(t^{-1})\ \ \text{as }t\to\infty ,
\end{align}
 and hence the momentum $\xi$ is comparable to $\sqrt{2\lambda}|x|^{-1}x$. 
However we would like to express it in a stationary manner without time parameter. 
For that purpose note 
\begin{align*}
|x|=\sqrt{2\lambda}t+\mathcal O(1)\ \ \text{as }t\to\infty.
\end{align*}
This implies that the quantity $|x|$ is an `effective time' up to a constant factor, 
allowing us to replace the time parameter. 
Now let us introduce 
the `classical gamma observables' as 
\begin{align}\label{eq:gammaClas}
 \begin{split}
 \gamma^{\mathrm{cl}}&
=\xi-(\nabla S_0),
\quad 
\gamma^{\mathrm{cl}}_{\|}=(\nabla S_0)\cdot \gamma^{\mathrm{cl}}
;\quad 
S_0=\sqrt{2\lambda}|x|
,
\end{split}
\end{align} 
cf.\ \eqref{eq:6}. 
Then we obtain a stationary expression of \eqref{eq:clasLarge} as 
\begin{align}
\gamma^{\mathrm{cl}}=\mathcal O(|x|^{-1}).
\label{eq:200907}
\end{align}
Furthermore it follows that 
\begin{align}
\label{eq:quad1b}
 \gamma^{\mathrm{cl}}_{\|}=
H^{\mathrm{cl}}_0-\lambda-\tfrac12(\gamma^{\mathrm{cl}})^2
=-\tfrac12(\gamma^{\mathrm{cl}})^2
=\mathcal O(|x|^{-2}).
\end{align} 
Note that the bound $\gamma^{\mathrm{cl}}_{\|}=\mathcal O(|x|^{-2})$
is sharper than the bound resulting from substituting
\eqref{eq:200907} into the middle expression of \eqref{eq:gammaClas} . 
This is our starting point.

\subsubsection{Perturbed Hamiltonian}
Next we turn to the general case. 
We discuss the perturbed classical Hamiltonian 
\begin{align}
H^{\mathrm{cl}}(x,\xi)
=\tfrac12\xi^2+\chi_R(x)V(x)
\ \ \text{for }
(x,\xi)\in \mathbb R^{2d}, 
\label{eq:20062820}
\end{align}
cf.\ \eqref{eq:2212116}. The associated Hamilton equations are given as 
\begin{align}
\dot x=\xi,\quad
\dot\xi=-\nabla (\chi_RV).
\label{eq:200702}
\end{align}
We are interested in the asymptotic stationary estimates along a
{forward scattering orbit} $(x(t),\xi(t))$ (meaning beyond 
\eqref{eq:200702} that $|x(t)|\to \infty$ for $t\to
\infty$) 
with a fixed positive energy 
\begin{align*}
\lambda=H^{\mathrm{cl}}(x(t),\xi(t))>0.
\end{align*}
Parallel to the free case we can study details of propagation along
this orbit in
terms of 
the classical gamma observables given as 
\begin{align} 
\gamma=\gamma^{\mathrm{cl}}=\xi-(\nabla S)\,\mand \, 
\gamma_\|=\gamma^{\mathrm{cl}}_\|=(\nabla S)\cdot\gamma^{\mathrm{cl}}, 
\label{eq:220918}
\end{align} 
where as in \eqref{eq:6} the function $S$ comes from Theorem~\ref{thm:main result2}. 
Until the end of the subsection we drop the superscript ${}^{\mathrm{cl}}$ for short.

\begin{proposition}\label{prop:200909}
Fix any $\lambda>0$, and define $\gamma$ and $\gamma_\|$ as above.
Let $(x(t),\xi(t))$, $t\in\mathbb R$, be a classical orbit 
for the Hamiltonian \eqref{eq:20062820} with energy $\lambda$ such that 
\[
|x(t)|\to \infty\ \text{ as }t\to+\infty
.
\] 
Then for any $\beta\in (0,2)$ there exists $C>0$ such that for $t\ge
0$ and along the
orbit $(x(t),\xi(t))$
\begin{align*}
\gamma^2\le Cf^{-\beta}
\,\mand \, 
|\gamma_{\|}|\le Cf^{-\beta}.
\end{align*}
\end{proposition}

\begin{remark}
This is the classical counterpart of Theorem~\ref{thm:proof-strong-bound}
 intuitively explaining why the operator $\gamma_\|$ accepts a doubled weight compared to the
 operators $\gamma_i$. 
See also \cite{HS}. 
On the other hand, 
the accepted range of the exponent $\beta$ in
Theorem~\ref{thm:proof-strong-bound} is in general smaller than 
in Proposition~\ref{prop:200909}, 
because higher order derivatives of $S$ are involved in the proof for
the quantum
mechanical case. 
\end{remark}

\begin{proof}
All the quantities below are considered along the {forward scattering
 orbit} $(x(t),\xi(t))$, 
and the dependence on $t\ge 0$ is suppressed. 
It suffices to argue only for large $t\ge 0$,
so that we may consider $x$ to stay away from the origin, in fact
quantitatively as
$|x(t)|>2R$.
Hence by \eqref{eq:2212116} we have the identity for large $t$ 
\begin{align}
\label{eq:quad1}
0=H^{\mathrm{cl}}-\lambda
=\tfrac12\gamma^2+\gamma_\|.
\end{align}
This identity is important (it will be used repeatedly) and deserves to be named the \textit{classical increment or decrement identity}.
By \eqref{eq:quad1} it suffices to show the bound 
\begin{align}
P^{\mathrm{cl}}:=(f^{\beta}\gamma_\|)^2\le C_1 ,
\label{eq:20090718}
\end{align}
and for that we compute the time-derivative of $P^{\mathrm{cl}}$. 
However the time-derivative coincides with the Poisson bracket, hence
$D:=\tfrac{\mathrm d}{\mathrm dt}=\{H^{\mathrm{cl}},\cdot\}$. In any
case we easily compute 
\begin{equation}
DP^{\mathrm{cl}}
=2\beta f^{2\beta-1}\gamma_\|^2(D f)
+2f^{2\beta}\gamma_\|\bigl(D\gamma_\|\bigr)
\label{eq:20090719}, 
\end{equation} motivating us to show that the right-hand side 
eventually is negative along the orbit.

We compute and bound the first term on the right-hand side
by using \eqref{eq:220918}, \eqref{eq:quad1} and \eqref{eq:2212116} as 
\begin{align*}
2\beta f^{2\beta-1}\gamma_\|^2(D f)
&=
2\beta f^{2\beta-1}\gamma_\|^2(\nabla f)\cdot\xi
\\&
=
2\beta (2\lambda)^{-1/2}
f^{2\beta-1}\gamma_\|^2
\bigl(\gamma_\|+|\nabla S|^2\bigr)
\\&
=
2\beta (2\lambda)^{-1/2}
f^{2\beta-1}\gamma_\|^2
\bigl(-\tfrac12\gamma^2+2\lambda-2\chi_RV\bigr)
\\&
\le 
-\beta(2\lambda)^{-1/2}f^{2\beta-1}\gamma_\|^2\gamma^2
+2\beta (2\lambda)^{1/2}f^{2\beta-1}\gamma_\|^2
+C_2f^{2\beta-1-\sigma}\gamma_\|^2
.
\end{align*}
To compute the second term of \eqref{eq:20090719} we note the identity obtained 
by differentiating the equation \eqref{eq:2212116}
\begin{align}
(\nabla^2S)\nabla S=\tfrac12\nabla |\nabla S|^2
=-\nabla (\chi_RV).
\label{eq:220919537}
\end{align}
Then by \eqref{eq:220918}, \eqref{eq:220919537} and \eqref{eq:quad1} 
this second term 
 is computed as 
\begin{align*}
2f^{2\beta}\gamma_\|\bigl(D\gamma_\|\bigr)
&
=
2f^{2\beta}\gamma_\|\bigl(\xi \cdot(\nabla^2 S)\gamma
+(\nabla S)\cdot \bigl(-(\nabla \chi_RV)-(\nabla^2S)\xi\bigr)\bigr)
\\&
=
2f^{2\beta}\gamma_\|\bigl(\gamma \cdot(\nabla^2 S)\gamma
-(\nabla S)\cdot (\nabla^2S)(\nabla S)
-(\nabla S)\cdot (\nabla \chi_RV)\bigr)
\\&
= 
-(2\lambda)^{1/2}f^{2\beta}\gamma^2\gamma \cdot(\nabla^2 f)\gamma
.
\end{align*}
Combining the above computations and using \eqref{eq:5b} and \eqref{eq:quad1}, 
we bound \eqref{eq:20090719} as 
\begin{align*}
DP^{\mathrm{cl}}
&\le 
-(2\lambda)^{1/2}f^{2\beta-1}\gamma^2\gamma\cdot \bigl(\beta(\nabla f)\otimes(\nabla f)+(f\nabla^2 f)\bigr)\gamma
\\&\phantom{{}={}}{}
+2\beta (2\lambda)^{1/2}f^{2\beta-1}\gamma_\|^2
+C_2f^{2\beta-1-\sigma}\gamma_\|^2
\\&
\le 
-(\min\{1,\beta\})(2\lambda)^{1/2}f^{2\beta-1}\gamma^4\bigl(1-C_3f^{-\sigma}\bigr)
\\&\phantom{{}={}}{}
+2\beta(2\lambda)^{1/2}f^{2\beta-1}\gamma_\|^2
+C_2f^{2\beta-1-\sigma}\gamma_\|^2
\\&
\le 
-(\min\{4-2\beta,2\beta\})(2\lambda)^{1/2}f^{2\beta-1}\gamma_\|^2
+C_4f^{2\beta-1-\sigma}\gamma_\|^2
,
\end{align*}
so that for any sufficiently large $t$ 
\begin{align}
DP^{\mathrm{cl}}
&\le 
-c_1f^{-1}P^{\mathrm{cl}}\leq 0.
\label{eq:20090913}
\end{align}
Hence indeed $P^{\mathrm{cl}}$ is bounded as $t\to\infty$, 
verifying \eqref{eq:20090718}, and we are done.
\end{proof}
\begin{remark}
We have presented a stationary proof without explicit time parameter,
so that the scheme extends to the quantum setup, see Proposition~\ref{prop:220602}
and its proof. 
We could have given a simpler proof computing $D(f^\beta\gamma_\|)$,
however we are not aware of any analogous procedure in Quantum
Mechanics. Note that in the literature on scattering theory of
Schr\"odinger operators conjugate
operators are usually of first order, and consequently our scheme
of proof (not being of this sort) 
is rather non-conventional.
\end{remark}

\subsection{Main propositions of the section}\label{subsec:22091822}

Clearly we can assume that the arbitrary compact subset
$I'\subseteq I$ in Theorem~\ref{thm:proof-strong-bound} is
taken as $I'= I$, and hence that also $ I$ is compact. We make these
assumption in the remaining part of the section.

To prove Theorem~\ref{thm:proof-strong-bound} 
we are going to compute and bound a quantum observable corresponding
to the expression 
$D((f^\beta \gamma^{\mathrm{cl}}_\|)^2)$ 
appearing in the proof of Proposition~\ref{prop:200909}. 
As a quantum observable corresponding to \eqref{eq:20090718}, 
we consider 
\begin{align}
P=\gamma_\|\theta^{2\beta}\gamma_\|,\quad \beta\in (0,2).
\label{eq:220920}
\end{align}
Here the weight $\theta$ is defined as 
\begin{align*}
\theta=\theta_{\delta, \nu}(f)=\int_0^{f}(1+t/\nu)^{-1-\delta}\,\mathrm dt
;\quad \nu\ge 1,\ \ {\delta>0}.
\end{align*}
It is a refinement of the so-called \textit{Yosida approximation} of $f$.
Note that $\theta$ is bounded for each $\nu\ge 1$, but that 
\[\theta'\uparrow 1\text{ \ and\ \ }\theta\uparrow f\ \ \text{pointwise as }\nu\to\infty, \]
where $\theta'$ denotes the derivative of $\theta$ as a function of $f$.

The main propositions of the section present upper and lower bounds of a `distorted commutator'
\begin{align*}
2\mathop{\mathrm{Im}}(P(H-z))
\end{align*}
with 
\begin{align*}
z=\lambda\pm\mathrm i\Gamma\in I_\pm
:=\bigl\{z=\lambda\pm\mathrm i\Gamma\in\mathbb C\,|\,\lambda\in I,\ \Gamma\in(0,1)\bigr\}
.
\end{align*}
Note that it is comparable with 
$
DP^{\mathrm{cl}}=\{H^{\mathrm{cl}},P^{\mathrm{cl}}\}
$ considered before in Classical Mechanics.
In fact, if $z=\lambda\in I$, it is nothing but the commutator $\mathrm i[H,P]$.

\begin{proposition}\label{prop:220602}
Define $P$ and $\theta$ as above 
with $\beta\in (0,2)$ and {$\delta>0$} arbitrarily fixed. 
Then there 
exist $c>0$ and $\nu_0\ge 1$ such that for all 
$z=\lambda\pm \mathrm i\Gamma\in I_\pm$, $R\geq R_0$ 
and $\nu\ge \nu_0$ (with the constant $C>0$ being independent of
$\lambda\in I$ but possibly depending on $R\geq R_0$), 
\begin{align*}
 \pm2\mathop{\mathrm{Im}}(P(H-z))
&\leq -c\gamma_\|\theta'\theta^{2\beta-1}\gamma_\|
+C\Gamma f^{-2}\theta^{2\beta}
 \\&\phantom{{}={}}{}
+Cf^{-1-2\beta_c+3\delta}\theta^{2\beta}
+C(H-z)^*f^{1+\delta}\theta^{2\beta-\delta}(H-z)
\end{align*}
as quadratic forms on $\vD(H)$. 
\end{proposition}

\begin{proposition}\label{prop:22060215}
Suppose the same setting of Proposition~\ref{prop:220602}, and let 
 $\epsilon\in (0,1]$. Then there exists $\nu_0\ge 1$ that for all
$z=\lambda\pm \mathrm i\Gamma\in I_\pm$, $R\geq R_0$ 
and $\nu\ge \nu_0$ (with the constant $C>0$ being independent of
$\lambda\in I$ but possibly depending on $R\geq R_0$), 
\begin{align*}
\begin{split}
 \pm2\mathop{\mathrm{Im}}(P(H-z))
&\geq 
-\epsilon \gamma_\|\theta'\theta^{2\beta-1}\gamma_\|
-C \Gamma f^{-2}\theta^{2\beta}
 \\&\phantom{{}={}}{}
-Cf^{-1-2\beta_c+3\delta}\theta^{2\beta}
-C(H-z)^*f^{1+\delta}\theta^{2\beta-\delta}(H-z)
\end{split}
\end{align*}
as quadratic forms on $\vD(H)$.
\end{proposition}

\begin{remark}
Proposition~\ref{prop:220602} is obviously a 
quantum analogue of \eqref{eq:20090913} 
with some negligible errors coming from commutation of 
observables,
while Proposition~\ref{prop:22060215} only says that 
the left-hand side of Proposition~\ref{prop:220602} is negligible too. 
In the proof of Theorem~\ref{thm:proof-strong-bound} 
we shall take expectation of these bounds 
in the state $\phi=R(z)\psi$, and take the limits $\Gamma\to 0_+$
and $\nu\to\infty$. 
Then the second, third and fourth terms on the right-hand sides 
of Propositions~\ref{prop:220602} and \ref{prop:22060215}
are in fact negligible for $\beta<\beta_\c$ and
 $\delta$ taken small. This is 
due to the factor $\Gamma$, 
the limiting absorption principle bounds 
and cancellation of $H-z$ and $R(z)$, respectively. 
We will give the details in Subsection~\ref{subsec:22091825}. 
\end{remark}

The propositions will be proved 
in Subsection~\ref{subsec:22091823},
after some preliminaries in Subsection~\ref{subsec:22091823b}.

\subsection{Preliminaries}\label{subsec:22091823b}

Here we gather some identities and estimates 
that will be frequently cited in the later subsections. 
Throughout the subsection we adapt the setting of Proposition~\ref{prop:220602}.
In particular, 
$\beta\in (0,2)$ and {$\delta>0$} are fixed
along with $I$, $R$, $S$ and $f$ 
from Subsection~\ref{subsec:2302052}.

We first record several bounds for the weight $\theta$.
We denote the derivatives of $\theta$ in $f$ by primes or similar superscripts such as
$\theta',\theta'',\dots,\theta^{(k)}$. 

\begin{lemma}\label{lem:basic3} 
There exist $c_0,C_0>0$ such that uniformly in $\lambda\in I$ and $\nu\ge 1$
\begin{align*} 
&c_0 \min\{\nu,f\}\le \theta\le \min\{C_0\nu,f\},
\quad 
\theta'\le f^{-1}\theta.
\end{align*}
Furthermore, for any $k\in\mathbb N$
there exist $c_k,C_k>0$ such that uniformly in $\lambda\in I$ and $\nu\ge 1$
\begin{align*}
&c_k\nu^{1-k}f^{-k-\delta}\theta^{k+\delta}\le (-1)^{k-1}\theta^{(k)}
\le C_k\nu^{1-k}f^{-k-\delta}\theta^{k+\delta}.
\end{align*}
\end{lemma}
\begin{remark}
In the rest of the section we will mostly use the simplified bounds 
\[
c_1f^{-1-\delta}\theta^{1+\delta}
\le \theta'\le f^{-1}\theta,\quad 
|\theta^{(k)}|\le C_k'f^{-k}\theta
.\]
\end{remark}
\begin{proof}
According to whether $f\le \nu$ or $f\ge \nu$, we have 
\[
\theta\ge \int_0^f2^{-1-\delta}\,\mathrm ds=2^{-1-\delta}f
\quad \text{or}\quad 
\theta\ge \int_0^\nu 2^{-1-\delta}\,\mathrm ds=2^{-1-\delta}\nu,
\]
respectively. On the other hand
\begin{align*}
\theta
=
\delta^{-1}\nu\bigl(1-(1+f/\nu)^{-\delta}\bigr)\le \delta^{-1}\nu\,\mand\, 
\theta
\le 
\int_0^{f}\mathrm ds=f.
\end{align*}
Hence we obtain the asserted bounds for $\theta$. 

We also have 
\[\theta'=(1+f/\nu)^{-1-\delta}= f^{-1}\int_0^f(1+f/\nu)^{-1-\delta}\,\mathrm ds\le f^{-1}\theta.\]
In addition, for any $k\in\mathbb N$ we can find a constant $C_k>0$ such that 
\begin{align*}
\theta^{(k)}
=
(-1)^{k-1}C_k\nu^{1-k}(1+f/\nu)^{-k-\delta}
=
(-1)^{k-1}C_k\nu^{1-k}f^{-k-\delta}(f^{-1}+\nu^{-1})^{-k-\delta}.
\end{align*}
Here the last factor satisfies 
\[
\tfrac12\min\{\nu,f\} \le (f^{-1}+\nu^{-1})^{-1}\le \min\{\nu,f\}.
\]
Hence we are done. 
\end{proof}

We next present a simple key identity corresponding to \eqref{eq:quad1}. 
It involves gamma observables of different orders, 
and hence it can be used either to increment or decrement the order
of a differential operator (it will be used both ways many times).

\begin{lemma}\label{lem:22101720}
For any $z=\lambda\pm\mathrm i\Gamma\in I_\pm$ on the subset $\{|x|>2R\}$ the identities 
\begin{equation}\label{eq:fundd}
 H-z=\tfrac12\gamma^2\pm\gamma_\|\mp\mathrm i\Gamma
\end{equation}
hold, respectively.
\end{lemma}
\begin{proof}
The assertion is straightforward by \eqref{eq:2212116} and \eqref{eq:6}. 
(Recall that throughout the section, $q\equiv 0$ in \eqref{eq:23011416}.) 
\end{proof}

The following commutator relations of the gamma observables are also important. 

\begin{lemma}
\begin{subequations} 
For any $\lambda\in I$ and $i,j=1,\dots,d$ 
one has on $\{|x|\ge 2\}$
\begin{align}
[\gamma_i,\gamma_j]&=0,
\label{eq:22091918}
\\
[\gamma_\|,\gamma_i]
&=\mathrm i\mathop{\mathrm{Re}}((\nabla \partial_iS)\cdot \gamma),
\label{eq:22091918b}
\\
[\gamma_\|,\gamma^2]
&=
2\mathrm i\gamma\cdot(\nabla^2 S)\gamma
-\tfrac{\mathrm i}2 (\Delta^2 S)
.
\label{eq:22091919}
\end{align}
\end{subequations} 
\end{lemma}
\begin{proof}
By conjugation by $\mathrm e^{\mp\mathrm iS}$ the relation \eqref{eq:22091918} reduces to that for 
$p$,
which is trivial. 
We can verify \eqref{eq:22091918b} immediately by \eqref{eq:6} and \eqref{eq:22091918}. 
As for \eqref{eq:22091919} we use \eqref{eq:22091918b} to compute 
\begin{align*}
[\gamma_\|,\gamma^2]
&=
\bigl(\mathrm i\gamma\cdot(\nabla \partial_iS)-\tfrac12(\partial_i\Delta S)\bigr)\gamma_i
+
\gamma_i\bigl(\mathrm i(\nabla \partial_iS)\cdot \gamma+\tfrac12(\partial_i\Delta S)\bigr)
\\&=
2\mathrm i\gamma\cdot(\nabla^2 S)\gamma
-\tfrac{\mathrm i}2(\Delta^2 S)
.
\end{align*}
Thus we obtain the assertion.
\end{proof}

Lastly in this subsection we present several handy 'ellipticity estimates'. 
We will often use the following identities holding 
for any $a\in C^\infty(\mathbb R^d;\mathbb R)$: 
\begin{subequations}
\begin{align}
 p\cdot ap&= \mathop{\mathrm{Re}}(ap^2)+\tfrac12(\Delta a),
\label{eq:22103116}
\\ 
 p\cdot(p\cdot ap)p
&=p^2ap^2-p\cdot (\nabla^2a)p+p\cdot (\Delta a)p.
\label{eq:22103116b}
\end{align}
\end{subequations}
The lemma below implies that 
any compactly supported differential operator of order at most four is 
bounded by the last two terms of Propositions~\ref{prop:220602} and \ref{prop:22060215}. 
For short we shall denote their sum as 
\begin{align}
Q=f^{-1-2\beta_c+3\delta}\theta^{2\beta}+(H-z)^*f^{1+\delta}\theta^{2\beta-\delta}(H-z).
\label{eq:2211021946}
\end{align}

\begin{lemma}\label{lem:221031}
For any $\eta\in C^\infty_{\mathrm c}(\mathbb R^d;\mathbb R)$ 
there exists $C>0$ such that 
uniformly in $z=\lambda\pm\mathrm i\Gamma\in I_\pm$,
 $R\geq R_0$ and $\nu\ge 1$ 
\begin{align}
\eta
&\le 
CQ
,
\quad 
\gamma\cdot\eta \gamma
\le 
CQ
,\quad 
\gamma\cdot (\gamma\cdot\eta \gamma)\gamma
\le
CQ
.
\label{eq:22103110}
\end{align}
\end{lemma}
\begin{proof}
The first bound from \eqref{eq:22103110} is
trivial. For the second and
 third bounds it suffices to show the bounds with
 $\gamma$ replaced by $p$. 
As for the second, we use \eqref{eq:22103116} and
 the Cauchy--Schwarz inequality to bound it for any $N\ge 0$ as 
\begin{align*}
p\cdot\eta p
&=
2\mathop{\mathrm{Re}}(\eta (H-z))
-2\eta (V-z)
+\tfrac12(\Delta \eta)
\\&
\le 
C_1f^{-N}
+C_1(H-z)^*f(H-z)
.
\end{align*}
The second bound follows by choosing $N=5$. 
Lastly we rewrite the left-hand side of the third bound of 
 \eqref{eq:22103110} by using \eqref{eq:22103116b} as 
\begin{align*}
p\cdot(p\cdot \eta p)p
&=
4(H-z)^*\eta(H-z)
-8\mathop{\mathrm{Re}}\bigl(\eta(V-z^*)(H-z)\bigr)
\\&\phantom{{}={}}{}
+4\eta|V-z|^2
-p\cdot (\nabla^2\eta)p+p\cdot (\Delta \eta)p.
\end{align*}
The third bound then follows by the Cauchy--Schwarz inequality
and the first and second bounds. 
\end{proof}

The next lemma implies various forms of negligible terms 
can be absorbed into the leading term 
$\gamma_\|\theta'\theta^{2\beta-1}\gamma_\|$, 
or $\gamma\cdot(\gamma\cdot \theta'\theta^{2\beta-1}\gamma)\gamma$,
cf.\ Lemma~\ref{lem:2210312250}.

\begin{lemma}\label{lem:algebra} 
 For any $\epsilon>0$ 
and $R\geq R_0$ there exist $C=C(\delta)>0$ and $\nu_0=\nu_0(\delta)\ge 1$, such that 
 for all $z=\lambda\pm\mathrm i\Gamma\in I_\pm$ and $\nu\ge \nu_0$ 
\begin{subequations}
 \begin{align}
\gamma_\| f^{-1-\delta}\theta^{2\beta}\gamma_\|
&\le 
\epsilon \gamma_\|\theta'\theta^{2\beta-1}\gamma_\|
+CQ
,
 \label{eq:item:1a}
\\
 \gamma \cdot f^{-1-\beta_c+\delta}\theta^{2\beta}\gamma
 &\leq 
\epsilon \gamma_\|\theta'\theta^{2\beta-1}\gamma_\|
+CQ
.
 \label{eq:item:1}
\\
\gamma\cdot(\gamma\cdot f^{-1-\delta}\theta^{2\beta}\gamma)\gamma
&\le 
\epsilon \gamma\cdot(\gamma\cdot \theta'\theta^{2\beta-1}\gamma)\gamma
+CQ
,
 \label{eq:item:1b}
 \end{align}
\end{subequations}
\end{lemma}
\begin{proof} 
Let $\epsilon>0$.
By Lemma~\ref{lem:basic3}
we can find $\eta\in C^\infty_{\mathrm c}(\mathbb R^d;\mathbb R)$, such
that uniformly in $\nu\ge \nu_0$ for some $\nu_0\ge 1$ chosen sufficiently large
\begin{align*}
f^{-1-\delta}\theta^{2\beta}
\le 
\epsilon \theta'\theta^{2\beta-1}+\eta
.
\end{align*}
Then we conclude \eqref{eq:item:1a} and
 \eqref{eq:item:1b} by Lemma~\ref{lem:221031}. 
To prove \eqref{eq:item:1}
note by conjugation by $\e^{\pm\i \chi_1S}$ we can rewrite \eqref{eq:22103116} as 
 \begin{align}
 \gamma\cdot a\gamma= \mathop{\mathrm{Re}}(a\gamma^2)+\tfrac12(\Delta a).
\label{221029118}
 \end{align}
In addition, set 
\begin{align*}
t=1+\beta_c-\delta.
\end{align*}
Then by \eqref{221029118}, \eqref{eq:14.1.7.23.24b}, Lemmas~\ref{lem:22101720} and \ref{lem:basic3}, \cs and Lemma~\ref{lem:221031} 
we can estimate (using the cut-off function $\chi_{2R}$ from \eqref{eq:14.1.7.23.24b}) 
\begin{align*}
\gamma\cdot f^{-t}\theta^{2\beta}\gamma&
\leq 
\mp2\mathop{\mathrm{Re}}(\chi_{2R} f^{-t}\theta^{2\beta} \gamma_\|)
+2\mathop{\mathrm{Re}}(\chi_{2R} f^{-t}\theta^{2\beta} (H-z))
\\&\phantom{{}={}}{}
+C_1f^{-t-2}\theta^{2\beta}
+\mathop{\mathrm{Re}}\bigl((1-\chi_{2R})f^{-t}\theta^{2\beta}\gamma^2\bigr)
\\&
\le 
\gamma_\| f^{-1-\delta}\theta^{2\beta}\gamma_\|
+{C_1}f^{-t-2}\theta^{2\beta}
+C_2 f^{-2t+1+\delta}\theta^{2\beta}
\\&\phantom{{}={}}{}
+C_2(H-z)^*f^{-1-\delta}\theta^{2\beta}(H-z)
+C_2Q
.
\end{align*} 
Applying in turn \eqref{eq:item:1a}, this yields \eqref{eq:item:1}. 
\end{proof}

The final lemma says $4\gamma_\|\theta'\theta^{2\beta-1}\gamma_\|$ and 
$\gamma\cdot \bigl(\gamma\cdot\theta'\theta^{2\beta-1}\gamma\bigr)\gamma$ 
are interchangeable up to small errors. 

\begin{lemma}\label{lem:2210312250}
For any $\epsilon>0$ and $R\geq R_0$ there exist $C>0$ and $\nu_0\ge 1$ such that 
for all $z=\lambda+\mathrm i\Gamma\in I_+$ and $\nu\ge \nu_0$
\begin{align*}
&\pm
\bigl\{
4\gamma_\|\theta'\theta^{2\beta-1}\gamma_\|
-\gamma\cdot \bigl(\gamma\cdot\theta'\theta^{2\beta-1}\gamma\bigr)\gamma
\bigr\}
\le 
\epsilon\gamma_\|\theta'\theta^{2\beta-1}\gamma_\|
+C\Gamma f^{-2}\theta^{2\beta}
+CQ
.
\end{align*}
The same bounds also hold uniformly in $z=\lambda-\mathrm i\Gamma\in I_-$ and $\nu\ge \nu_0$.
\end{lemma}
\begin{proof}
We discuss only $z=\lambda+\mathrm i\Gamma\in I_+$ 
since the other is proved in the same manner. 
Similarly to \eqref{221029118}, by conjugation by $\e^{-\i \chi_1S}$ 
we have \eqref{eq:22103116b} rewritten as 
 \begin{align}
 \gamma\cdot(\gamma\cdot a\gamma)\gamma=\gamma^2a\gamma^2-\gamma\cdot (\nabla^2a)\gamma
+\gamma\cdot (\Delta a)\gamma.
\label{221029119}
 \end{align}
Then it follows by Lemma~\ref{lem:basic3} that 
\begin{align}
\pm\gamma\cdot \bigl(\gamma\cdot\theta'\theta^{2\beta-1}\gamma\bigr)\gamma
&\le 
\pm\gamma^2\theta'\theta^{2\beta-1}\gamma^2
+C_1\gamma\cdot f^{-3}\theta^{2\beta}\gamma.
\label{eq:22103121}
\end{align}
The second term on the right-hand side of \eqref{eq:22103121} 
can be bounded by \eqref{eq:item:1},
and hence it suffices to discuss the first term of \eqref{eq:22103121}. 
With the localization factor $\chi_{2R}$ inserted in
 this term we are allowed to isolate $\gamma^2$ in
 \eqref{eq:fundd} and substitute into the
 two appearing factors of $\gamma^2$. After 
expansion 
we then use the Cauchy--Schwarz inequality to absorb cross terms into the diagonal ones.
Hence it follows that 
\begin{align*}
\pm\gamma^2\theta'\theta^{2\beta-1}\gamma^2
&=
\pm4\bigl(-\gamma_\|+\mathrm i\Gamma +(H-z)\bigr)^*
\chi_{2R}\theta'\theta^{2\beta-1}
\bigl(-\gamma_\|+\mathrm i\Gamma +(H-z)\bigr)
\\&\phantom{{}={}}{}
\pm\gamma^2(1-\chi_{2R})\theta'\theta^{2\beta-1}\gamma^2
\\&
\le 
(\pm4+\epsilon)\gamma_\|\chi_{2R}\theta'\theta^{2\beta-1}\gamma_\|
+C_2\Gamma^2f^{-1}\theta^{2\beta}
+C_2Q
.
\end{align*}
We can remove $\chi_{2R}$ from the first term on the right-hand side above 
by retaking $C_2>0$ larger,
and hence it suffices to discuss the second term. 
By the expression \eqref{eq:6}, the Cauchy--Schwarz inequality 
and Lemma~\ref{lem:basic3} we can proceed as 
\begin{align}
\begin{split}
\Gamma^2 f^{-1}\theta^{2\beta}
&=
\tfrac12\Gamma\mathop{\mathrm{Re}}\bigl((\nabla f^{-1}\theta^{2\beta})\cdot \gamma\bigr)
+\tfrac12\Gamma(\nabla f^{-1}\theta^{2\beta})\cdot (\nabla S)
\\&\phantom{{}={}}{}
-\Gamma\mathop{\mathrm{Im}}\bigl(f^{-1}\theta^{2\beta}(H-z)\bigr)
\\&\le 
C_3\gamma_\|f^{-2}\theta^{2\beta}\gamma_\|
+C_3\Gamma f^{-2}\theta^{2\beta}
+C_3Q.
\end{split}
\label{eq:221030}
\end{align}
The first term on the right-hand side of \eqref{eq:221030} 
can be bounded as asserted by using \eqref{eq:item:1a}. 
Hence we obtain the assertion.
\end{proof}

\subsection{Upper and lower bounds for the distorted commutator}\label{subsec:22091823}

Here we prove Propositions~\ref{prop:220602} and \ref{prop:22060215}. 
We start with Proposition~\ref{prop:220602}.

\begin{proof}[Proof of Proposition~\ref{prop:220602}]
\textit{Step I.}
Let us prove the assertion only for the upper sign, since the lower one follows 
in the same manner. 
With reference to the constants $0<c\leq C$ of \eqref{eq:22061920aaa} and
the function $\chi$ of \eqref{eq:14.1.7.23.24} we 
introduce the smooth cut-off function
 $\widetilde\chi=\chi(f/4CR)$. Then 
\begin{align*}
\mathop{\mathrm{supp}}\widetilde\chi\subseteq \{|x|>2R\},\quad 
\widetilde\chi=1\ \ \text{on }\{|x|\ge 8RC/c\}.
\end{align*}
 In particular \eqref{eq:fundd} applies on the support of
 $\widetilde\chi$. Hence by \eqref{eq:220920} and Lemma~\ref{lem:22101720} 
we can split the distorted commutator on the left-hand side of the assertion as 
\begin{align}
\begin{split}
2\mathop{\mathrm{Im}}(P(H-z))
&=
\mathop{\mathrm{Im}}\bigl(\gamma_\|\widetilde\chi \theta^{2\beta}\gamma_\|\gamma^2\bigr)
+2\mathop{\mathrm{Im}}\bigl(\gamma_\|\widetilde\chi \theta^{2\beta}\gamma_\|^2\bigr)
\\&\phantom{{}={}}{}
+2\mathop{\mathrm{Im}}\bigl(\gamma_\|(1-\widetilde\chi) \theta^{2\beta}\gamma_\|(H-\lambda)\bigr)
-2\Gamma \gamma_\|\theta^{2\beta}\gamma_\|.
\end{split}
\label{eq:22092018}
\end{align}
In the following steps 
we will further compute and bound each term of \eqref{eq:22092018}.

First we comment on our notation. 
Throughout the proof we fix $\epsilon>0$ such that 
 for some $c_1>0$, it follows that uniformly in $\lambda\in I$
\begin{align}
(2\lambda)^{1/2}\min\{4-2\beta,2\beta\}-19\epsilon>c_1>0.
\label{eq:221029}
\end{align}
We shall consider only $\nu\ge \nu_0$ with 
appropriate $\nu_0\ge 1$ tacitly retaken 
each time we apply Lemmas~\ref{lem:algebra} or \ref{lem:2210312250}. 
In addition, we shall adopt the notation $Q$ from \eqref{eq:2211021946}. 
We particularly note that, when computing \eqref{eq:22092018}, 
once a derivative hits $\widetilde\chi$, 
the corresponding term is immediately bounded by $C_1Q$ for some $C_1>0$ 
due to Lemma~\ref{lem:221031}. 
We shall also tacitly implement such estimates.

\smallskip
\noindent
\textit{Step II.}
Now we start with the first term on the right-hand side of \eqref{eq:22092018}. 
By \eqref{eq:6} and \eqref{eq:22091919} we can compute it as 
\begin{align*}
\mathop{\mathrm{Im}}\bigl(\gamma_\|\widetilde\chi \theta^{2\beta}\gamma_\|\gamma^2\bigr)
&=
\mathop{\mathrm{Im}}\bigl(\gamma_\|[\widetilde\chi \theta^{2\beta},\gamma]\cdot\gamma\gamma_\|\bigr)
+\mathop{\mathrm{Im}}\bigl(\gamma_\|\widetilde\chi \theta^{2\beta}[\gamma_\|,\gamma^2]\bigr)
\\&
=
(2\lambda)^{-1/2}\mathop{\mathrm{Re}}\bigl(\gamma_\|(\widetilde\chi \theta^{2\beta})'\gamma_\|^2\bigr)
+2\mathop{\mathrm{Re}}\bigl(\gamma_\|\widetilde\chi \theta^{2\beta}\gamma\cdot(\nabla^2 S)\gamma\bigr)
\\&\phantom{{}={}}{}
-\tfrac12\mathop{\mathrm{Re}}\bigl(\gamma_\|\widetilde\chi \theta^{2\beta}(\Delta^2 S)\bigr)
.
\end{align*}
\begin{subequations}
By the Cauchy--Schwarz inequality, \eqref{eq:220919537},
\eqref{eq:3final2p}, and Lemmas \ref{lem:221031} and ~\ref{lem:algebra},
we can bound it as 
\begin{align}
\begin{split}
\mathop{\mathrm{Im}}\bigl(\gamma_\|\widetilde\chi \theta^{2\beta}\gamma_\|\gamma^2\bigr)
&
\le 
2\beta(2\lambda)^{-1/2}\mathop{\mathrm{Re}}\bigl(\gamma_\|^2\widetilde\chi \theta'\theta^{2\beta-1}\gamma_\|\bigr)
\\&\phantom{{}={}}{}
+2\mathop{\mathrm{Re}}\bigl(\gamma_\|\gamma\cdot\widetilde\chi \theta^{2\beta}(\nabla^2 S)\gamma\bigr)
+\epsilon\gamma_\|\theta'\theta^{2\beta-1}\gamma_\|
+C_2Q
.
\end{split}
\label{eq:22092020}
\end{align}
The second term of \eqref{eq:22092018} can be computed by \eqref{eq:6}, \eqref{eq:2212116} and Lemma~\ref{lem:algebra} as 
\begin{align}
\begin{split}
2\mathop{\mathrm{Im}}\bigl(\gamma_\|\widetilde\chi \theta^{2\beta}\gamma_\|^2\bigr)
&=
\gamma_\|(\nabla S)\cdot\bigl(\nabla\widetilde\chi \theta^{2\beta}\bigr)\gamma_\|
\\&
\le 
\bigl(2\beta(2\lambda)^{1/2}+\epsilon\bigr)\gamma_\|\theta'\theta^{2\beta-1}\gamma_\|
+C_3Q
.
\end{split}
\label{eq:221016185}
\end{align}
The third and fourth terms of \eqref{eq:22092018} is bounded trivially as 
\begin{align}
2\mathop{\mathrm{Im}}\bigl(\gamma_\|(1-\widetilde\chi) \theta^{2\beta}\gamma_\|(H-\lambda)\bigr)
-2\Gamma \gamma_\|\widetilde\chi \theta^{2\beta}\gamma_\|
\le C_4Q.
\label{eq:22101620}
\end{align} 
\end{subequations}

Hence by \eqref{eq:22092018}, and \eqref{eq:22092020}--\eqref{eq:22101620} we obtain 
\begin{align}
\begin{split}
2\mathop{\mathrm{Im}}(P(H-z))
&\le 
2\beta(2\lambda)^{-1/2}\mathop{\mathrm{Re}}\bigl(\gamma_\|^2\widetilde\chi \theta'\theta^{2\beta-1}\gamma_\|\bigr)
+2\mathop{\mathrm{Re}}\bigl(\gamma_\|\gamma\cdot\widetilde\chi \theta^{2\beta}(\nabla^2 S)\gamma\bigr)
\\&\phantom{{}={}}{}
+\bigl(2\beta(2\lambda)^{1/2}+2\epsilon\bigr)\gamma_\|\theta'\theta^{2\beta-1}\gamma_\|
+C_5Q
.
\end{split}
\label{eq:22101621}
\end{align}

\smallskip
\noindent
\textit{Step III.}
We continue to compute \eqref{eq:22101621}. 
We next increment the order of the first and second terms of \eqref{eq:22101621} 
by substituting the following version of \eqref{eq:fundd},
\begin{align}
 \label{eq:3} 
\gamma_\|&
=-\tfrac 12 \gamma^2-\mathrm i\Gamma+(H-z)^*,
\end{align}
In fact, using \eqref{eq:3} we can rewrite the first and second terms of \eqref{eq:22101621} 
as 
\begin{align}
\begin{split}
&
2\beta(2\lambda)^{-1/2}\mathop{\mathrm{Re}}\bigl(\gamma_\|^2\widetilde\chi \theta'\theta^{2\beta-1}\gamma_\|\bigr)
+2\mathop{\mathrm{Re}}\bigl(\gamma_\|\gamma\cdot\widetilde\chi \theta^{2\beta}(\nabla^2 S)\gamma\bigr)
\\&
=
-\beta(2\lambda)^{-1/2}\mathop{\mathrm{Re}}\bigl(\gamma^2\gamma_\|\widetilde\chi \theta'\theta^{2\beta-1}\gamma_\|\bigr)
-\mathop{\mathrm{Re}}\bigl(\gamma^2\gamma\cdot\widetilde\chi \theta^{2\beta}(\nabla^2 S)\gamma\bigr)
\\&\phantom{{}={}}{}
+2\beta(2\lambda)^{-1/2}\mathop{\mathrm{Re}}\bigl((H-z)^*\gamma_\|\widetilde\chi \theta'\theta^{2\beta-1}\gamma_\|\bigr)
\\&\phantom{{}={}}{}
+2\mathop{\mathrm{Re}}\bigl((H-z)^*\gamma\cdot\widetilde\chi \theta^{2\beta}(\nabla^2 S)\gamma\bigr)
.
\end{split}
\label{eq:22101623}
\end{align}
\begin{subequations}
Let us discuss each term on the right-hand side. 
The first term of \eqref{eq:22101623} can be bounded 
by using \eqref{eq:22091918b}, the Cauchy--Schwarz inequality and Lemma~\ref{lem:algebra} as 
\begin{align}
\begin{split}
&-\beta(2\lambda)^{-1/2}\mathop{\mathrm{Re}}\bigl(\gamma^2\gamma_\|\widetilde\chi \theta'\theta^{2\beta-1}\gamma_\|\bigr)
\\&
= 
-\beta(2\lambda)^{-1/2}
\Bigl\{\gamma\cdot \gamma_\|\widetilde\chi \theta'\theta^{2\beta-1}\gamma_\|\gamma
+\mathop{\mathrm{Re}}\bigl(\gamma\cdot[\gamma,\gamma_\|]\widetilde\chi \theta'\theta^{2\beta-1}\gamma_\|\bigr)
\\&\phantom{{}={}}{}
+\mathop{\mathrm{Re}}\bigl(\gamma\cdot\gamma_\|[\gamma,\widetilde\chi \theta'\theta^{2\beta-1}]\gamma_\|\bigr)
+\mathop{\mathrm{Re}}\bigl(\gamma\cdot\gamma_\|\widetilde\chi \theta'\theta^{2\beta-1}[\gamma,\gamma_\|]\bigr)\Bigr\}
\\&
\le 
-\beta(2\lambda)^{-1/2}\gamma\cdot \gamma_\|\widetilde\chi \theta'\theta^{2\beta-1}\gamma_\|\gamma
+\epsilon\gamma_\|\theta'\theta^{2\beta-1}\gamma_\|
\\&\phantom{{}={}}{}
+\epsilon \gamma\cdot(\gamma\cdot \theta'\theta^{2\beta-1}\gamma)\gamma
+C_6Q
.
\end{split}
\label{eq:22101623a}
\end{align}
The second term of \eqref{eq:22101623} can be bounded by using 
\eqref{eq:3final2p} and Lemma~\ref{lem:algebra} as 
\begin{align}
\begin{split}
-\mathop{\mathrm{Re}}\bigl(\gamma^2\gamma\cdot\widetilde\chi \theta^{2\beta}(\nabla^2 S)\gamma\bigr)
&
=
-\gamma\cdot\bigl(\gamma\cdot\widetilde\chi \theta^{2\beta}(\nabla^2 S)\gamma\bigr)\gamma
+\tfrac12\gamma\cdot\bigl(\Delta\widetilde\chi \theta^{2\beta}(\nabla^2 S)\bigr)\gamma
\\
&\le 
-\gamma\cdot\bigl(\gamma\cdot\widetilde\chi \theta^{2\beta}(\nabla^2 S)\gamma\bigr)\gamma
+\epsilon\gamma_\|\theta'\theta^{2\beta-1}\gamma_\|
+C_7Q.
\end{split}
\label{eq:22101623b}
\end{align}
As for the third and fourth terms of \eqref{eq:22101623}, 
by the Cauchy-Schwarz inequality, \eqref{eq:3final2p} and Lemma~\ref{lem:algebra} we have 
\begin{align}
\begin{split}
&
2\beta(2\lambda)^{-1/2}\mathop{\mathrm{Re}}\bigl((H-z)^*\gamma_\|\widetilde\chi \theta'\theta^{2\beta-1}\gamma_\|\bigr)
+2\mathop{\mathrm{Re}}\bigl((H-z)^*\gamma\cdot\widetilde\chi \theta^{2\beta}(\nabla^2 S)\gamma\bigr)
\\&
\le 
C_{8}\bigl(\gamma_\|\theta'\theta^{2\beta-1}\gamma_\|\bigr)f^{1-\delta}\theta^{-2\beta}\bigl(\gamma_\|\theta'\theta^{2\beta-1}\gamma_\|\bigr)
\\&\phantom{{}={}}{}
+C_{8}\bigl(\gamma\cdot\theta^{2\beta}(\nabla^2 S)\gamma \bigr)f^{1-\delta}\theta^{-2\beta}\bigl(\gamma\cdot\theta^{2\beta}(\nabla^2 S)\gamma\bigr)
+C_{8}Q
\\&
\le 
\epsilon\gamma_\|\theta'\theta^{2\beta-1}\gamma_\|
+\epsilon \gamma\cdot(\gamma\cdot \theta'\theta^{2\beta-1}\gamma)\gamma
+C_{9}Q
.
\end{split}
\label{eq:22101623d}
\end{align} 
\end{subequations}
Hence by \eqref{eq:22101621}, \eqref{eq:22101623} and \eqref{eq:22101623a}--\eqref{eq:22101623d}
we conclude that 
\begin{align}
\begin{split}
2\mathop{\mathrm{Im}}(P(H-z))
&\le 
-\beta(2\lambda)^{-1/2}\gamma\cdot \gamma_\|\widetilde\chi \theta'\theta^{2\beta-1}\gamma_\|\gamma
-\gamma\cdot\bigl(\gamma\cdot\widetilde\chi \theta^{2\beta}(\nabla^2 S)\gamma\bigr)\gamma
\\&\phantom{{}={}}{}
+2\epsilon \gamma\cdot(\gamma\cdot\theta'\theta^{2\beta-1}\gamma)\gamma
+\bigl(2\beta(2\lambda)^{1/2}+5\epsilon\bigr)\gamma_\|\theta'\theta^{2\beta-1}\gamma_\|
\\&\phantom{{}={}}{}
+C_{10}Q
.
\end{split}
\label{eq:22101621b}
\end{align}

\smallskip
\noindent
\textit{Step IV.}
Now we use \eqref{eq:5b} to the right-hand side of \eqref{eq:22101621b}. 
Also using Lemmas~\ref{lem:basic3} and \ref{lem:algebra},
we obtain 
\begin{align}
\begin{split}
2\mathop{\mathrm{Im}}(P(H-z))
&\le 
-\bigl((2\lambda)^{1/2}\min\{1,\beta\}-3\epsilon\bigr)\gamma\cdot \bigl(\gamma\cdot \theta'\theta^{2\beta-1}\gamma\bigr)\gamma
\\&\phantom{{}={}}{}
+\bigl(2\beta(2\lambda)^{1/2}+6\epsilon\bigr)\gamma_\|\theta'\theta^{2\beta-1}\gamma_\|
+C_{11}Q
.
\end{split}
\label{eq:22101621bb}
\end{align}
Next we rewrite the first term of \eqref{eq:22101621bb} by using Lemma~\ref{lem:2210312250},
so that 
\begin{align*}
2\mathop{\mathrm{Im}}(P(H-z))
&\le 
-\bigl((2\lambda)^{1/2}\min\{4-2\beta,2\beta\}-19\epsilon\bigr)\gamma_\|\theta'\theta^{2\beta-1}\gamma_\|
\\&\phantom{{}={}}{}
+C_{12}\Gamma f^{-2}\theta^{2\beta}
+C_{12}Q
.
\end{align*}
Hence by \eqref{eq:221029} we obtain the assertion. 
\end{proof}

Next we prove Proposition~\ref{prop:22060215}.
Compared to Proposition~\ref{prop:220602},
it is much simpler.

\begin{proof}[Proof of Proposition~\ref{prop:22060215}]
Let us discuss only the upper sign. 
Similarly to the proof of Proposition~\ref{prop:220602}, 
we adopt $Q$ from \eqref{eq:2211021946}.
By the definition \eqref{eq:220920}, 
the Cauchy--Schwarz inequality, Lemma~\ref{lem:basic3}, 
\eqref{eq:6} and Lemma~\ref{lem:algebra}
we can bound for any $\epsilon\in (0,1]$
\begin{align*}
\mathop{\mathrm{Im}}(P(H-z))
&
\ge 
-\epsilon\gamma_\|\theta^{2\beta}\gamma_\|\theta'\theta^{-2\beta-1}\gamma_\|\theta^{2\beta}\gamma_\|
-C_1\epsilon^{-1}Q
\\&
\ge 
-C_2\epsilon\gamma\cdot\bigl(\gamma\cdot\theta'\theta^{2\beta-1}\gamma\bigr)\gamma
-C_2\epsilon\gamma_\|\theta'\theta^{2\beta-1}\gamma_\|
-C_2\epsilon^{-1}Q,
\end{align*}
where $C_*>0$ are independent of $\epsilon\in (0,1]$. 
Then by Lemma~\ref{lem:2210312250}
\begin{align*}
\mathop{\mathrm{Im}}(P(H-z))
&
\ge 
-C_3\epsilon \gamma_\|\cdot\theta'\theta^{2\beta-1}\gamma_\|
-C_3\Gamma f^{-2}\theta^{2\beta}
-C_3\epsilon^{-1}Q
,
\end{align*}
and we are done. 
\end{proof}

\begin{remark}\label{rem:230112}
We can modify the arguments of the subsection to be applicable to the case $l=3$, 
avoiding fourth derivatives of $S$.
For that we should employ 
\begin{align}
\widetilde\gamma_{\|}
=(\nabla S)\cdot \gamma
-\tfrac{\mathrm i(d-1)}2(2\lambda)^{1/2} f^{-1}
,\quad 
\widetilde\beta_c=1+\sigma
\label{eq:221102}
\end{align}
instead of $\gamma_\|$, $\beta_c$, respectively. 
Note, although $\widetilde\gamma_{\|}$ is not symmetric, it well approximates 
$\gamma_\|$ thanks to Theorem~\ref{thm:main result2},
and we can avoid $\Delta^2S$ coming from \eqref{eq:22091919}. 
The fourth order derivatives of $S$ appear also from other parts of the 
above arguments, 
 but we can manage them by the Cauchy--Schwarz inequality. 
We omit the details. 
Note also that, although $\widetilde\beta_c$ is worse than $\beta_c$,
it is still greater than $1$, and the associated radiation condition bounds 
are stronger than the ordinary ones.
\end{remark}

\subsection{Proof of strong radiation condition bounds}\label{subsec:22091825}

Finally in this section we prove Theorem~\ref{thm:proof-strong-bound}.
We will use the standard limiting absorption principle bounds on the following form. 

\begin{thm}\label{thm:221105}
There exists $C>0$ such that 
uniformly in $z\in I_\pm$ and $\psi\in\mathcal B$
\begin{align*}
\|R(z)\psi\|_{\mathcal B^*}
\le C\|\psi\|_{\mathcal B},\quad 
\|\Delta R(z)\psi\|_{\mathcal B^*}
\le C\|\psi\|_{\mathcal B}.
\end{align*}
Moreover, for any $t>1/2$ there exist uniform limits 
in $\lambda\in I$:
\[
R(\lambda\pm \mathrm i0)=\lim_{I_\pm \ni z\to \lambda} R(z),\quad 
\Delta R(\lambda\pm \mathrm i0)=\lim_{I_\pm \ni z\to \lambda} \Delta R(z)
\]
in the norm topology of $\mathcal L(L^2_t,L^2_{-t})$.
\end{thm}
\begin{remark}
We do not need to assume $q\equiv 0$ for this result.
\end{remark}
\begin{proof}
This is the standard result in the theory of the Schr\"odinger operators, 
and we omit a proof. We refer the reader to \cite{AIIS2}. 
\end{proof}

\begin{proof}[Proof of Theorem~\ref{thm:proof-strong-bound}]
We note by the density argument
it suffices to prove the asserted bounds 
\eqref{eq:gamma1a}--\eqref{eq:gamma2} for $\psi\in C^\infty_{\mathrm c}(\mathbb R^d)$.

\smallskip
\noindent
\textit{(1)}\ 
Let $\beta\in (0,\beta_c)$, and choose {$\delta>0$} such that 
\begin{align}
2\beta+3\delta<2\beta_c,\quad \delta\le 2\beta.
\label{eq:22110523}
\end{align}
By Propositions~\ref{prop:220602} and \ref{prop:22060215}
we can find $C_1>0$ and $\nu_0\ge 1$ such that 
uniformly in $z=\lambda\pm \mathrm i\Gamma\in I_\pm$ and $\nu\ge \nu_0$
\begin{align}
\begin{split}
\gamma_\|\theta'\theta^{2\beta-1}\gamma_\|
&\le 
C_1\Gamma f^{-2}\theta^{2\beta}
+C_1f^{-1-2\beta_c+3\delta}\theta^{2\beta}
 \\&\phantom{{}={}}{}
+C_1(H-z)^*f^{1+\delta}\theta^{2\beta-\delta}(H-z)
.
\end{split}
\label{eq:221101}
\end{align}
Take the expectation of the above inequality in the state $\phi=R(z)\psi$ for any $z=\lambda\pm\mathrm i\Gamma\in I_\pm$ and 
$\psi\in C^\infty_{\mathrm c}(\mathbb R^d)$, and we obtain by \eqref{eq:22110523} and 
Theorem~\ref{thm:221105}
\begin{align*}
\|\theta'^{1/2}\theta^{\beta-1/2}\gamma_\|\phi\|^2_{\mathcal H}
&\le 
C_1\Gamma \|f^{-1}\theta^{\beta}\phi\|_{\mathcal H}^2
+C_1\|f^{-1/2-\beta_c+3\delta/2}\theta^{\beta}\phi\|_{\mathcal H}^2
 \\&\phantom{{}={}}{}
+C_1\|f^{(1+\delta)/2}\theta^{\beta-\delta/2}\psi\|_{\mathcal H}^2
\\&
\le 
C_1\Gamma \|f^{-1}\theta^{\beta}\phi\|_{\mathcal H}^2
+C_2\|f^\beta\psi\|_{L^2_{1/2}}^2
.
\end{align*}
Next we take the limit $\Gamma\to 0_+$, and obtain by Theorem~\ref{thm:221105} 
\begin{align*}
\|\theta'^{1/2}\theta^{\beta-1/2}\gamma_\|R(\lambda\pm\mathrm i0)\psi\|_{\mathcal H}
\le 
C_2\|f^\beta\psi\|_{L^2_{1/2}}
.
\end{align*} 
Finally we let $\nu\to \infty$,
and then by the monotone convergence theorem the bound \eqref{eq:gamma1a} follows.

Combining Lemma~\ref{lem:2210312250} and \eqref{eq:221101}, we also have 
\begin{align*}
\gamma\cdot\bigl(\gamma\cdot\theta'\theta^{2\beta-1}\gamma\bigr)\gamma
&\le 
C_3\Gamma f^{-2}\theta^{2\beta}
+C_3f^{-1-2\beta_c+3\delta}\theta^{2\beta}
 \\&\phantom{{}={}}{}
+C_3(H-z)^*f^{1+\delta}\theta^{2\beta-\delta}(H-z)
.
\end{align*}
Hence we can verify \eqref{eq:gamma1b} similarly to \eqref{eq:gamma1a}. 

\smallskip
\noindent
\textit{(2)}\ 
Let $\psi,\psi'\in C^\infty_{\mathrm c}(\mathbb R^d)$, and consider a quantity 
\[F(\zeta)
=\bigl\langle \psi',(\gamma\cdot f^{2\beta'}\gamma)^{\zeta}f^{-t}R(\lambda\pm\mathrm i0)f^{-t}f^{-2\beta'\zeta}\psi\bigr\rangle
.
\]
It is obviously analytic in $0<\mathop{\mathrm{Re}}\zeta<1$. 
For $\mathop{\mathrm{Re}}\zeta=0$ we have by Theorem~\ref{thm:221105}
\[
|F(\zeta)|
\le C_4\|\psi'\|_{\mathcal H}\|\psi\|_{\mathcal H}, 
\]
and for $\mathop{\mathrm{Re}}\zeta=1$ by \eqref{eq:gamma1b}, \eqref{eq:gamma1a} and Theorem~\ref{thm:221105}
\[
|F(\zeta)|
\le C_5\|\psi'\|_{\mathcal H}\|\psi\|_{\mathcal H}. 
\]
Hence we obtain by the Hadamard three-lines theorem
\[
\bigl|
\bigl\langle \psi',(\gamma\cdot f^{2\beta'}\gamma)^{1/2}f^{-t}R(\lambda\pm\mathrm i0)f^{-t}f^{-\beta'}\psi\bigr\rangle\bigr|
=
|F(\tfrac12)|
\le C_6\|\psi'\|_{\mathcal H}\|\psi\|_{\mathcal H},
\]
or for any $\psi''\in C^\infty_{\mathrm c}(\mathbb R^d)$
\[
\|(\gamma\cdot f^{2\beta'}\gamma)^{1/2}f^{-t}R(\lambda\pm\mathrm i0)\psi''\|_{\mathcal H}
\le C_6\|f^{\beta'}\psi''\|_{L^2_t}.
\]
However, we can rewrite the square of the left-hand side 
by using the inner product and Theorem~\ref{thm:221105} as 
\begin{align*}
\|(\gamma\cdot f^{2\beta'}\gamma)^{1/2}f^{-t}R(\lambda\pm\mathrm i0)\psi''\|_{\mathcal H}^2
&=
\sum_{i=1}^d\|f^{\beta'}\gamma_i f^{-t}R(\lambda\pm\mathrm i0)\psi''\|_{\mathcal H}^2
\\&
\ge 
\sum_{i=1}^d\|f^{\beta'}\gamma_i R(\lambda\pm\mathrm i0)\psi''\|_{L^2_{-t}}^2
-C_7\|\psi''\|_{\mathcal B}^2.
\end{align*}
Therefore we obtain \eqref{eq:gamma2}.
\end{proof}

\section{Stationary scattering theory}
\label{sec:Generalized Fourier transform}

In this section we discuss the stationary scattering theory for $H$, 
proving Theorems~\ref{thm:comp-gener-four}, \ref{thm:char-gener-eigenf-1} and \ref{thm:221207} and Corollary~\ref{cor:230623}. 
In order to use the strong radiation condition bounds
 of Theorem~\ref{thm:proof-strong-bound} 
we need more regularity for the potential than required in these assertions. 
Subsection~\ref{subsec:230329} reviews a decomposition $V=V_{\mathrm S}+V_{\mathrm L}$ 
due to H\"ormander \cite[Lemma~30.1.1]{H1}, 
so that the strong radiation condition bounds are available for $H_{\mathrm L}=-\tfrac12\Delta+V_{\mathrm L}$. 
Then in Subsection~\ref{subsec:Generalized Fourier transform for rho>1/2 I} 
we introduce the spherical eikonal coordinates associated with $V_{\mathrm L}$, and study its geometry. 
In these coordinates we can quickly construct the 
stationary wave operators for $H_{\mathrm L}$, mimicking the procedure of \cite{Sk}. 
They are then fused to those for $H$ by a change of coordinates and the second resolvent identity. 
This is implemented in Subsection~\ref{subsubsec:Generalized Fourier transform for rho=1},
and the proofs of Theorem~\ref{thm:comp-gener-four} and Corollary~\ref{cor:230623} are done. 
The proofs of Theorems~\ref{thm:char-gener-eigenf-1} and \ref{thm:221207} are rather routine,
and they are presented in Subsections~\ref{subsec:23032917} and \ref{subsec:23032918},
respectively. Then our stationary scattering theory is completed.

\subsection{H\"ormander's regularization}\label{subsec:230329}

The asserted Theorem~\ref{thm:proof-strong-bound} requires 
 four deri\-va\-tives on the potential $V$, which clearly is not at
 disposal 
for a $2$-admissible potential. Consequently, to implement the radiation condition
bounds of the theorem
 we need first to 
 regularize $V$. This is done by using the scheme of H\"ormander \cite[Lemma~30.1.1]{H1}.

\begin{lemma}\label{lem:230111}
Suppose Condition~\ref{cond:220525} with $l=2$. 
\begin{enumerate}
\item\label{item:2307260}
For any $\rho\in (0,\sigma)$ there exists a splitting 
\[V=V_{\mathrm S}+V_{\mathrm L};\quad 
V_{\mathrm S}\in C^2(\mathbb R^d;\mathbb R),\ \ V_{\mathrm L}\in C^\infty(\mathbb R^d;\mathbb R),\]
satisfying: 
There exists $C>0$ such that for any $|\alpha|\le 2$ and $x\in\mathbb R^d$
\begin{subequations}
\begin{align}
|\partial^\alpha V_{\mathrm S}(x)|
\le C\langle x\rangle^{-1-\sigma+\rho-|\alpha|(\rho+1)/2},
\label{eq:230115}
\end{align}
and for any $\alpha\in\mathbb N_0^d$ 
there exists $C_\alpha>0$ such that for any $x\in\mathbb R^d$
\begin{equation}\label{eq:cond22}
|\partial^\alpha V_{\mathrm L}(x)|\le C_\alpha\langle x\rangle^{-m(|\alpha|)},
 \end{equation} 
where $m$ is defined by
 \eqref{eq:cond22bb} with the parameters $\sigma$ and $\rho$ (and for
 any $l\geq 2$). 
\end{subequations}
\item\label{item:2307261}
For any $\rho\in (0,\sigma)$ and $\epsilon,\delta>0$ there exists a splitting 
\[V=V_{\mathrm S}'+V_{\mathrm L}';\quad 
V_{\mathrm S}'\in C^2(\mathbb R^d;\mathbb R),\ \ V_{\mathrm L}'\in C^\infty(\mathbb R^d;\mathbb R),\]
satisfying: For 
any $|\alpha|\le 2$ and $x\in\mathbb R^d$
\begin{subequations}
\begin{align*}
|\partial^\alpha V_{\mathrm S}'(x)|
\le \epsilon\langle x\rangle^{-1-\sigma+\rho+\delta-|\alpha|(\rho+1)/2},
\end{align*}
and for any $\alpha\in\mathbb N_0^d$ 
there exists $C_\alpha>0$ such that for any $x\in\mathbb R^d$
\begin{equation*}
|\partial^\alpha V_{\mathrm L}'(x)|\le C_\alpha\langle x\rangle^{-m(|\alpha|)},
 \end{equation*} 
 where $m$ is given as in \eqref{eq:cond22}. 
The constants $C_\alpha$ can for $|\alpha|\leq 2$  be chosen independently of  $\epsilon,\delta>0$.
\end{subequations}
\end{enumerate}
\end{lemma}
\begin{remarks}
\begin{enumerate}[1)]
\item
Ikebe--Isozaki \cite{II} adopted a decomposition similar to
(\ref{item:2307260}) for classical $C^4$ long-range potentials.
Note that our $\partial^\alpha V_{\mathrm L}$ has worse decay rate than theirs for $|\alpha|\ge 3$. 
\item
The assertion (\ref{item:2307261}) will only be employed in Subsection~\ref{subsubsec:230726}, 
the last part of the paper, for the proof of
Theorem~\ref{thm:time-depend-theory2} (\ref{item:23071914}). 
\end{enumerate}
\end{remarks}

\begin{proof}
\textit{(\ref{item:2307260})}\ 
Although the bounds for $V_{\mathrm S}$ are slightly better than those in \cite[Lemma~30.1.1]{H1}, the same proof works well. 
Let us review it since we will use its modification below for the assertion (\ref{item:2307261}). 
Fix any real $\eta\in C^\infty(\{|x|<2\})$ with $\eta=1$ for $|x|\le 1$, and set 
\[
V_0(x)=\eta(x) V(x),\quad V_n(x)=\bigl(\eta(2^{-n}x)-\eta(2^{1-n}x)\bigr)V(x)\ \ \text{for }n\in\mathbb N. 
\]
We also take a real $\chi\in C^\infty_{\mathrm c}(\mathbb R^d)$ such that 
\[
\int_{\mathbb R^d}\chi\,\mathrm dx=1,\quad 
\int_{\mathbb R^d}x_j \chi(x)\,\mathrm dx=0\ \ \text{for }j=1,\dots,d, 
\]
and set 
\[
\chi_n(x)=2^{-dn(1+\rho)/2}\chi(2^{-n(1+\rho)/2}x).
\]
Then we define 
\[
V_{\mathrm L}=\sum_{n\in\mathbb N_0}\chi_n*V_n,\quad V_{\mathrm S}=V-V_{\mathrm L},
\]
and they satisfy the asserted bounds. We omit further details. 

\smallskip
\noindent
\textit{(2)}\ 
For $N_1, N_2\in\mathbb N$ we consider 
\[
V_{\mathrm L}'=\sum_{n\geq N_1}\chi_{n}*V_n+\sum_{n<N_1}\chi_{N_2}*V_n,\quad V_{\mathrm S}'=V-V_{\mathrm L}'.
\]
With a proper adjustment of the parameters (fix first $N_1$ large and
then suitably large $N_2$) indeed $V_{\mathrm S}'$ and $V_{\mathrm L}'$ satisfy the asserted bounds. 
We are done. 
\end{proof}

The main part of the section is devoted to the analysis of $V_{\mathrm
 L}$ from Lemma~\ref{lem:230111} (\ref{item:2307260}),
while the effects from $V_{\mathrm S}+q$ are taken into account only
in the last steps. 
Note that the bound \eqref{eq:cond22} clearly agrees with Condition~\ref{cond:220525} for any $l\ge 2$,
and thus the results from the previous sections are available for $V_{\mathrm L}$. 
We denote 
\begin{align*}
H_{\mathrm L}=-\tfrac12\Delta+V_{\mathrm L},
\quad 
R_{\mathrm L}(z)=(H_{\mathrm L}-z)^{-1}
\ \ \text{for }z\in\mathbb C\setminus\sigma(H_{\mathrm L}),
\end{align*}
and 
\begin{align*}
R_{\mathrm L}(\lambda\pm \i 0)
=\swslim_{ z\to \lambda\pm \i 0_+}R_{\mathrm L}(z)\ \ \text{in }\mathcal L(\mathcal B,\mathcal B^*)
\ \ \text{for }\lambda>0.
\end{align*}
Throughout the section we fix any closed interval $I\subset \R_+$,
and let $S_{\mathrm L}\in C^{l'}(I\times(\mathbb R^d\setminus\{0\}))$ 
and $s_{\mathrm L}\in C^{l'}(I\times\mathbb R^d)$ be 
given as in Theorem~\ref{thm:main result2} for $V_{\mathrm L}$ and
any fixed $l'>1+2/\rho$. 
We will actually possibly need to take $R>0$ larger than needed
for Theorem~\ref{thm:main result2} in 
Subsections~\ref{subsubsec:Normalized asymptotic velocity} and \ref{subsec:23070422},
implementing Corollary~\ref{cor:epsSmall}, 
but the $R$-dependence is suppressed. 

\begin{remark}\label{remark:l_prime}
This specific requirement $l'>1+2/\rho$ will be needed only in the proof of Lemma~\ref{lem:time-depend-theory}. 
Otherwise it suffices to take $l'=4$, so that Theorem~\ref{thm:proof-strong-bound} is available. 
Although Theorem~\ref{thm:main result2} does not provide bounds for
$k+|\alpha|> l'$, nevertheless the function $s_{\mathrm L}\in
C^\infty(I
\times \R^d)$, and similarly for $S_{\mathrm L}$. This is a
consequence of the fact that 
$V_{\mathrm L}\in C^\infty(\mathbb R^d)$ and the implicit function
theorem, see the proof of Theorem~\ref{thm:main result2} \eqref{item:22120919a}. 
\end{remark}

\subsection{Spherical eikonal coordinates}
\label{subsec:Generalized Fourier transform for rho>1/2 I}

\subsubsection{Eikonal flow at fixed energy}

In order to define so-called spherical eikonal coordinates, 
in which the function $S_{\mathrm L}(\lambda,\cdot)$, $\lambda\in I$, from the last
subsection plays the roll
of \textit{eikonal distance} from the origin, 
consider the \textit{eikonal flow} $y$ satisfying, for any given $(\lambda,\theta)\in I\times\mathbb S^{d-1}$, 
\begin{subequations}
\begin{equation}
 \label{eq:flow}
 \tfrac{\partial}{\partial t}y(\lambda,t,\theta)
=\parb{|\nabla S_{\mathrm L}|^{-2}\nabla S_{\mathrm L}}(\lambda,y(\lambda,t,\theta))
\end{equation}
with
\begin{equation}
 \label{eq:flowb}
\lim_{t\to 0_+}y(\lambda,t,\theta)=0\quad \mand \quad 
\lim_{t\to 0_+}\tfrac{\partial}{\partial t}y(\lambda,t,\theta)=(2\lambda)^{-1/2}\theta.
\end{equation} 
\end{subequations}

\begin{lemma}\label{lem:221210}
The solution $y$ to \eqref{eq:flow} with \eqref{eq:flowb} is smooth in 
$(\lambda,t,\theta)\in I\times\R_+\times \S^{d-1}$, 
and for any $\lambda\in I$ it induces a (smooth) diffeomorphism 
\begin{align}
y(\lambda,\cdot,\cdot)\colon \R_+\times \S^{d-1}\to \R^d\setminus\{0\}.
\label{eq:22120922}
\end{align}
In addition, for any $(\lambda,t,\theta)\in I\times\R_+\times \mathbb S^{d-1}$ one has 
\begin{align}
S_{\mathrm L}(\lambda,y(\lambda,t,\theta))=t.
\label{eq:2212092256}
\end{align}
\end{lemma}
\begin{remarks}\label{rem:230218}
 \begin{enumerate}[1)]
 \item \label{item:geom1} From a geometric point of view the flow 
 $y(\lambda,\cdot,\cdot)$ is nothing but the exponential map from
 the unit-sphere in the tangent space at the origin in the 
 space $\R^d$ equipped with the metric $g_L=2\bigl(\lambda-\chi_RV_L\bigr)\,\d x^2$.
 \item \label{item:geom2}
 The flow $y(\lambda,\cdot,\cdot)$ constitutes a
family of reparametrized classical orbits of energy $\lambda$ 
for the classical Hamiltonian 
\[
H^{\mathrm{cl}}_{\mathrm L}(x,\xi)
=\tfrac12\xi^2+\chi_R(x)V_{\mathrm L}(x)
.
\] 
In fact, if we set 
\begin{align*}
z(\tau)=y(\lambda,t,\theta),\quad 
\tau=\int_0^t|\nabla S_{\mathrm L}(\lambda,y(\lambda,s,\theta))|^{-2}\,\mathrm ds,
\end{align*}
then by using \eqref{eq:2212116}
\begin{align*}
\tfrac{\mathrm d}{\mathrm d\tau}z=\nabla S_{\mathrm L},\quad 
\tfrac{\mathrm d^2}{\mathrm d\tau^2}z
=(\nabla^2 S_{\mathrm L})(\nabla S_{\mathrm L})
=-\nabla (\chi_RV_{\mathrm L}),
\end{align*}
cf.\ \eqref{eq:200702}. 
The reparametrizing factor $|\nabla S_{\mathrm L}|^{-2}$ in \eqref{eq:flow} 
 is the proper normalization guaranteeing \eqref{eq:2212092256}.
This point of view was taken in the proof of
an analogous statement \cite[Proposition 2.2]{ACH}. 
\end{enumerate}
\end{remarks}
\begin{proof}[Proof of Lemma~\ref{lem:221210}]
By Theorem~\ref{thm:main result2} \eqref{item:22120919}
we conclude that $y$ is defined at least on $I\times (0,(2\lambda)^{1/2}R]\times \mathbb S^{d-1}$,
and 
\begin{equation}
 y(\lambda,t,\theta)=(2\lambda)^{-1/2} t\theta \ \ \text{for }(\lambda,t,\theta)
\in I\times \bigl(0,(2\lambda)^{1/2}R\bigr]\times \mathbb S^{d-1}.
 \label{eq:22120923}
\end{equation} 
On the other hand, by \eqref{eq:flow}, \eqref{eq:flowb} and \eqref{eq:22120923}
\[\tfrac{\partial }{\partial t}S_{\mathrm
 L}(\lambda,y(\lambda,t,\theta))=1\quad \mand \quad
\lim_{t\to 0_+}S_{\mathrm L}\bigl(\lambda,y\bigl(\lambda,t,\theta\bigr)\bigr)=0.\]
This implies $y(\lambda,t,\theta)$ never hits the origin for
$t>(2\lambda)^{1/2}R$, and neither it can 
 reach infinity in finite time.
Hence the vector field $|\nabla S_{\mathrm L}|^{-2}\nabla S_{\mathrm L}$ is forward complete,
and $y$ is globally defined on $I\times \R_+\times \S^{d-1}$,
satisfying \eqref{eq:2212092256}. 

Next, we note that $y(\lambda,\cdot,\cdot)$ is bijective. 
In fact, by the uniqueness for the initial-value problem of ODEs the injectivity follows. 
To see the surjectivity, starting at any given
 point
 in $x\in\R^d\setminus\{0\}$, we solve the ODE
 \eqref{eq:flow} in the backward time-direction. 
Then we  obtain a
 `crossing' initial
 angle $\theta$ from where indeed the forward flow for a proper 
 time $t$ satisfies $y(\lambda,t,\theta)=x$. (Alternatively, $y(\lambda,\cdot,\cdot)$ is bijective since any $x\in\mathbb R^d\setminus\{0\}$ 
can be connected to the origin by a unique geodesic, cf. Theorems~\ref{thm:main result2} \eqref{item:22120919}
and \ref{thm:main result} \eqref{item: Condition 1}.)

{Finally we show that $y(\lambda,\cdot,\cdot)$ is a diffeomorphism. 
Note that the maps $y(\lambda,\cdot,\cdot)$ and $y(\cdot,\cdot,\cdot)$ are
 smooth, viewed as a solution to an initial-value problem with
 data specified on
 the sphere $(2\lambda)^{-1/2}\S^{d-1}\simeq \S^{d-1}$ at time $t=1$. Thus}
 it suffices to check the non-degeneracy of the map. 
Take any local coordinates $\theta'=(\theta_2',\dots,\theta_d')$ of $\mathbb S^{d-1}$,
and let $J'$ be the Jacobian of \eqref{eq:22120922}
in these local coordinates. 
Now we claim that for any $(\lambda,t,\theta')$
\begin{align}
\partial_tJ'(\lambda,t,\theta' )
=\bigl(\nabla\cdot |\nabla S_{\mathrm L}|^{-2}\nabla S_{\mathrm L}\bigr)(\lambda,y(\lambda,t,\theta' ))
J'(\lambda,t,\theta' ).
\label{eq:2212115}
\end{align}
In fact, differentiating the defining expression of the Jacobian, we can write 
\[\partial_t J'
=\sum_{i=1}^d \det \mathcal J^{(i)},\]
where $\mathcal J^{(i)}$ are matrix-valued functions whose components are given by 
\[\mathcal J^{(i)}_{jk}=\begin{cases}
\partial_ky_j &\text{for }j\neq i,\\
\partial_t\partial_ky_i&\text{for }j=i
\end{cases};
\quad 
j=1,\dots,d,\ \, k=t, \theta'_2,\dots,\theta'_d.\]
However, thanks to the flow equation \eqref{eq:flow}, we can compute 
\begin{align*}
\partial_t\partial_ky_i
=\partial_k\bigl(|\nabla S_{\mathrm L}|^{-2}(\nabla S_{\mathrm L})_i\bigr)
=\bigl(\partial_l|\nabla S_{\mathrm L}|^{-2}(\nabla S_{\mathrm L})_i\bigr)\partial_ky_l.
\end{align*}
Thus we obtain the claimed identity \eqref{eq:2212115}, 
noting that the determinant is alternating and multilinear. 
By \eqref{eq:22120923} $J'$ is non-vanishing for $t\in (0,(2\lambda)^{1/2}R]$, 
and hence with \eqref{eq:2212115} we can conclude that so it is for
all $t>0$. We are done. 
\end{proof}

Now the spherical eikonal coordinates are defined as follows. 
\begin{defn}
The \textit{spherical eikonal coordinates} on $\mathbb
R^d\setminus\{0\}$ at energy $\lambda\in I$
are the entries of the inverse of \eqref{eq:22120922}.
We denote them by $(t,\theta)$, or by $(\lambda,t,\theta)$ to clarify the $\lambda$-dependence. 
They are also denoted by $(t,\theta')$ or $(\lambda,t,\theta')$ if local coordinates $\theta'=(\theta_2',\dots,\theta_d')$ 
of $\mathbb S^{d-1}$ are specified. We call $t=S_{\mathrm
 L}(\lambda,y(\lambda,t,\theta)) $ and $\theta$ (or $\theta'$) the \emph{radial and
spherical components} of the spherical eikonal coordinates, respectively.
\end{defn}

The Euclidean metric splits into the radial and spherical components in the spherical
eikonal coordinates. 
Let us present it as a corollary, although we will not use it in the present paper. 
The entries of $\theta'$ are distinguished by Greek indices always
running over $2,\dots,d$, while the entries of $y$ canonically are
distinguished by Roman ones always running over $1,\dots,d$. 

\begin{corollary}\label{cor:22121914}
Let $(t,\theta')$ be spherical eikonal coordinates at any $\lambda\in I$. Then 
\begin{align*}
(\partial_ty_i(\lambda,t,\theta'))(\partial_t y_i(\lambda,t,\theta'))&=|(\nabla S_{\mathrm L})(\lambda,y(\lambda,t,\theta'))|^{-2}
,\\ 
(\partial_ty_i(\lambda,t,\theta'))(\partial_\alpha y_i(\lambda,t,\theta'))&=0
\quad 
\text{for }\alpha=2,\dots,d.
\end{align*}
In addition, the Euclidean metric takes the form
\begin{align*}
\mathrm dx^2
&=
|\nabla S_{\mathrm L}|^{-2}\,\mathrm dt^2
+g_{\alpha\beta}\,\mathrm d\theta'_\alpha\mathrm d\theta'_\beta
. 
\end{align*}

\end{corollary}
\begin{proof}
The former formulas follow from \eqref{eq:flow} and \eqref{eq:2212092256},
and the last formula is an immediate consequence of the former ones. 
\end{proof}

\subsubsection{Volume and surface measures}

Let $J\colon I\times\R_+\times\mathbb S^{d-1}\to[0,\infty)$ be a function 
such that the Euclidean volume measure can be expressed in $(\lambda,t,\theta)$ as 
\begin{align}
\mathrm dx(\lambda,t,\theta)=\mathrm dx_1\cdots\mathrm dx_d(\lambda,t,\theta)=J(\lambda,t,\theta)\,\mathrm dt\mathrm dA(\theta),
\label{eq:23010713}
\end{align}
where $\mathrm dA$ denotes the standard surface measure on $\S^{d-1}$.
In fact, if we take any local coordinates $\theta'=(\theta'_2,\dots,\theta'_d)$ of $\mathbb S^{d-1}$,
and let $J'$ be the Jacobian from the proof of Lemma~\ref{lem:221210}, 
then $J$ can be computed through the relation 
\begin{align}
\begin{split}
 \label{eq:surf_meas}
\d A_{\lambda,t}(\theta(\theta')):=J(\lambda,t,\theta(\theta'))\,\mathrm dA(\theta(\theta'))
=J'(\lambda,t,\theta')\,\mathrm d\theta'_2\cdots\mathrm d\theta'_d.
\end{split}
\end{align}
Note for any $\phi\in L^1(\mathbb R^d)$
\begin{align}
 \label{eq:co_area}
\int_{\mathbb R^d} \phi(x)\, \d x&
=\int_0^\infty \!\d t\, \int_{\mathbb
 S^{d-1}}\phi(y(\lambda,t,\theta))\,\d
 A_{\lambda,t}(\theta). 
\end{align}
By the co-area formula \cite[Theorem C.5]{Ev} the element
 $|\nabla S_{\mathrm L}|(\lambda,y(\lambda,t,\theta)) \,\d
 A_{\lambda,t}(\theta)$ is the
 Euclidean surface element on the
 distorted sphere $\set{S_L=t}$.

\begin{lemma}\label{lem:230107}
\begin{enumerate}
\item
One has an explicit formula 
\begin{align*}
J(\lambda,t,\theta)
&=
\sqrt{2\lambda} R^{d-1}\bigl|(\nabla S_{\mathrm L})(\lambda,y(\lambda,t,\theta))\bigr|^{-2}
\\&\phantom{{}={}}{}
\cdot\exp\biggl(\int_{(2\lambda)^{1/2}R}^t\bigl(|\nabla S_{\mathrm L}|^{-2}\Delta S_{\mathrm L}\bigr) (
\lambda,y (\lambda,\tau,\theta))\, \d\tau \biggr)
.
\end{align*}

\item
There exist the following limits 
uniformly in $(\lambda,\theta)\in I\times \mathbb S^{d-1}$:
\begin{align*}
\qquad\quad 
\lim_{t\to 0_+}t^{-(d-1)}J(\lambda,t,\theta)&=(2\lambda)^{-d/2},\\
\lim_{t\to \infty}t^{-(d-1)}J(\lambda,t,\theta)&=:J_+(\lambda,\theta)>0.
\end{align*}

\end{enumerate}
\end{lemma}
\begin{remark*}
For similar assertions in a wider 
geometric setting, see \cite{ IS2}.
\end{remark*}

\begin{proof}
\textit{(1)}\ 
Note that for any $(\lambda,t,\theta)\in I\times\R_+\times\mathbb S^{d-1}$
\begin{align}
\partial_t J(\lambda,t,\theta )
=\bigl(\nabla\cdot |\nabla S_{\mathrm L}|^{-2}\nabla S_{\mathrm L}\bigr)(\lambda,y(\lambda,t,\theta ))
J(\lambda,t,\theta ).
\label{eq:2212115b}
\end{align}
In fact \eqref{eq:2212115b} follows from \eqref{eq:2212115} 
since in any local coordinates of $\mathbb S^{d-1}$, 
$J'/J$ is a function independent of $t$, cf.\ \eqref{eq:surf_meas}. 
Then, since 
\begin{align*}
\nabla\cdot \bigl(|\nabla S_{\mathrm L}|^{-2}\nabla S_{\mathrm L}\bigr)
&
=
|\nabla S_{\mathrm L}|^{-2}\Delta S_{\mathrm L}
-\partial_t \ln \bigl(|\nabla S_{\mathrm L}|^{2}\bigr),
\end{align*}
it follows from \eqref{eq:2212115b} that 
for some $C(\lambda,\theta)\ge 0$ 
\begin{align*}
J(\lambda,t,\theta)
&=C(\lambda,\theta)\bigl|(\nabla S_{\mathrm L})(\lambda,y(\lambda,t,\theta))\bigr|^{-2}
\\&\phantom{{}={}}{}\cdot
\exp\bigg(\int_{(2\lambda)^{1/2}R}^t\bigl(|\nabla S_{\mathrm L}|^{-2}\Delta S_{\mathrm L}\bigr) (\lambda,y (\lambda,\tau,\theta))\, \d\tau \biggr).
\end{align*}
We can determine $C(\lambda,\theta)=\sqrt{2\lambda} R^{d-1}$ 
by \eqref{eq:2212116} and \eqref{eq:22120923},
which shows the assertion. 

\smallskip
\noindent
\textit{(2)}\ 
The assertion is clear for $t\to 0_+$ thanks to the explicit expression \eqref{eq:22120923}. 
Let $t\ge (2\lambda)^{1/2}R$. 
Then by \eqref{eq:2212116}, Theorem~\ref{thm:main result2}
and \eqref{eq:2212092256} it follows that 
\begin{align*}
&\bigl||(\nabla S_{\mathrm L}) (\lambda,y (\lambda,t,\theta))|^{-2}
-(2\lambda)^{-1}\bigr|\le C_1t^{-\sigma},\\
&\bigl|\bigl(|\nabla S_{\mathrm L}|^{-2}\Delta S_{\mathrm L}\bigr) (\lambda,y (\lambda,\tau,\theta))
-(d-1)\tau^{-1}\bigr|\le C_1\tau^{-1-\sigma}.
\end{align*}
This verifies the asymptotics for $t\to \infty$. 
\end{proof}

\subsubsection{Change of coordinates at infinity}
\label{subsubsec:Normalized asymptotic velocity}

Here we investigate relation between the spherical eikonal coordinates and the ordinary spherical coordinates. 
For each $\lambda\in I$ we can change from 
$(t,\theta)\in \R_+\times\mathbb S^{d-1}$ 
to 
$(r,\omega)\in \R_+\times\mathbb S^{d-1}$ 
through 
\begin{align}
r(\lambda,t,\theta)=|y(\lambda,t,\theta)|,\quad 
\omega(\lambda,t,\theta)=\hat y(\lambda,t,\theta)=|y(\lambda,t,\theta)|^{-1}y(\lambda,t,\theta).
\label{eq:23010712}
\end{align}
We claim that \eqref{eq:23010712} induces a $C^1$-diffeomorphism
of $\mathbb S^{d-1}$ at infinity provided `$s_{\mathrm L}$ is
small' (which thanks to Corollary~\ref{cor:epsSmall} can be assumed by taking $R$ sufficiently
big). More precisely we claim  the following assertion.

\begin{lemma}\label{lem:2301011702}
Uniformly in $\lambda\in I$ there exists the limit 
\begin{equation}
 \omega_+(\lambda,\cdot):=\lim_{t\to \infty}\omega(\lambda,t,\cdot)\colon\mathbb S^{d-1}\to\mathbb S^{d-1}
\label{eq:221217}
\end{equation}
in the $C^1$-topology.
Moreover, possibly for enlarged $R>0$ only, the map $\omega_+(\lambda,\cdot)$
 is a $C^1$-diffeomorphism on $\mathbb S^{d-1}$
depending continuously on $\lambda\in I$, and for any $(\lambda,\theta)\in I\times\mathbb S^{d-1}$
\begin{align}
\mathrm dA(\omega_+(\lambda,\theta))=(2\lambda)^{d/2}J_+(\lambda,\theta)\,\mathrm dA(\theta),
\label{eq:23010710}
\end{align} 
where $J_+$ is from Lemma~\ref{lem:230107}.
\end{lemma}
\begin{remark*}
 For fixed $\lambda\in I$ the map $\omega_+(\lambda,\cdot)$ on $\S^{d-1}$
 is called 
the \emph{asymptotic direction map} and its inverse
$\theta_+(\lambda,\cdot)=\omega_+^{-1}(\lambda,\cdot)=(\omega_+(\lambda,\cdot))^{-1}$ the \emph{inverse asymptotic direction map}.
\end{remark*}

\begin{proof} 
Take $R_0>0$ as in Corollary~\ref{cor:epsSmall}, 
and let $\sigma'\in (0,\sigma)$. 
For the moment all the estimates below are uniform in $R\ge R_0$.
We will possibly need  $R$ to be larger in Step~IV.

\smallskip
\noindent
\textit{Step I.}\ 
For any fixed $(\lambda,\theta) \in I\times\S^{d-1}$ we prove the
existence of the limit \eqref{eq:221217} in the pointwise sense. 
For each $i=1,\dots, d$ 
we compute and bound, omitting the arguments, as 
\begin{align}\label{eq:deta}
 \begin{split}
 \tfrac{\partial}{\partial t}\omega_i
&=
|y|^{-1}|\nabla S_{\mathrm L}|^{-2}(\nabla S_{\mathrm L})_i
-|y|^{-3}y_iy_j|\nabla S_{\mathrm L}|^{-2}(\nabla S_{\mathrm L})_j
\\
&=
|y|^{-1}|\nabla S_{\mathrm L}|^{-2}\bigl(\delta_{ij}-|y|^{-2}y_iy_j\bigr)
\bigl( (\nabla S_{\mathrm L})_j-(2\lambda)^{1/2}|y|^{-1}y_j\bigr)
\\
&=\vO (R^{-\sigma+\sigma'}t^{-1-\sigma'}),
 \end{split}
\end{align} 
where we have used Theorem~\ref{thm:main result2} 
and Lemma~\ref{lem:221210}. 
The integrability of \eqref{eq:deta} in $t\in (1,\infty)$
implies that there exists the limit \eqref{eq:221217} in the pointwise sense. 
In addition, $\omega_+\colon I\times \mathbb S^{d-1}\to\mathbb S^{d-1}$ is continuous 
since \eqref{eq:deta} is uniform in $(\lambda,\theta)\in I\times\mathbb S^{d-1}$.

\smallskip
\noindent
\textit{Step II.}\ 
Take any local coordinates $\theta'=(\theta'_2,\dots,\theta'_d)$ of $\mathbb S^{d-1}$, 
and set for any $\alpha=2,\dots, d$ 
\[
\Omega(\lambda,t,\theta')=(\partial_\alpha \omega_1(\lambda,t,\theta'))^2+\dots +(\partial_\alpha \omega_d(\lambda,t,\theta'))^2.
\]
We claim that for any compact subset $K$ of the
associated (open) coordinate region 
there exist $c_1,C_1>0$ such that for any $(\lambda,t,\theta')\in I\times\R_+\times K$
\begin{align}
c_1\le \Omega(\lambda,t,\theta')\le C_1.
\label{eq:2212192}
\end{align}

For that we compute and bound the $t$-derivative of $\Omega$ as in Step I. 
Since we have 
\begin{align}
(\partial_\alpha\omega_i)
=|y|^{-1}(\partial_\alpha y_i)-|y|^{-3}y_iy_j(\partial_\alpha y_j),
\label{221219}
\end{align}
we can write by using \eqref{eq:flow}
\begin{align}
\begin{split}
\tfrac{\partial}{\partial t}(\partial_\alpha\omega_i)
&=
-|y|^{-3}y_j|\nabla S_{\mathrm L}|^{-2}(\nabla S_{\mathrm L})_j(\partial_\alpha y_i)
+|y|^{-1}(\nabla|\nabla S_{\mathrm L}|^{-2})_j(\nabla S_{\mathrm L})_i(\partial_\alpha y_j)
\\&\phantom{{}={}}{}
+|y|^{-1}|\nabla S_{\mathrm L}|^{-2}(\nabla^2 S_{\mathrm L})_{ij}(\partial_\alpha y_j)
\\&\phantom{{}={}}{}
+3|y|^{-5}|\nabla S_{\mathrm L}|^{-2}(\nabla S_{\mathrm L})_ky_iy_jy_k(\partial_\alpha y_j)
\\&\phantom{{}={}}{}
-|y|^{-3}|\nabla S_{\mathrm L}|^{-2}(\nabla S_{\mathrm L})_iy_j(\partial_\alpha y_j)
-|y|^{-3}y_i|\nabla S_{\mathrm L}|^{-2}(\nabla S_{\mathrm L})_j(\partial_\alpha y_j)
\\&\phantom{{}={}}{}
-|y|^{-3}y_iy_j(\nabla|\nabla S_{\mathrm L}|^{-2})_k(\nabla S_{\mathrm L})_j(\partial_\alpha y_k)
\\&\phantom{{}={}}{}
-|y|^{-3}y_iy_j|\nabla S_{\mathrm L}|^{-2}(\nabla^2 S_{\mathrm L})_{jk}(\partial_\alpha y_k)
\\
&=:B_1+\dots+B_8
.
\end{split}
\label{eq:221219139}
\end{align}
To bound the terms on the right-hand side of \eqref{eq:221219139}
we will first prove that 
\begin{align}
|y|^{-1}(\partial_\alpha y_i)
=(\partial_\alpha\omega_i)+\mathcal O(R^{-\sigma+\sigma'}\inp{t}^{-\sigma'})\Omega^{1/2}
.
\label{eq:2212191}
\end{align}
To keep the notation simple we prefer henceforth to state errors
 like $ \mathcal O(R^{-\sigma+\sigma'}\inp{t}^{-\sigma'})$ as $ \mathcal O(R^{-\sigma+\sigma'}{t}^{-\sigma'})$. 
Now with this convention, it follows by \eqref{eq:2212092256} and Theorem~\ref{thm:main result2} that
\begin{align}
|y|=(2\lambda)^{-1/2}t+\mathcal O(R^{-\sigma+\sigma'}t^{1-\sigma'}),\quad 
(\partial_\alpha y_i)(\nabla S_{\mathrm L})_i=0,
\label{eq:221219143}
\end{align}
so that (again thanks to Theorem~\ref{thm:main result2}) we can rewrite \eqref{221219} as 
\begin{align*}
|y|^{-1}(\partial_\alpha y_i)
&=(\partial_\alpha\omega_i)
-|y|^{-2}y_i\bigl((2\lambda)^{-1/2}(\nabla S_{\mathrm L})_j-|y|^{-1}y_j\bigr)(\partial_\alpha y_j)
\\&
=(\partial_\alpha\omega_i)
+\mathcal O(R^{-\sigma+\sigma'}t^{-\sigma'})\sum^d_{j=1}\,|y|^{-1}|\partial_\alpha y_j|,
\end{align*} Hence by a summation and subtraction
\begin{equation*}
 \sum^d_{j=1}\,|y|^{-1}|\partial_\alpha y_j|\leq
 C\sum^d_{i=1}\,|\partial_\alpha\omega_i|\leq C\sqrt d\,\Omega^{1/2},
\end{equation*}
which verifies \eqref{eq:2212191}. 

Now we bound term by term
\begin{equation*}
 B_1+\dots+B_8=(B_2+B_7)+ (B_4+B_5+B_6)+B_8 +(B_1+B_3).
\end{equation*}
By \eqref{eq:2212116} and \eqref{eq:2212191} 
\begin{align*}
B_2+B_7=\mathcal O(R^{-\sigma+\sigma'}t^{-1-\sigma'})\Omega^{1/2}
\end{align*}
By \eqref{eq:221219143}, \eqref{eq:2212191} and Theorem~\ref{thm:main result2}
\begin{align*}
B_4+B_5+B_6
&=3|y|^{-4}|\nabla S_{\mathrm L}|^{-2}(\nabla S_{\mathrm L})_ky_iy_k
\bigl((2\lambda)^{-1/2}(\nabla S_{\mathrm L})_j-|y|^{-1}y_j\bigr)(\partial_\alpha y_j)
\\&\phantom{{}={}}{}
+|y|^{-2}|\nabla S_{\mathrm L}|^{-2}(\nabla S_{\mathrm L})_i
\bigl((2\lambda)^{-1/2}(\nabla S_{\mathrm L})_j-|y|^{-1}y_j\bigr)(\partial_\alpha y_j)
\\&
=\mathcal O(R^{-\sigma+\sigma'}t^{-1-\sigma'})\Omega^{1/2}.
\end{align*}
By \eqref{eq:220919537}, \eqref{eq:221219143}, \eqref{eq:2212191} and Theorem~\ref{thm:main result2}
\begin{align*}
B_8
&=
|y|^{-2}y_i\bigl((2\lambda)^{-1/2}(\nabla S_{\mathrm L})_j-|y|^{-1}y_j\bigr)|\nabla S_{\mathrm L}|^{-2}(\nabla^2 S_{\mathrm L})_{jk}(\partial_\alpha y_k)
\\&\phantom{{}={}}{}
-\tfrac12(2\lambda)^{-1/2}|y|^{-2}y_i|\nabla S_{\mathrm L}|^{-2}(\nabla |\nabla S_{\mathrm L}|^2)_k
(\partial_\alpha y_k)
\\&
=\mathcal O(R^{-\sigma+\sigma'}t^{-1-\sigma'})\Omega^{1/2}.
\end{align*}
Finally by 
\eqref{eq:221219143}, \eqref{eq:2212191}, Theorem~\ref{thm:main result2} 
and \eqref{eq:5b} 
(see also Remark \ref{rem:230331}) 
\begin{align*}
B_1+B_3
&=
-(2\lambda)^{-1/2}|y|^{-2}(\partial_\alpha y_i)
\\&\phantom{{}={}}{}
+|y|^{-2}\bigl((2\lambda)^{-1/2}(\nabla S_{\mathrm L})_j-|y|^{-1}y_j\bigr)|\nabla S_{\mathrm L}|^{-2}(\nabla S_{\mathrm L})_j(\partial_\alpha y_i)
\\&\phantom{{}={}}{}
+(2\lambda)^{1/2}f^{-1}|y|^{-1}|\nabla S_{\mathrm L}|^{-2}\bigl(
(\partial_if)(\partial_jf)+f(\nabla^2 f)_{ij}\bigr)(\partial_\alpha y_j)
\\&
=\mathcal O(R^{-\sigma+\sigma'}t^{-1-\sigma'})\Omega^{1/2}.
\end{align*}
Therefore, combining the above estimates, we obtain 
\begin{align}
\begin{split}
\tfrac{\partial}{\partial t}(\partial_\alpha\omega_i)
&=\mathcal O(R^{-\sigma+\sigma'}t^{-1-\sigma'})\Omega^{1/2},
\ \ \text{or}\ \ 
\tfrac{\partial}{\partial t}\Omega
=
\mathcal O(R^{-\sigma+\sigma'}t^{-1-\sigma'})\Omega
.
\end{split}
\label{eq:221219215}
\end{align}
This implies $\ln \Omega$ is bounded, and thus the claim \eqref{eq:2212192} is verified.

\smallskip
\noindent
\textit{Step III.}\ 
Next we show $\omega_+(\lambda,\cdot)\colon\mathbb S^{d-1}\to\mathbb S^{d-1}$ 
is continuously differentiable. 
This is straightforward due to Step II. 
By \eqref{eq:221219215} and \eqref{eq:2212192} it follows that 
\begin{align}
\tfrac{\partial}{\partial t}(\partial_\alpha\omega_i)
=\mathcal O(R^{-\sigma+\sigma'}t^{-1-\sigma'})
\label{eq:230401}
\end{align}
uniformly in $\lambda\in I$ and locally uniformly in $\theta'$
in the coordinate region. 
It is integrable in $t\in (1,\infty)$, and this implies there exists the limit
\begin{align*}
\lim_{t\to\infty}\partial_\alpha\omega_i(\lambda,t,\theta')
\end{align*}
uniformly in $\lambda\in I$ and locally uniformly in $\theta'$ in this region. 
Hence the convergence \eqref{eq:221217} is in the $C^1$-topology, 
and $\omega_+(\lambda,\cdot)$ is continuously differentiable as wanted. 
Note that $\partial_\alpha\omega_+$ is continuous also in $\lambda\in I$ since 
these estimates are uniform in $\lambda\in I$.

\smallskip
\noindent
\textit{Step IV.}\ 
We prove that $\omega_+(\lambda,\cdot)$ is a $C^1$-diffeomorphism,
provided $R\ge R_0$ is large. 
By \eqref{eq:deta}, \eqref{eq:230401} and the fact that 
$\omega(\lambda,t,\theta)=\theta$ for any $t\in(0,(2\lambda)^{1/2}R]$, see \eqref{eq:22120923},
 it follows that for any sufficiently large $R\ge R_0$ and any
 $(\lambda,\theta)\in I\times \mathbb S^{d-1}$ 
\begin{equation}
 \label{eq:fixedP3}
 \abs{{\omega_+(\lambda,\theta)-\theta}}\le \tfrac12\quad \mand \quad 
\abs{\nabla_\theta(\omega_+(\lambda,\theta)-\theta)}\leq\tfrac12.
\end{equation}

Due to the second bound of \eqref{eq:fixedP3} we can apply the inverse function
theorem to $\omega_+(\lambda,\cdot)$, and hence the image $\omega_+(\lambda, \mathbb S^{d-1})
\subseteq \mathbb S^{d-1}$ is open. 
On the other hand, since $\mathbb S^{d-1}$ is compact and $\omega_+$ is continuous, 
$\omega_+(\lambda, \mathbb S^{d-1})$ is also closed. 
Thus $\omega_+(\lambda,\cdot)$ is surjective due to the connectedness of $\mathbb S^{d-1}$. 

Now it suffices to show the injectivity of $\omega_+(\lambda,\cdot)$.
Suppose that for some $\lambda\in I$ and $\theta_1,\theta_2\in \S^{d-1}$
 \begin{equation*}
 \omega_+(\lambda,\theta_1)=\omega_+(\lambda,\theta_2).
 \end{equation*} 
Then, if we let $\gamma$ be a grand circle segment, or a geodesic of minimal length, on $\mathbb S^{d-1}$ 
connecting $\theta_1$ and $\theta_2$, we can estimate by using \eqref{eq:fixedP3}
 \begin{align*}
 \mathop{\mathrm{dist}}\nolimits_{\mathbb R^d}(\theta_1,\theta_2)
&=\mathop{\mathrm{dist}}\nolimits_{\mathbb R^d}\bigl(\theta_1-\omega_+(\lambda,\theta_1),\theta_2-\omega_+(\lambda,\theta_2)\bigr)
\\&\leq\int_0^1\,
 \abs[\big]{\nabla_\theta(\theta-\omega_+(\lambda,\theta))_{|\theta=\gamma(t)}}\,\abs{\dot\gamma(t)}\,\d t
\\&\leq
 \tfrac12 \mathop{\mathrm{dist}}\nolimits_{\S^{d-1}}(\theta_1,\theta_2)
\leq \tfrac \pi 4\mathop{\mathrm{dist}}\nolimits_{\mathbb R^d}(\theta_1,\theta_2),
 \end{align*} 
where $\mathop{\mathrm{dist}}\nolimits_{\mathbb R^d}$ and 
$\mathop{\mathrm{dist}}\nolimits_{\mathbb S^{d-1}}$ are the standard metrics there. 
This implies $\theta_1=\theta_2$ since $\pi<4$. 
Thus we conclude $\omega_+(\lambda,\cdot)\colon \S^{d-1}\to\S^{d-1}$ is a $C^1$--diffeomorphism.

\smallskip
\noindent
\textit{Step V.}\ 
Finally we show the identity \eqref{eq:23010710} by 
rewriting the Euclidean volume measure as follows. 
In the standard spherical coordinates $(r,\omega')$ we can write 
\[
\mathrm dx(r,\omega')
=r^{d-1}\,\mathrm dr\mathrm dA(\omega(\omega'))
=r^{d-1}a(\omega')\,\mathrm dr\mathrm d\omega'_2\cdots\mathrm d\omega'_d.
\]
Let us further change the variables to $(\lambda,t,\theta')$ by using \eqref{eq:23010712}.
We compute the Jacobian by the cofactor expansion. 
Some of the first order derivatives of \eqref{eq:23010712} are bounded by \eqref{eq:deta} and \eqref{eq:2212192},
and we only have to note that by Theorem~\ref{thm:main result2}, \eqref{eq:2212191} and \eqref{eq:2212192} 
\begin{align*}
 \tfrac{\partial}{\partial t}r
&=|y|^{-1}y\cdot (\nabla S_{\mathrm L})/|\nabla S_L|^2
= (2\lambda)^{-1/2}+\mathcal O(t^{-\sigma'}),\\
\partial_ \alpha r
&=|y|^{-1}y \cdot \partial_\alpha y= \parb{|y|^{-1}y -(2\lambda)^{-1/2}\nabla S_{\mathrm L}}\cdot \partial_\alpha
 y= \mathcal O(t^{1-\sigma'}).
\end{align*}
 Then the cofactor expansion yields 
\begin{align*}
&\mathrm dx(r(\lambda,t,\theta'),\omega'(\lambda,t,\theta'))
\\&
=t^{d-1}\bigl((2\lambda)^{-d/2}K'(\lambda,t,\theta')+\mathcal O(t^{-\sigma'})\bigr)a(\omega'(\lambda,t,\theta'))\,\mathrm dt\mathrm d\theta'_2\cdots\mathrm d\theta'_d ,
\end{align*}
where $K'(\lambda,t,\cdot)$ is the Jacobian of $\omega(\lambda,t,\cdot)\colon\mathbb S^{d-1}\to\mathbb S^{d-1}$ 
in the local coordinates $\theta'$ and $\omega'$. 
The above expression has to coincide with \eqref{eq:23010713}, so that 
\begin{align*}
t^{d-1}\bigl((2\lambda)^{-d/2}K'(\lambda,t,\theta')+\mathcal O(t^{-\sigma'})\bigr)a(\omega'(\lambda,t,\theta'))
\,\mathrm d\theta'_2\cdots\mathrm d\theta'_d
=J(\lambda,t,\theta(\theta'))\,\mathrm dA(\theta(\theta')),
\end{align*}
or 
\begin{align*}
\mathrm dA(\omega(\omega'(\lambda,t,\theta')))
&=K'(\lambda,t,\theta')a(\omega'(\lambda,t,\theta'))\,\mathrm d\theta'_2\cdots\mathrm d\theta'_d
\\&=\bigl((2\lambda)^{d/2}t^{-(d-1)}J(\lambda,t,\theta(\theta'))+\mathcal O(t^{-\sigma})\bigr)\,\mathrm dA(\theta(\theta'))
.
\end{align*}
Hence by letting $t\to\infty$ we obtain \eqref{eq:23010710}.
We are done.
\end{proof}

\subsubsection{Comparison of eikonal distances}
\label{subsubsec:Generalized Fourier transform for rho>1/2 II}

Here we compare $S_{\mathrm L}$ (solving the eikonal equation for
$V_{\mathrm L}$) with the function $S$ from
Theorem~\ref{thm:comp-gener-four} (solving the eikonal equation for
$V$, with $\lambda$ in the same interval $I$). 
We show that their difference has a radial limit at
infinity (as a  step of the proof  we first establish the eikonal radial limit). 
We present a slightly generalized assertion localized to a conic subset of $\mathbb R^d$. 
For any $R'>0$ and any open subset $U\subseteq \S^{d-1}$ we set 
\begin{equation*}
 \Gamma_{R',U}=\bigl\{x\in \R^d\,\big|\, \abs{x}> R'\mand \hat x \in U\bigr\},\quad \hat x=|x|^{-1}x.
\end{equation*}

\begin{lemma}\label{lemma:uniq-asympt-norm}
Let $R>0$ be large enough as in Lemma~\ref{lem:2301011702}. 
Let $R'>0$, and $U\subseteq \S^{d-1}$ be open, 
and assume $S=\sqrt{2\lambda}|x|(1+s)\in C(I;C^2(\Gamma_{R',U}))$ satisfies: 
\begin{enumerate}[(i)]
\item\label{item:230617}
For each $\lambda\in I$, $S(\lambda,\cdot)$ solves \eqref{eq:4c} on $\Gamma_{R',U}$. 
\item\label{item:230618}
For any compact subset $I'\subseteq I$ there exist $\epsilon,C>0$ 
such that for any $|\alpha|\le 2$ and $(\lambda,x)\in I'\times \Gamma_{R',U}$ 
\begin{align*}
\left| \partial_x^\alpha s(\lambda,x)\right|
&\le C\langle x\rangle^{-\epsilon-|\alpha|}.
\end{align*}
\end{enumerate}
Then the following assertions hold.
\begin{enumerate}[(1)]
\item \label{item:sigma1}
There exists the limit
\begin{equation*}
\Sigma_{\mathrm L}(\lambda,\theta):=
\lim_{t\to \infty}\parb{ S(\lambda,y(\lambda,t,\theta))-S_{\mathrm L}(\lambda,y(\lambda,t,\theta))}
\end{equation*}
taken locally uniformly in $(\lambda,\theta)$ with $\lambda\in I$ and $\theta\in \omega_+^{-1}(\lambda,U)=\theta_+(\lambda,U)$.
In particular, if $S_1\in C(I;C^2(\Gamma_{R',U}))$ also satisfies 
the above conditions \ref{item:230617} and \ref{item:230618}, then there exists the limit
\[\Sigma(\lambda,\theta):=\lim_{t\to \infty}\parb{S_1(\lambda,y(\lambda,t,\theta))-S(\lambda,y(\lambda,t,\theta))}\]
taken locally uniformly in $(\lambda,\theta)$ with $\lambda\in I$ and
$\theta\in \theta_+(\lambda,U)$.
\item \label{item:sigma2} The quantities in \ref{item:sigma1} can {also} be computed
 as the limits
\begin{align*}
\Theta_{\mathrm L}(\lambda,\omega)&:=\Sigma_{\mathrm L}(\lambda,\theta_+(\lambda,\omega))=
\lim_{r\to \infty}\parb{ S(\lambda,r\omega)-S_{\mathrm
 L}(\lambda,r\omega)},\\
\Theta(\lambda,\omega)&:=\Sigma(\lambda,\theta_+(\lambda,\omega))=
\lim_{r\to \infty}\parb{ S_1(\lambda,r\omega)-S(\lambda,r\omega)},
\end{align*} both taken locally uniformly in $(\lambda,\omega)\in I\times U$.
\end{enumerate}
\end{lemma} 

\begin{proof}
\textit{Step I.}\ 
To prove the assertions of \ref{item:sigma1} we only show the first one.
The
second assertion is obvious from the first one.
The following bounds are locally uniform 
in $(\lambda,\theta)$ with $\lambda\in I$ and $\theta\in \theta_+(\lambda,U)$, however we shall not elaborate on that
feature.
By \eqref{eq:flow}, \eqref{eq:4c} and \eqref{eq:2212116} 
we compute for large $t> 0$ 
 \begin{align}
\begin{split}
 \tfrac{\d}{\d t}(S-S_{\mathrm L})
&
=|\nabla S_{\mathrm L}|^{-2}(\nabla S-\nabla S_{\mathrm L})\cdot (\nabla S_{\mathrm L})\\
&
=
-\tfrac12|\nabla S_{\mathrm L}|^{-2}
(\nabla S-\nabla S_{\mathrm L})^2
-|\nabla S_{\mathrm L}|^{-2}V_{\mathrm S}.
\end{split}
\label{eq:230210}
 \end{align} 
We are going to show the integrability of \eqref{eq:230210} at infinity. 
The second term on the right-hand side of \eqref{eq:230210} is clearly integrable
due to Theorem~\ref{thm:main result2} and \eqref{eq:2212092256},
and we discuss only the first term. 
It suffices to show that for some small $\delta>0$
\begin{align}
u:=(\nabla S-\nabla S_{\mathrm L})^2=\mathcal O(t^{-1-\delta}),
\label{eq:230715}
\end{align}
see also \cite[Theorem 2.3]{HS} and its proof. 
Similarly to \eqref{eq:230210} we can compute the derivative of $u$ as 
\begin{align}
\begin{split}
\tfrac12|\nabla S_{\mathrm L}|^{2}\tfrac{\mathrm d}{\mathrm dt}u
&=
(\nabla S-\nabla S_{\mathrm L})\cdot
(\nabla^2 S-\nabla^2 S_{\mathrm L})(\nabla S_{\mathrm L})
\\&
=
-(\nabla S-\nabla S_{\mathrm L})\cdot(\nabla^2 S)(\nabla S-\nabla S_{\mathrm L})
-(\nabla V_{\mathrm S})\cdot(\nabla S-\nabla S_{\mathrm L})
.
\end{split}
\label{eq:23021010}
\end{align}
From the assumptions on $S$ we deduce that 
\begin{align*}
 \nabla^2S
\ge 
2\lambda S^{-1}I
-C_2 S^{-1-\epsilon}I
-S^{-1}(\nabla S)\otimes(\nabla S)
.
\end{align*}
We apply this bound to \eqref{eq:23021010}, use \eqref{eq:4c} and
\eqref{eq:2212116} (as in \eqref{eq:230210}), and conclude that 
\begin{align}
\begin{split}
\tfrac12|\nabla S_{\mathrm L}|^{2}\tfrac{\mathrm d}{\mathrm dt}u
&
\le 
-2\lambda S^{-1}u
+C_2 S^{-1-\epsilon}u
\\&\phantom{{}={}}{}\quad \quad
+S^{-1}\bigl((\nabla S-\nabla S_{\mathrm L})\cdot(\nabla S)\bigr)^2+|\nabla V_{\mathrm S}|u^{1/2}
\\&
=-2\lambda S^{-1}u
+C_2 S^{-1-\epsilon}u
+ 4^{-1}S^{-1}(u-2V_{\mathrm S})^2
+|\nabla V_{\mathrm S}|u^{1/2}
.
\end{split}
\label{eq:23021012}
\end{align}
By Theorem~\ref{thm:main result2}, \eqref{eq:2212092256} and the assumptions on $S$ 
we observe , letting $\delta\in (0,\epsilon)$ be small enough, that 
\[
|\nabla S_{\mathrm L}|^{2}=2\lambda+\mathcal O(t^{-\delta}),\quad 
S=t+\mathcal O(t^{1-\delta}),\quad 
u=(\nabla S-\nabla S_{\mathrm L})^2=\mathcal O(t^{-\delta})
.
\]
Substituting these estimates into \eqref{eq:23021012}, 
we obtain for large $t> 0$
\begin{align*}
\tfrac{\mathrm d}{\mathrm dt}u
&
\le 
-2t^{-1}u
+C_3 t^{-1-\delta}u
+C_3t^{-2-\delta}
\le 
-(2-\delta)t^{-1}u
+C_3t^{-2-\delta}
.
\end{align*}
This differential inequality implies \eqref{eq:230715}. Hence we are
done with \ref{item:sigma1}.

\smallskip
\noindent
\textit{Step II.}\ 
We prove the first assertion of \ref{item:sigma2}
(the other one in \ref{item:sigma2} follows from that). 
Fixing  any compact subset $K\subset I\times U$  we are going to prove 
\begin{align*}
\lim_{r\to\infty}\sup_{(\lambda,\omega)\in K}\bigl|\Sigma_{\mathrm L}(\lambda,\theta_+(\lambda,\omega))-S(\lambda,r\omega)+S_{\mathrm L}(\lambda,r\omega)\bigr|
=0
.
\end{align*}

Using   \eqref{eq:23010712} and Lemma
  \ref{lem:2301011702}, let us first
  note  that 
\begin{align}
\lim_{r\to\infty}\sup_{(\lambda,\omega)\in K}\bigl|\theta_+(\lambda,\omega)-\theta(\lambda,r\omega)\bigr|
=0
.
\label{eq:240722}
\end{align}
In fact,  changing  variables
from $(r,\omega)$ to $(t,\theta)$ and using  $t=S_{\mathrm L}(\lambda,r\omega)$, it follows
that  uniformly in $(\lambda,\omega)\in K$ the quantity $\tfrac {\d t}{\d r} = \nabla S_L\cdot \omega>
\sqrt \lambda$ for large $r$, 
and hence taking $r\to \infty$ corresponds to taking $t\to \infty$.
Thus for some compact subset $K'\subset \{(\lambda,\theta);\ \lambda\in I,\ \theta\in \theta_+(\lambda,U)\}$
we can compute (thanks to Lemma \ref{lem:2301011702})
\begin{align*}
\lim_{r\to\infty}\sup_{(\lambda,\omega)\in K}\bigl|\theta_+(\lambda,\omega)-\theta(\lambda,r\omega)\bigr|
&
\le 
\lim_{t\to\infty}\sup_{(\lambda,\theta)\in K'}\bigl|\theta_+(\lambda,\omega(\lambda,t,\theta))-\theta(\lambda,y(\lambda,t,\theta))\bigr|
\\&
=
\sup_{(\lambda,\theta)\in K'}\bigl|\theta_+(\lambda,\omega_+(\lambda,\theta))-\theta\bigr|=0.
\end{align*}

Now by \eqref{eq:240722}, the above change of variables and the assertion (1) it follows that 
\begin{align*}
&\lim_{r\to\infty}\sup_{(\lambda,\omega)\in K}\bigl|\Sigma_{\mathrm L}(\lambda,\theta_+(\lambda,\omega))-S(\lambda,r\omega)+S_{\mathrm L}(\lambda,r\omega)\bigr|
\\&
=
\lim_{r\to\infty}\sup_{(\lambda,\omega)\in K}\bigl|\Sigma_{\mathrm L}(\lambda,\theta(\lambda,r\omega))-S(\lambda,r\omega)+S_{\mathrm L}(\lambda,r\omega)\bigr|
\\&
=
\lim_{t\to\infty}\sup_{(\lambda,\theta)\in K'}\bigl|\Sigma_{\mathrm L}(\lambda,\theta(\lambda,y(\lambda,t,\theta)))-S(\lambda,y(\lambda,t,\theta))+S_{\mathrm L}(\lambda,y(\lambda,t,\theta))\bigr|
\\&
=
\lim_{t\to\infty}\sup_{(\lambda,\theta)\in K'}\bigl|\Sigma_{\mathrm L}(\lambda,\theta)-S(\lambda,y(\lambda,t,\theta))+S_{\mathrm L}(\lambda,y(\lambda,t,\theta))\bigr|
\\&=0.
\end{align*}
We are done.
\end{proof}

\begin{remark}\label{remark:inverse-asympt-norm}
Such a localized version has an application in $3$-body long-range stationary scattering theory \cite{Sk1},
for which we should take $U\subseteq \S^{d-1}$ such that the closure $\overline U$ does not intersect the `collision planes'.
For such $U$ the function
\begin{equation}
\e^{\i \Theta(\lambda,\omega)}=\e^{\i \Sigma(\lambda,\theta_+(\lambda,\omega))};\quad 
\theta_+(\lambda,\cdot)=\omega_+^{-1}(\lambda,\cdot), 
\label{eq:23061811}
\end{equation} 
induces a well-defined family of unitary multiplication operators on
$L^2(U)(\subseteq \vG)$ being strongly continuous in $\lambda$. 
Upon varying $U$ under the above constraint {the function \eqref{eq:23061811} 
is defined} almost everywhere on $\mathbb S^{d-1}$
and constitutes a strongly continuous $\vL\parb{\vG }$--valued function of $\lambda$.
The transformation
 factor $\e^{\i \Theta(\lambda,\omega)}$ (exhibiting `covariance') and \eqref {eq:221206} are
 applicably to  the
 $3$-body problem \cite{Sk1}, cf. Subsection \ref{subsubsec:three-body-problem}. In particular it is not possible to
 take $U=\mathbb S^{d-1}$ in that application. However in the
 present paper we only use Lemma \ref{lemma:uniq-asympt-norm} with
 $U=\mathbb S^{d-1}$, see for example Remark \ref{rem:freeS} \ref{item:S3}
 and Theorem~\ref{thm:221207}
(\ref{item:230205}).
\end{remark}

\subsection{Stationary wave operators}
\label{subsubsec:Generalized Fourier transform for rho=1}

\subsubsection{Construction for the regularized potential}
\label{subsubsec:Diagonalization}

Here we discuss an analogue of 
Theorem~\ref{thm:comp-gener-four} for $H_{\mathrm L}$ in the
spherical eikonal coordinates.
Once the strong radiation condition bounds from Theorem~\ref{thm:proof-strong-bound} 
are established 
and the spherical eikonal coordinates are fixed,
the construction is rather straightforward,
following the schemes of \cite{HS, Sk}. 
{Set for any $\xi\in\mathcal G$
\begin{align}
\phi_\pm^{S_{\mathrm L}}[\xi](\lambda,x)
&=
\tfrac{(2\pi)^{1/2}}{(2\lambda)^{1/4}}\chi(r)
r^{-(d-1)/2}\e^{\pm\i S_{\mathrm L}(\lambda,x)}\xi(\hat x)
,\quad r=|x|,\,\,\hat x=|x|^{-1}x,
\label{eq:230126bb}
\end{align}
respectively, where $\chi$ is from \eqref{eq:14.1.7.23.24} (see also \eqref{eq:230626}).}

\begin{proposition}\label{prop:22121414}
\begin{enumerate}
\item\label{item:221230}For any $\lambda\in I$ 
there exist unique $E^\pm(\lambda)\in\vL(\vB,\vG)$ such that for any
$\psi\in \mathcal B$ 
{
\begin{align}\label{eq:limFb}
R_{\mathrm L}(\lambda\pm\mathrm i0)\psi-\phi^{S_{\mathrm L}}_\pm[E^\pm(\lambda)\psi](\lambda,\cdot)
\in\mathcal B^*_0,
\end{align}
respectively. 
}

\item
The mappings $E^\pm\colon I\times \mathcal B\to \mathcal G$
are continuous.
\item
For any $\lambda\in I$ 
\begin{equation*}
E^\pm(\lambda)^*E^\pm(\lambda)
=\delta(H_{\mathrm L}-\lambda),
\end{equation*} 
respectively.
\end{enumerate}
\end{proposition}

Before proving Proposition~\ref{prop:22121414} 
we present a trace type theorem in a form appropriate for our application. 
Note that by Fubini's theorem we can identify 
\[L^2_{\mathrm{loc}}(\R_+\times \mathbb S^{d-1})
\simeq L^2_{\mathrm{loc}}(\R_+;\mathcal G).\]
To be precise 
we denote the above identification operator for the moment by $\iota$,
i.e.\ for any $\psi \in L^2_{\mathrm{loc}}(\R_+\times \mathbb S^{d-1})$
we let 
\[
\iota(\psi)(t)=\psi(t,\cdot)\in\mathcal G \ \ \text{for a.e.\ }t\in\R_+.\]

\begin{lemma}\label{lem:22122915}
Let $k\in\mathbb N_0$ and $\psi \in H^s_{\mathrm{loc}}(\R_+\times \mathbb S^{d-1})$ with $s>k+1/2$.
Then
\begin{align*}
\iota(\psi) \in C^k(\R_+;\mathcal G)
,\ \ \text{and}\ \ 
\tfrac{\mathrm d^l}{\mathrm dt^l}\iota(\psi) 
=\iota\bigl(\tfrac{\partial^l}{\partial t^l}\psi\bigr)
\ \text{for}\ l=0,\dots,k
.
\end{align*}
\end{lemma}

\begin{proof}[Proof of Lemma~\ref{lem:22122915}]
By a partition-of-unity argument we can reduce the
 claims to similar ones in 
a coordinate region. 
Then we can mimic the proof of the familiar Sobolev embedding theorem. 
We omit the details. 
\end{proof}


\begin{proof}[Proof of Proposition~\ref{prop:22121414}]
\textit{Step I.}\ 
Let 
$\lambda\in I$ and $\psi\in C^\infty_{\mathrm c}(\mathbb R^d)$ be
given, and then let 
\begin{equation*}
\Psi(t)=
J(\lambda,t,\cdot)^{1/2}\e^{\mp\i S_{\mathrm L}(\lambda,y(\lambda,t,\cdot))}\bigl(R_{\mathrm L}(\lambda\pm\i 0)\psi\bigr)(y(\lambda,t,\cdot))
\in\mathcal G;\quad t\in \R_+.
\end{equation*} Since $R_{\mathrm L}(\lambda\pm\mathrm
i0)\psi\in H^2_{\mathrm{loc}}(\mathbb R^d)$ it follows from 
Lemma~\ref{lem:22122915}, that $\Psi\in C^1(\R_+;\mathcal G)$.
We first show the existence of the limits
\begin{equation}\label{eq:limF}
D^\pm(\lambda)\psi=
\vGlim_{t\to \infty}
 \Psi(t). 
\end{equation} 
By the fundamental theorem of calculus we have 
\[
\Psi(t)=\Psi(1)+\int_1^t\tfrac{\mathrm d}{\mathrm dt}\Psi(\tau)\,\mathrm d\tau,
\]
and it suffices to show that the last integrand is integrable as a $\mathcal G$-valued function.
We can compute it 
by Lemma~\ref{lem:22122915}, 
\eqref{eq:2212115b}, \eqref{eq:2212092256}, \eqref{eq:flow} and \eqref{eq:6} as 
\begin{align*}
\tfrac{\mathrm d}{\mathrm dt}\Psi
&=
 J^{1/2}
\e^{\mp\i S_{\mathrm L}}\bigl(\mathrm i|\nabla S_{\mathrm L}|^{-2}(\nabla S_{\mathrm L})\cdot\gamma 
+\tfrac12(\nabla\cdot |\nabla S_{\mathrm L}|^{-2}\nabla S_{\mathrm L})
\bigr)R_{\mathrm L}(\lambda\pm\i 0)\psi
\\&
=
J^{1/2}\e^{\mp\i S_{\mathrm L}}
\bigl(\mathrm i|\nabla S_{\mathrm L}|^{-2}\gamma_\|
+
\tfrac12(\nabla |\nabla S_{\mathrm L}|^{-2})\cdot(\nabla S_{\mathrm L})\bigr)R_{\mathrm L}(\lambda\pm\i 0)\psi
. 
\end{align*}
By \eqref{eq:gamma1a} and \eqref{eq:220919537} we can find
$\delta>0$ and 
$\Phi\in L^2_{(1+\delta)/2}$ such that 
 \begin{equation*}
\tfrac{\mathrm d}{\mathrm dt}\Psi(t)
=
J(\lambda,t,\cdot)^{1/2}\Phi(\lambda,y(\lambda,t,\cdot))
.
\end{equation*}
Then by \cs and \eqref{eq:co_area}
\begin{align*}
\int_1^\infty\bigl\|\tfrac{\mathrm d}{\mathrm dt}\Psi(t)\bigr\|_{\mathcal G}\,\mathrm dt
&=
\int_1^\infty\!\mathrm dt
\biggl(\int_{\mathbb
 S^{d-1}}|\Phi(\lambda,y(\lambda,t,\theta))|^2\,\mathrm d
 A_{\lambda,t}(\theta)\biggr)^{1/2} 
\\&
\le 
C_1\biggl(\int_1^\infty t^{-1-\delta}\,\mathrm dt\biggr)^{1/2}
\\&\phantom{{}={}}{}
\cdot 
\biggl(\int_1^\infty\!\mathrm dt
\int_{\mathbb S^{d-1}}|(|\cdot|^{(1+\delta)/2}\Phi)(\lambda,y(\lambda,t,\theta))|^2\,\mathrm d A_{\lambda,t}(\theta)\biggr)^{1/2}
\\[.3em]&
\le 
C_2\|\Phi\|_{L^2_{(1+\delta)/2}}.
\end{align*}
Hence there exist the limits \eqref{eq:limF}. 
We note {$D^\pm(\lambda)\psi$} are continuous in $\lambda\in I$ 
since $\Psi(t)$ is continuous in $\lambda\in I$, and the above estimates are 
locally uniform in this variable. 

\smallskip
\noindent
{\textit{Step II.}\ 
Next we set for any $\lambda\in I$ and $\psi\in C^\infty_{\mathrm c}(\mathbb R^d)$
\begin{subequations}
 \begin{align}
 \begin{split}
E^\pm(\lambda)\psi&=
c(\lambda)(2\lambda)^{-d/4}
J_+(\lambda,\theta_+(\lambda,\cdot))^{-1/2}\bigl(D^\pm(\lambda)\psi\bigr)(\theta_+(\lambda,\cdot)),\\
c(\lambda)&=(2\pi)^{-1/2}(2\lambda)^{1/2}, 
 \end{split}
\label{24050220} 
\end{align}
 and verify that they satisfy
\eqref{eq:limFb}. For completeness of presentation note that 
\begin{equation}\label{eq:normIso}
 \bigl\|E^\pm(\lambda)\psi \bigr\|_{\mathcal G}=c(\lambda)\bigl\|D^\pm(\lambda)\psi \bigr\|_{\mathcal G}, 
\end{equation}
cf.\ Lemma~\ref{lem:2301011702}. 
\end{subequations}
By \eqref{eq:limF} it follows that 
\begin{align*}
\begin{split}
\lim_{t\to\infty}t^{-1}\int_0^t 
\bigl\|D^\pm(\lambda)\psi
-{}&J(\lambda,\tau,\cdot)^{1/2}
\e^{\mp\i S_{\mathrm L}(\lambda,y(\lambda,\tau,\cdot))}
(R_{\mathrm L}(\lambda\pm\i 0)\psi)(y(\lambda,\tau,\cdot))\bigr\|_{\mathcal G}^2\,\mathrm d\tau=0
,
\end{split}
\end{align*} 
and along with \eqref{eq:co_area},
 Lemma~\ref{lem:230107} and the asymptotics 
 $|y(\lambda,\tau,\cdot)|/\tau\to (2\lambda)^{-1/2}$ this implies 
\begin{align*}
\lim_{t\to\infty}t^{-1}\int_{\set{S_L\le t}}
\bigl|
{}&
(2\lambda)^{-(d-1)/4}|x|^{-(d-1)/2}\e^{\pm\i S_{\mathrm L}(\lambda,x)}
J_+(\lambda,\theta(\lambda,
 x))^{-1/2}(D^\pm(\lambda)\psi)(\theta(\lambda, x))
\\&{}
-(R_{\mathrm L}(\lambda\pm\i 0)\psi)(x)\bigr|^2\,\mathrm dx=0
.
\end{align*}
Hence it suffices to prove 
\begin{align*}
\lim_{t\to\infty}t^{-1}\int_{\set{S_L\le t}}
|x|^{-(d-1)/2}\bigl|
{}&
J_+(\lambda,\theta(\lambda,x))^{-1/2}
(D^\pm(\lambda)\psi)(\theta(\lambda,x))
\\&{}
-
J_+(\lambda,\theta_+(\lambda,\hat x))^{-1/2}\bigl(D^\pm(\lambda)\psi\bigr)(\theta_+(\lambda,\hat
 x))
\bigr|^2\,\mathrm dx=0
.
\end{align*}
In turn, if we let
$u(\lambda,\theta)=J_+(\lambda,\theta)^{-1/2}(D^\pm(\lambda)\psi)(\theta)$ 
 and again use eikonal spherical coordinates, it suffices to prove
 that 
\begin{align}
\lim_{t\to\infty}t^{-1}
\int_0^t
\bigl\|
u(\lambda,\cdot)-u\bigl(\lambda,\theta_+(\lambda,\omega(\lambda,\tau,\cdot))\bigr)
\bigr\|_{\mathcal G}^2\,\mathrm d\tau
=0
.
\label{24050219}
\end{align}
To prove \eqref{24050219}, first note that for any $v,w\in \mathcal G$ 
\begin{equation}
\bigl\|v\bigl(\theta_+(\lambda,\omega(\lambda,\tau,\cdot))\bigr)
-w\bigl(\theta_+(\lambda,\omega(\lambda,\tau,\cdot))\bigr)\bigr\|_{\mathcal G}
\le C_3\|v-w\|_{\mathcal G};\quad \tau\in [1,\infty)
.
\label{2405021}
\end{equation} 
Here we used that the coordinate change $\theta\to \theta_+(\lambda,\omega(\lambda,\tau,\cdot))$ converges 
to the identity map in the $C^1$-topology as $\tau\to\infty$. Next we
 estimate 
 for any $v\in C^\infty(\mathbb S^{d-1})$ 
\begin{align}
\begin{split}
&\bigl\|u(\lambda,\cdot)
-u\bigl(\lambda,\theta_+(\lambda,\omega(\lambda,\tau,\cdot)\bigr)\bigr\|_{\mathcal G}
\\&
\le 
\|u(\lambda,\cdot)-v\|_{\mathcal G}
+\bigl\|
v-v\bigl(\theta_+(\lambda,\omega(\lambda,\tau,\cdot)\bigr)\bigr\|_{\mathcal G}
\\&\phantom{{}={}}
+\bigl\|
v\bigl(\theta_+(\lambda,\omega(\lambda,\tau,\cdot))\bigr)
-u\bigl(\lambda,\theta_+(\lambda,\omega(\lambda,\tau,\cdot))\bigr)\bigr\|_{\mathcal G}
.
\end{split}
\label{2405022}
\end{align}
The first term on the right-hand side 
of \eqref{2405022} can be arbitrarily small by choosing 
appropriate $v\in C^\infty(\mathbb S^{d-1})$. Due to \eqref{2405021}
then also the third term is small (uniformly in $\tau$). 
 For any such $v$ fixed, clearly the second term converges to $0$ as $\tau\to\infty$. 
This verifies \eqref{24050219}.}

\smallskip
\noindent
{\textit{Step III.}}\ 
We next show that for any $\lambda\in I$ and $\psi\in C_{\mathrm c}^\infty(\R^d)$
\begin{align}
\|E^\pm(\lambda)\psi\|_{\mathcal G}^2=\langle \psi,\delta(H_{\mathrm L}-\lambda)\psi\rangle. 
\label{eq:2212301227}
\end{align}
Using for $T>0$ the function $\chi_T$ from \eqref{eq:14.1.7.23.24b}, we set 
\begin{equation}
\eta_T=1-\chi_T,\quad 
\eta_T'=-T^{-1}\chi'(|\cdot|/T). 
\label{eq:230211}
\end{equation}
Introducing also the notation $\phi=R_{\mathrm L}(\lambda\pm{\mathrm
 i}0)\psi$, we then write 
\begin{align*}
2\pi\langle \psi, \delta(H_{\mathrm L}-\lambda)\psi\rangle 
&= 
\pm2\mathop{\mathrm{Im}}\langle(H_{\mathrm L}-\lambda)\phi, \phi\rangle 
= 
\pm2\lim_{T\to\infty} \mathop{\mathrm{Im}}\langle(H_{\mathrm L}-\lambda)\phi, \eta_T\phi\rangle 
.
\end{align*}
By an integration by parts this leads to 
\begin{align}
\begin{split}
2\pi\langle \psi, \delta(H_{\mathrm L}-\lambda)\psi\rangle 
&= 
\mp\lim_{T\to\infty}\mathop{\mathrm{Re}}\langle \hat x\cdot p\phi, \eta_T'\phi\rangle
\\&= 
\mp\lim_{T\to\infty}\mathop{\mathrm{Re}}\langle \hat x\cdot \gamma \phi, \eta_T'\phi\rangle
-\lim_{T\to\infty}\langle |\nabla \chi_1S_{\mathrm L}|\phi, \eta_T'\phi\rangle
.
\end{split}
\label{eq:22123012}
\end{align}
The contribution from the first term on the right-hand side of \eqref{eq:22123012}
vanishes due to \eqref{eq:gamma2}. As for the second term 
we rewrite the integral in the standard spherical coordinates, 
substitute \eqref{eq:limFb} and conclude that 
\[
-\lim_{T\to\infty}\langle |\nabla \chi_1 S_{\mathrm L}|\phi, \eta_T'\phi\rangle
=2\pi\|E^\pm(\lambda)\psi\|_{\mathcal G}^2, 
\] 
hence the claim \eqref{eq:2212301227} for $\psi\in C_{\mathrm c}^\infty(\R^d)$. 

\smallskip
\noindent
{\textit{Step IV.}}\ 
Now we prove the assertions (1)--(3). 
The identity \eqref{eq:2212301227} immediately implies that 
$E^\pm(\lambda)$ extend continuously as $\mathcal B\to\mathcal G$,
and the extensions obviously satisfy \eqref{eq:limFb} and \eqref{eq:2212301227}.
This verifies the assertions (1) and (3). 
To see the joint continuity of $E^\pm\colon I\times\mathcal B\to\mathcal G$ 
we let $\lambda,\mu\in I$ and $\psi,\varphi\in \mathcal B$. 
We take another $\zeta\in C^\infty_{\mathrm c}(\mathbb R^d)$, and split 
\begin{align*}
\|E^\pm(\lambda)\psi-E^\pm(\mu)\varphi\|_{\mathcal G}
&\le 
\|E^\pm(\lambda)\psi-E^\pm(\lambda)\zeta\|_{\mathcal G}
+\|E^\pm(\lambda)\zeta-E^\pm(\mu)\zeta\|_{\mathcal G}
\\&\phantom{{}={}}{}
+\|E^\pm(\mu)\zeta-E^\pm(\mu)\varphi\|_{\mathcal G}
\\&\le 
\bigl\langle \psi-\zeta,\delta(H_{\mathrm L}-\lambda)(\psi-\zeta)\bigr\rangle^{1/2}
+\|E^\pm(\lambda)\zeta-E^\pm(\mu)\zeta\|_{\mathcal G}
\\&\phantom{{}={}}{}
+\bigl\langle \zeta-\varphi,\delta(H_{\mathrm L}-\mu)(\zeta-\varphi)\bigr\rangle^{1/2}
.
\end{align*}
By the locally uniform boundedness of $R_{\mathrm L}(\lambda\pm\mathrm i0)\in\mathcal L(\mathcal B,\mathcal B^*)$
the first and third terms on the right-hand side above can be
arbitrarily small (uniformly in the spectral parameter)
if we choose $\psi,\varphi$ and $\zeta$ close to each other.
For fixed such $\zeta$ the second term can be arbitrarily small 
if $\lambda,\mu$ are close. This is easily seen using the formula 
 \eqref{24050220} and the continuity of $D^\pm(\cdot)\zeta$ recorded
 in Step I. 
Hence (2) is verified. 
\end{proof}

\subsubsection{Construction in general} 
\label{subsec:Comparison of generalized Fourier transforms}

Now we prove Theorem~\ref{thm:comp-gener-four} and Corollary~\ref{cor:230623}.
We implement the effects from $V_{\mathrm S}+q$ 
by the second resolvent identities 
\begin{align}
R(\lambda\pm\mathrm i0)
=R_{\mathrm L}(\lambda\pm\mathrm i0)
\bigl(1-(V_{\mathrm S}+q)R(\lambda\pm\mathrm i0)\bigr)
\in
{\mathcal L(\mathcal B,\mathcal B^*).}
\label{eq:2301153}
\end{align}

 \begin{proof}[Proof of Theorem~\ref{thm:comp-gener-four} (\ref{item:23062110}) and (\ref{item:23062111})] 
{Take the function $\Sigma_{\mathrm L}$ from
Lemma~\ref{lemma:uniq-asympt-norm} with $U=\mathbb S^{d-1}$, and 
we define $F^\pm(\lambda)\psi\in\mathcal G$ for any $(\lambda,\psi)\in I\times \mathcal B$ as 
\begin{equation}
\label{eq:diag90}
F^\pm(\lambda)\psi
=
\mathrm e^{\mp\mathrm i\Theta_{\mathrm L}(\lambda,\cdot)}
E^\pm(\lambda)\bigl(1-(V_{\mathrm S}+q)R(\lambda\pm\mathrm i0)\bigr)\psi,
\end{equation}
where 
\begin{align*}
\Theta_{\mathrm L}(\lambda,\omega)=\Sigma_{\mathrm L}(\lambda,\theta_+(\lambda,\omega));\quad 
\theta_+(\lambda,\cdot)=\omega_+(\lambda,\cdot)^{-1}
.
\end{align*}
Then we can deduce \eqref{eq:221206} 
by \eqref{eq:2301153}, \eqref{eq:diag90} and \eqref{eq:limFb}, 
verifying the assertion (\ref{item:23062110}).}
As for (\ref{item:23062111}), note that the mappings 
\[I\times\mathcal B\to \mathcal B,\ (\lambda,\psi)\mapsto \bigl(1-(V_{\mathrm S}+q)R(\lambda\pm\mathrm i0)\bigr)\psi\] 
are continuous thanks to Theorem \ref{thm:221105}.
Then the assertion (\ref{item:23062111}) is clear from Lemmas~\ref{lem:2301011702} and \ref{lemma:uniq-asympt-norm} 
and Proposition~\ref{prop:22121414}. 
\end{proof}

To prove the assertions
(\ref{item:23062112}) and (\ref{item:23062113}) in
Theorem~\ref{thm:comp-gener-four} we will use the 
\textit{Sommerfeld uniqueness} for $H$, 
or a characterization of the limiting resolvents $R(\lambda\pm\mathrm i0)$.
The following version of the property is almost a direct consequence from Theorem~\ref{thm:proof-strong-bound} and \eqref{eq:2301153}, 
cf. \cite{AIIS2,Is},
however let us present it for completeness of the paper. 

\begin{proposition}\label{prop:13.9.9.8.23}
Let $\lambda\in I$, $\psi\in \mathcal B$ 
and $\phi\in \vB^*$.
Then 
$\phi=R(\lambda\pm\mathrm i0)\psi$ holds if and only if
both of the following assertions hold:
\begin{enumerate}
\item\label{item:13.7.29.0.29}
$\phi$ solves the Helmholtz equation $(H-\lambda)\phi=\psi$ in the distributional sense.
\item\label{item:13.7.29.0.28}
$\phi$ satisfies the outgoing/incoming radiation condition $\gamma_{\|}\phi\in \mathcal B^*_0$.
\end{enumerate}
\end{proposition}

\begin{remarks}\label{remark:construction-general}
\begin{enumerate}[1)]
\item\label{item:1re}
 Here $\gamma_{\|}$ is defined by \eqref{eq:6} with respect to
 $S=S_{\mathrm L}$. In the proposition we can equally well use the more natural
 $\gamma_{\|}$ defined by \eqref{eq:6} with respect to
 the general $S$, or in fact $\gamma_{\|}$ given in terms of the
 expression $S=S_0=\sqrt{2\lambda}|x|$.
\item\label{item:2re} For $\psi=0$ the above result implies the (sharp)
 version of a \emph{Rellich theorem}: If $\phi\in \vB_0^*$ solves
 $(H-\lambda)\phi=0$, then $\phi=0$.
\end{enumerate}
\end{remarks}

\begin{proof}
The necessity is clear from \eqref{eq:2301153}, 
\eqref{eq:230115} and \eqref{eq:gamma1a}.
Thus it remains to prove the sufficiency. 
Assume \eqref{item:13.7.29.0.29} and \eqref{item:13.7.29.0.28} of the assertion,
and set
\begin{align*}
\phi'=\phi-R(\lambda\pm\mathrm i0)\psi\in \mathcal B^*.
\end{align*}
Then by \eqref{eq:2301153}, 
\eqref{eq:230115} and \eqref{eq:gamma1a}
$\phi'$ satisfies 
\[(H-\lambda)\phi'=0,\quad \gamma_{\|}\phi'\in \mathcal B^*_0.\]
We can further verify $\phi'\in \mathcal B^*_0$.
Using a notation similar to
\eqref{eq:230211}:
\begin{equation*}
\eta_T=1-\chi(\chi_1S_L/T),\quad 
\eta_T'=-T^{-1}\chi'(\chi_1S_L/T), 
\end{equation*} we have 
\begin{align*}
2\mathop{\mathrm{Im}}\bigl(\eta_T(H-\lambda)\bigr)
=\pm |\nabla \chi_1S_{\mathrm L}|^2\eta_T'+\mathop{\mathrm{Re}}(\eta_T'\gamma_\|).
\end{align*}
Hence 
\begin{align*}
0\le -\bigl\langle \phi',|\nabla \chi_1S|^2\eta_T'\phi'\bigr\rangle
\le 
\pm \mathop{\mathrm{Re}}\bigl\langle \phi',\eta_T'\gamma_\|\phi'\bigr\rangle.
\end{align*}
By letting $T\to \infty$, 
it follows that indeed $\phi'\in\mathcal B^*_0$.

Since $H$ does not have positive eigenvalues it does not neither have generalized eigenfunctions with positive eigenvalues 
in $\mathcal B^*_0$,
see \cite[Theorem~1.4]{AIIS2}, 
and we certainly obtain that $\phi'=0$. Hence $\phi=R(\lambda\pm\mathrm i0)\psi$.
\end{proof}

The Sommerfeld uniqueness provides the following useful
representations. 
Recalling \eqref{eq:230126bb} we define for any $\xi\in C^\infty(\mathbb S^{d-1})$ 
\begin{equation*}
\psi_\pm^{S_{\mathrm L}}[\xi](x)=\psi_\pm^{S_{\mathrm L}}[\xi](\lambda,x)=(H-\lambda)\phi_\pm^{S_{\mathrm L}}[\xi](\lambda,x)=(H-\lambda)\phi_\pm^{S_{\mathrm L}}[\xi](x).
\end{equation*}

\begin{proposition}\label{prop:14.5.14.3.27}
\begin{subequations}
Let $(\lambda,\xi)\in I\times C^\infty(\mathbb S^{d-1})$.
Then 
\begin{align}
\phi_\pm^{S_{\mathrm L}}[\xi]\in \mathcal B^*,\quad 
\gamma_\|\phi_\pm^{S_{\mathrm L}}[\xi] \in \mathcal B^*_0,\quad 
\psi_\pm^{S_{\mathrm L}}[\xi]\in \mathcal B.
\label{eq:230126}
\end{align}
Moreover, 
\begin{align}
\phi_\pm^{S_{\mathrm L}}[\xi]
&=R(\lambda\pm\mathrm i0)\psi_\pm^{S_{\mathrm L}}[\xi]
\label{eq:23062916}
\end{align}
and 
\begin{align}
\begin{split}
F^\pm(\lambda)^*\bigl(\mathrm e^{\mp\mathrm i\Theta_{\mathrm L}(\lambda,\cdot)}\xi\bigr)
&=
\pm \tfrac{1}{2\pi\mathrm i}
\bigl(\phi^{S_{\mathrm L}}_\pm[\xi]-R(\lambda\mp\mathrm i0)\psi^{S_{\mathrm L}}_\pm[\xi]\bigr)
,
\end{split}
\label{eq:23062917}
\end{align} 
\end{subequations}
respectively.
\end{proposition}
\begin{proof}
The first inclusion from \eqref{eq:230126} is obvious. 
To prove the last one 
we use \eqref{eq:2212116} to rewrite it for $|x|\ge 2R$ as 
\begin{align*}
\psi_\pm^{S_{\mathrm L}}[\xi]
&=
\tfrac{(2\pi)^{1/2}}{2(2\lambda)^{1/4}}
r^{-(d-1)/2}\e^{\pm\i S_{\mathrm L}}
\Bigl(
p^2\pm 2(\nabla S_{\mathrm L})\cdot p
+\mathrm i(d-1)r^{-1}(\nabla r)\cdot p
\mp\mathrm i(\Delta S_{\mathrm L})
\\&\qquad\qquad\qquad\quad
\pm \mathrm i(d-1)r^{-1}(\nabla r)\cdot(\nabla S_{\mathrm L})
+\tfrac{(d-1)(d-3)}4r^{-2}
+2(V_{\mathrm S}+q)\Bigr)\xi.
\end{align*}
Noting that
\begin{align*}
(\nabla S_{\mathrm L})\cdot p
&=
\bigl((\nabla S_{\mathrm L})\cdot(\nabla r)\bigr)(\nabla r)\cdot p
+{\bigl(\nabla (S_{\mathrm L}-\sqrt{2\lambda}r)\bigr)}\cdot \bigl(1-(\nabla r)\otimes(\nabla r)\bigr)p,\\
(\nabla r)\cdot p \xi&=0,\quad p \xi=\vO(r^{-1} )
\end{align*}
and using Theorem~\ref{thm:main result2}, 
we obtain the last inclusion of \eqref{eq:230126}. 
The second one can be verified similarly. 

Now \eqref{eq:23062916} follows from \eqref{eq:230126} 
and Proposition~\ref{prop:13.9.9.8.23}, 
and it remains to verify \eqref{eq:23062917}. 
We can write for any $\psi\in C^\infty_{\mathrm c}(\mathbb R^d)$ 
\begin{align*}
\bigl\langle \psi,F^\pm(\lambda)^*\bigl(\mathrm e^{\mp\mathrm i\Theta_{\mathrm L}(\lambda,\cdot)}\xi\bigr)\bigr\rangle
&=
\bigl\langle F^\pm(\lambda)\psi,\mathrm e^{\mp\mathrm i\Theta_{\mathrm L}(\lambda,\cdot)}\xi\bigr\rangle_{\mathcal G}
\\&
=
\tfrac{(2\lambda)^{1/2}}{2\pi}\lim_{T\to\infty}
\bigl\langle R(\lambda\pm\mathrm i0)\psi,\chi_T'\phi_\pm^{S_{\mathrm L}}[\xi]\bigr\rangle
,
\end{align*}
where $\chi_T'=T^{-1}\chi'(|x|/T)$, cf. \eqref{eq:230211}. 
We use Theorem~\ref{thm:main result2} to proceed as 
\begin{align*}
\bigl\langle \psi,F^\pm(\lambda)^*\bigl(\mathrm e^{\mp\mathrm i\Theta_{\mathrm L}(\lambda,\cdot)}\xi\bigr)\bigr\rangle
&=
\pm \tfrac{1}{2\pi}\lim_{T\to\infty}
\bigl\langle R(\lambda\pm\mathrm i0)\psi,\mathop{\mathrm{Re}}((\nabla \chi_T)\cdot p)\phi_\pm^{S_{\mathrm L}}[\xi]\bigr\rangle
\\&=
\pm \tfrac{1}{2\pi\mathrm i}\lim_{T\to\infty}
\bigl\langle R(\lambda\pm\mathrm i0)\psi,[H-\lambda,1-\chi_T]\phi_\pm^{S_{\mathrm L}}[\xi]\bigr\rangle
\\&=
\pm \tfrac{1}{2\pi\mathrm i}\bigl\langle \psi,\phi_\pm^{S_{\mathrm L}}[\xi]-R(\lambda\mp\mathrm i0)\psi_\pm^{S_{\mathrm L}}[\xi]\bigr\rangle
.
\end{align*} 
This implies \eqref{eq:23062917}. 
\end{proof}

\begin{proof}[Proof of Theorem~\ref{thm:comp-gener-four} (\ref{item:23062112}) and (\ref{item:23062113})] 
To prove the first identity of (\ref{item:23062112}) it suffices to show that 
for any $\lambda\in I$, $\xi\in\mathcal G$ and $\psi\in C^\infty_{\mathrm c}(\mathbb R^d)$
\begin{align}
\langle (H-\lambda)F^\pm(\lambda)^*\xi,\psi\rangle 
=0
.
\label{eq:230624}
\end{align}
However, by Proposition~\ref{prop:13.9.9.8.23} we have 
\[
R(\lambda\pm\mathrm i0)(H-\lambda)\psi=\psi\in C^\infty_{\mathrm c}(\mathbb R^d)
,
\ \ \text{so that }\ F^\pm(\lambda)(H-\lambda)\psi=0.
\]
Thus \eqref{eq:230624} follows. 
On the other hand, by {\eqref{eq:diag90}, 
Proposition~\ref{prop:22121414}} and \eqref{eq:2301153} we have for any $\psi\in\mathcal B$ 
\begin{align}
\begin{split}
\|F^\pm(\lambda)\psi\|_{\mathcal G}^2
&=
\bigl\langle (1-(V_{\mathrm S}+q)R(\lambda\pm\mathrm i0))\psi,
\\&\qquad 
\delta(H_{\mathrm L}-\lambda)(1-(V_{\mathrm S}+q)R(\lambda\pm\mathrm i0))\psi\big\rangle
\\&=
\langle \psi,\delta(H-\lambda)\psi\rangle.
\end{split}
\label{eq:23062918}
\end{align}
This implies the second identity of (\ref{item:23062112}).

To prove (\ref{item:23062113}) we use \eqref{eq:23062916}.
In fact, along with \eqref{eq:221206} and Lemma~\ref{lemma:uniq-asympt-norm}, 
it says for any $(\lambda,\xi)\in I\times C^\infty(\mathbb S^{d-1})$ 
\begin{align}
F^\pm(\lambda)\psi_\pm^{S_{\mathrm L}}[\xi]=\mathrm e^{\mp\mathrm i\Theta_{\mathrm L}(\lambda,\cdot)}\xi
\label{eq:230630}
\end{align}
or
$\mathrm e^{\mp\mathrm i\Theta_{\mathrm L}(\lambda,\cdot)} C^\infty(\mathbb S^{d-1})
\subseteq 
F^\pm(\lambda)\mathcal B\subseteq\mathcal G$. 
Hence we obtain the assertion (\ref{item:23062113}). 
\end{proof}

\begin{proof}[Proof of Corollary~\ref{cor:230623}]
The existence of the wave operators $F^\pm(\lambda)$ is already shown in 
Theorem~\ref{thm:comp-gener-four}. 
Next by Theorem~\ref{thm:comp-gener-four} (\ref{item:23062112}) and (\ref{item:23062113}) 
the scattering matrix $\mathsf S(\lambda)$ is defined at least on a dense subspace of $\mathcal G$,
and in fact it preserves the norm and maps onto a dense set. 
Therefore $\mathsf S(\lambda)$ extends uniquely to a unitary operator on $\mathcal G$. 

Finally we are left with the strong continuity. 
By \eqref{eq:scattering_matrix} and \eqref{eq:230630}
 it follows that for any $\eta\in C^\infty(\mathbb S^{d-1})$ 
\begin{equation}\label{eq:Sform}
 F^+(\lambda)\psi_-^{S_{\mathrm L}}[\eta]
=
\mathsf S(\lambda)F^-(\lambda)\psi_-^{S_{\mathrm L}}[\eta]
=
\mathsf S(\lambda)\mathrm e^{\mathrm i\Theta_{\mathrm L}(\lambda,\cdot)}\eta.
\end{equation}

Note the above left-hand side is continuous in $\lambda\in I$,
and so is the right-hand side. 
Now we fix any $\xi \in \mathcal G$ and $\lambda\in I$, 
and for any $\epsilon>0$ choose $\eta\in C^\infty(\mathbb S^{d-1})$ and $\delta>0$
such that for any $\mu\in (\lambda-\delta, \lambda+\delta)$
\[
\bigl\|\xi -\mathrm e^{\mathrm i\Theta_{\mathrm L}(\mu,\cdot)}\eta\bigr\|_{\mathcal G}<\epsilon. 
\]
Then by the unitarity of the scattering matrix, for any $\mu\in (\lambda-\delta, \lambda+\delta)$ 
\begin{align*}
\|\mathsf S(\lambda)\xi-\mathsf S(\mu)\xi\|_{\mathcal G}
&\le 
\bigl\|\mathsf S(\lambda)\bigl(\xi-\mathrm e^{\mathrm i\Theta_{\mathrm L}(\lambda,\cdot)}\eta\bigr)\bigr\|_{\mathcal G}
+\bigl\|\mathsf S(\lambda)\mathrm e^{\mathrm i\Theta_{\mathrm L}(\lambda,\cdot)}\eta
-\mathsf S(\mu)\mathrm e^{\mathrm i\Theta_{\mathrm L}(\mu,\cdot)}\eta\bigr\|_{\mathcal G}
\\&\phantom{{}={}}{}
+\bigl\|\mathsf S(\mu)\bigl(\mathrm e^{\mathrm i\Theta_{\mathrm L}(\mu,\cdot)}\eta-\xi\bigr)\bigr\|_{\mathcal G}
\\&
<
2\epsilon
+\bigl\|\mathsf S(\lambda)\mathrm e^{\mathrm i\Theta_{\mathrm L}(\lambda,\cdot)}\eta
-\mathsf S(\mu)\mathrm e^{\mathrm i\Theta_{\mathrm L}(\mu,\cdot)}\eta\bigr\|_{\mathcal G}.
\end{align*}
By letting $\delta>0$ be smaller if necessary the above right-hand side is bounded by $3\epsilon$. 
Thus we obtain the desired strong continuity. 
\end{proof}

\subsection{Generalized eigenfunctions}\label{subsec:23032917}
We next prove Theorem~\ref{thm:char-gener-eigenf-1}. 
 
\begin{proof}[Proof of Theorem~\ref{thm:char-gener-eigenf-1}]
\textit{Step I.}\ 
We first show that, if $\phi\in \mathcal E_\lambda$ and $\xi_\pm\in\mathcal G$ satisfy \eqref{eq:gen1}, 
then \eqref{eq:aEigenf2w} holds. 
For that we first compute 
 \begin{align*}
\lim_{m\to\infty}2^{-m}\|1_m\phi\|^2
&=
\lim_{m\to\infty}2^{-m}\bigl\|1_m\bigl(\phi_+^S[\xi_+]-\phi_-^S[\xi_-]\bigr)\bigr\|^2
\\&
=\tfrac{\pi}{(2\lambda)^{1/2}}\biggl(\|\xi_+\|_{\vG}^2+\|\xi_-\|_{\mathcal G}^2
\\&\phantom{{}={}}{}
-\lim_{m\to\infty}
2^{2-m}\mathop{\mathrm{Re}}
\int_{[2^{m-1},2^m)\times \mathbb S^{d-1}}
\mathrm e^{2\mathrm iS(\lambda,r\omega)}\overline{\xi_-(\omega)}\xi_+(\omega)
\,\mathrm dr \mathrm dA(\omega)\biggr).
 \end{align*} 
 Here the last limit vanishes. In fact, we can integrate by parts as 
 \begin{align*}
\int_{2^{m-1}}^{2^m}\mathrm e^{2\mathrm iS(\lambda,r\omega)}\,\mathrm dr
&=
\tfrac1{2\mathrm i}(\partial_rS(\lambda,r\omega))^{-1}\mathrm e^{2\mathrm iS(\lambda,r\omega)}\Big|_{2^{m-1}}^{2^m}
\\&\phantom{{}={}}{}
+\tfrac1{2\mathrm i}\int_{2^{m-1}}^{2^m}(\partial_rS(\lambda,r\omega))^{-2}(\partial_r^2S(\lambda,r\omega))\mathrm e^{2\mathrm iS(\lambda,r\omega)}\,\mathrm dr, 
 \end{align*}
 and it does not contribute to the limit by the conditions of Theorem~\ref{thm:comp-gener-four}. 
Thus we obtain 
 \begin{align}\label{eq:aEigenf2wb}
 \|\xi_+\|_{\vG}^2+\|\xi_-\|_{\mathcal
 G}^2&=
\tfrac{(2\lambda)^{1/2}}{\pi} \lim_{m\to\infty}2^{-m}\|1_m\phi\|^2
.
 \end{align}
On the other hand, proceeding as in the proof of
Proposition~\ref{prop:14.5.14.3.27} and using in the last step the
integration by parts from above, 
we compute 
\begin{align*}
 0
&
=\lim_{T\to \infty}\bigl\langle\phi,\i[H-\lambda, 1-\chi_T]\phi\bigr\rangle
\\&
=-\lim_{T\to \infty}\bigl\langle\phi,\mathop{\mathrm{Re}}((\nabla r)\cdot p)\chi_T'\phi\bigr\rangle
\\&
=-\lim_{T\to \infty}\bigl\langle\phi_+^S[\xi_+]-\phi_-^S[\xi_-],\mathop{\mathrm{Re}}((\nabla r)\cdot p)\chi_T'\phi\bigr\rangle
\\&
=-(2\lambda)^{1/2}\lim_{T\to \infty}\bigl\langle\phi_+^S[\xi_+]+\phi_-^S[\xi_-],\chi_T'\phi\bigr\rangle
\\&
=2\pi\bigl(\|\xi_-\|_{\vG}^2-\|\xi_+\|_{\mathcal G}^2\bigr). 
\end{align*}
In combination with \eqref{eq:aEigenf2wb} this verifies \eqref{eq:aEigenf2w}.

\smallskip
\noindent
\textit{Step II.}\ 
Here we prove the uniqueness asserted in \ref{item:14.5.13.5.40}. 
Suppose $\phi'\in \mathcal E_\lambda$ and $\xi_\pm'\in\mathcal G$ also satisfy \eqref{eq:gen1}. 
Then we have 
\begin{align}
(\phi-\phi') -\phi_+^S[\xi_+-\xi_+']+\phi_-^S[\xi_--\xi_-']\in
 \vB_0^*.
\label{eq:230702}
\end{align}
If $\phi=\phi'$, it follows that $\xi_\pm=\xi_\pm'$ by the result of Step I.
On the other hand, if either of $\xi_\pm=\xi_\pm'$ hold, 
then we have $\xi_\mp=\xi_\mp'$, respectively, again by the result of Step I. 
This and \eqref{eq:230702} imply $\phi-\phi'\in\mathcal B_0$, 
but then it follows that $\phi-\phi'=0$ thanks to Remark \ref{remark:construction-general} \ref{item:2re}.
Thus we obtain the uniqueness.

\smallskip
\noindent
\textit{Step III.}\ 
{Here we} complete the assertions \ref{item:14.5.13.5.40} and \ref{item:14.5.13.5.41}. 
Note for any $\xi\in\mathcal G$ 
\begin{align}\label{eq:brep}
F^\pm(\lambda)^*\xi\mp\tfrac1{2\pi\mathrm i}\bigl(\phi^S_\pm[\xi]-\phi^S_\mp[\mathsf S(\lambda)^{\mp 1}\xi]\bigr)
\in\mathcal B^*_0
.
\end{align} 
In fact, by \eqref{eq:23062917}, \eqref{eq:221206}, 
\eqref{eq:scattering_matrix} and \eqref{eq:230630} (the latter applied as in \eqref{eq:Sform})
we have for any $\xi\in C^\infty(\mathbb S^{d-1})$ 
\begin{align*}
F^\pm(\lambda)^*\bigl(\mathrm e^{\mp\mathrm i\Theta_{\mathrm L}(\lambda,\cdot)}\xi\bigr)
\mp \tfrac{1}{2\pi\mathrm i}
\bigl(\phi^{S}_\pm\bigl[\mathrm e^{\mp\mathrm i\Theta_{\mathrm L}(\lambda,\cdot)}\xi\bigr]
-\phi^{S}_\mp\bigl[\mathsf S(\lambda)^{\mp1} \mathrm e^{\mp\mathrm i\Theta_{\mathrm L}(\lambda,\cdot)}\xi\bigr]\bigr)
\in \mathcal B^*_0
,
\end{align*} 
and then by density of $\mathrm e^{\mp\mathrm i\Theta_{\mathrm L}(\lambda,\cdot)}C^\infty_{\mathrm c}(\mathbb R^d)\subseteq\mathcal G$ 
and the continuity of $F^\pm(\lambda)^*$, $\phi^S_\pm$ and $\mathsf S(\lambda)^{\pm1}$ 
we obtain \eqref{eq:brep}. 
Now, if either of $\xi_\pm\in\mathcal G$ is given,
then the vectors $\phi\in\mathcal E_\lambda$ and $\xi_\mp\in \mathcal
G$ are given by \eqref{eq:aEigenfw}, respectively, and 
obviously satisfy \eqref{eq:gen1} thanks to \eqref{eq:brep}. 
By the uniqueness from Step II we are done with the case where either of 
$\xi_\pm\in\mathcal G$ is given first. 

Next, let $\phi\in\mathcal E_\lambda$ be given first. 
By the above arguments and Step II it suffices to show there exist $\xi_\pm\in \vG$ 
satisfying $\phi=2\pi\mathrm iF^\pm(\lambda)^*\xi_\pm$. 
For each $T\ge 1$ we can find $\xi_{\pm,T}\in \mathcal G$ 
such that for any $\eta\in\mathcal G$ 
\[\langle \eta, \xi_{\pm,T}\rangle_{\mathcal G} 
=\pm\tfrac{(2\lambda)^{1/2}}{2\pi}\bigl\langle \phi^{S}_\pm[\eta],\chi_T'\phi\bigr\rangle. \]
Obviously such $\xi_{\pm,T}\in\mathcal G$ are uniformly bounded for $T\ge 1$, and 
we can choose weakly convergent subsequences $(\xi_{\pm,T_n})_{n\in\mathbb N}$, cf.\ \cite[Theorem~1, p.~126]{Yo}. 
Denote the weak limits by $\xi_\pm\in\vG$.
Then for any $\psi\in C^\infty_{\mathrm c}(\mathbb R^d)$ we compute 
\begin{align*}
\langle \psi,F^\pm(\lambda)^*\xi_\pm\rangle
&=
\langle F^\pm(\lambda)\psi,\xi_\pm\rangle_{\mathcal G}
\\&
=
\pm\tfrac{(2\lambda)^{1/2}}{2\pi}\lim_{n\to\infty}\bigl\langle \phi^S_{\pm}[F^\pm(\lambda)\psi],\chi_{T_n}'\phi\bigr\rangle
\\&
=
\pm\tfrac{(2\lambda)^{1/2}}{2\pi}\lim_{n\to\infty}\bigl\langle R(\lambda\pm\mathrm i0)\psi,\chi_{T_n}'\phi\bigr\rangle.
\end{align*}
Then, as in the proof of Proposition~\ref{prop:14.5.14.3.27}, 
we use Proposition~\ref{prop:13.9.9.8.23}, Theorem~\ref{thm:main
 result2} and the assumption $\phi\in\mathcal E_\lambda$ to proceed as 
\begin{align*}
\langle \psi,F^\pm(\lambda)^*\xi_\pm\rangle
&=
\tfrac1{2\pi}\lim_{n\to\infty}
\bigl\langle \mathop{\mathrm{Re}}((\nabla r)\cdot p)R(\lambda\pm\mathrm i0)\psi,\chi_{T_n}'\phi\bigr\rangle
\\&=
\tfrac1{2\pi\mathrm i}\lim_{n\to\infty}
\bigl\langle R(\lambda\pm\mathrm i0)\psi,[H-\lambda,1-\chi_{T_n}]\phi\bigr\rangle
\\&=
\tfrac1{2\pi\mathrm i}\langle \psi,\phi\rangle
.
\end{align*}
Thus we obtain that $\phi=2\pi\mathrm iF^\pm(\lambda)^*\xi_\pm$. 
The assertions \ref{item:14.5.13.5.40} and \ref{item:14.5.13.5.41} are done.

\smallskip
\noindent
\textit{Step IV.}\ 
Here we prove \ref{item:14.5.13.5.42}. 
Note that the identities \eqref{eq:aEigenf2w} are already established
in Step I. 
Then in combination with \ref{item:14.5.13.5.40} and
\ref{item:14.5.13.5.41} we see $F^\pm(\lambda)^*\colon\vG\to \vE_\lambda\subseteq \mathcal B^*$
 are indeed bi-continuous. Hence we obtain \ref{item:14.5.13.5.42}.

\smallskip
\noindent
\textit{Step V.}\ 
Finally we prove \ref{item:14.5.14.4.17}. 
Since $F^\pm(\lambda)^*$ are injective with closed ranges
in $\mathcal B^*$ by \ref{item:14.5.13.5.42},
Theorem~\ref{thm:comp-gener-four} (\ref{item:23062113}) and Banach's closed range theorem \cite[Theorem p.~205]{Yo} 
imply that the ranges of $F^\pm(\lambda)$ coincide with $\mathcal G$. 
This along with \ref{item:14.5.13.5.42} and
Theorem~\ref{thm:comp-gener-four} (\ref{item:23062112}) in turn implies 
that the range of $\delta(H-\lambda)$
 coincides with $\mathcal E_\lambda$. 
Hence we are done.
\end{proof}

\subsection{Generalized Fourier transforms}\label{subsec:23032918}

We close this section with the proof of Theorem~\ref{thm:221207},
which is 
 rather routine thanks to Theorem~\ref{thm:comp-gener-four}, see also \cite{Sk, IS2}. 

\begin{proof}[Proof of Theorem~\ref{thm:221207}]
\textit{Step I.}\ 
{
We may let $I$ be compact. 
In fact, if $I$ is unbounded, decompose 
\[
I=\bigcup_{n\in \mathbb N}[\lambda_n,\lambda_{n+1}];\quad 
\lambda_1<\lambda_2<\dots<\lambda_n\to\infty\ \ \text{as }n\to\infty,
\]
and, supposing that the assertion holds true for compact intervals, we define 
\[
\mathcal F_I=\bigoplus_{n\in\mathbb N}\mathcal F_{[\lambda_n,\lambda_{n+1}]}\colon \mathcal H_I\to\widetilde{\mathcal H}_I.
\]
Then the assertion for $I$ follows 
due to absence of positive eigenvalues for $H_I$ and $M_\lambda$ and closedness of $H_I$ and $M_\lambda$.
Thus we let $I$ be compact in the following. 
}

\smallskip
\noindent
\textit{Step II.}\ 
Let us construct the 
isometries $\mathcal F^\pm\colon\mathcal H_I\to\widetilde{\mathcal H}_I$.
By Theorem~\ref{thm:comp-gener-four} and Stone's formula \cite[Theorem VII.13]{RS} 
it follows that for any $\psi\in \mathcal B$ 
\begin{align}
\|\mathcal F_0^\pm \psi\|_{\widetilde{\mathcal H}_I}^2
=\int_I\|F^\pm(\lambda)\psi\|_{\mathcal G}^2\,\mathrm d\lambda
=\int_I\langle \psi,\delta(H-\lambda)\psi\rangle\,\mathrm d\lambda
=\|P_{H}(I)\psi\|^2_{\mathcal H_I}
.
\label{eq:230101}
\end{align}
Since $\mathcal B\subseteq\mathcal H$ is dense, 
also $P_{H}(I)\mathcal B\subseteq\mathcal H_I$ is dense.
Thus for any $\psi\in \mathcal H_I$ we can choose 
a sequence $(\psi_n)_{n\in\mathbb N}$ on $\mathcal B$ such that 
$P_{H}(I)\psi_n\to \psi$ in $\mathcal H_I$, and then we can define 
\[\mathcal F^\pm \psi=\lim_{n\to\infty}\mathcal F_0^\pm \psi_n. \]
By \eqref{eq:230101} these limits are well-defined and certainly
define isometries as wanted. 

To be used below we note that by construction for any $\psi\in\mathcal B$
\[\mathcal F^\pm P_H(I)\psi=\mathcal F^\pm_0\psi\in C(I;\mathcal G). \]

\smallskip
\noindent
\textit{Step III.}\ 
Next we show $\mathcal F^\pm H_I=M_\lambda\mathcal F^\pm$. 
Note all the involved operators are bounded. 
Take any $\psi\in C^\infty_{\mathrm c}(\mathbb R^d)$ and $\lambda\in I$. 
If we then set $\psi'=(H-\lambda)\psi$, it follows
from Proposition~\ref{prop:13.9.9.8.23} that 
$\psi=R(\lambda\pm\mathrm i0)\psi'$. Consequently 
\begin{align*}
R(\lambda\pm\mathrm i0)H\psi
=
\lambda R(\lambda\pm\mathrm i0)\psi+\psi.
\end{align*}
This implies that for any $\psi\in C^\infty_{\mathrm c}(\mathbb R^d)$
\begin{align*}
\mathcal F_0^\pm H\psi=M_\lambda\mathcal F_0^\pm\psi.
\end{align*}
Similarly to Step I, for any $\psi\in \mathcal H_I$ we can choose a sequence 
$(\psi_n)_{n\in\mathbb N}$ in $C^\infty_{\mathrm c}(\mathbb R^d)$ 
such that $P_H(I)\psi_n\to \psi$, hence $P_H(I)H\psi_n\to H_I\psi$, in $\mathcal H_I$. 
Then it follows that 
\begin{align*}
\mathcal F^\pm H_I\psi
=\lim_{n\to\infty} \mathcal F_0^\pm H\psi_n
=\lim_{n\to\infty} M_\lambda\mathcal F_0^\pm \psi_n
=M_\lambda\mathcal F^\pm \psi.
\end{align*}
The claim is verified.

\smallskip
\noindent
\textit{Step IV.}\ 
In order to complete (\ref{item:230205a}) and (\ref{item:230205b}) it remains to show 
$\mathcal F^\pm\colon\mathcal H_I\to\widetilde{\mathcal H}_I$ are surjective. 
It suffices to show that {the ranges $\mathcal F^\pm\mathcal H_{I}\subseteq
\widetilde{\mathcal H}_{I}$ are dense.} 
Take any $\Xi\in C(I;\mathcal G)\subseteq\widetilde{\mathcal H}_{I}$,
and fix any $\epsilon>0$.
By the compactness, the result from Step~II and
Theorem~\ref{thm:char-gener-eigenf-1} \ref{item:14.5.14.4.17} 
we can find a finite open covering $I\subseteq U_1\cup\dots\cup U_n$
and $\psi_1,\dots,\psi_n\in P_H(I)\mathcal B\subseteq \mathcal H_{I}$ 
such that for any $i=1,\dots,n$ and $\lambda\in I\cap U_i$
\[\bigl\|\Xi(\lambda)-(\mathcal F^\pm\psi_i)(\lambda)\bigr\|_{\mathcal G}<\epsilon. \]
If we let $\{\chi_i\}_i$ be a partition of unity subordinate to $\{U_i\}_i$,
then for any $\lambda\in I$
\[\norm[\Big]{\Xi(\lambda)
-\sum_{i=1}^n\chi_i(\lambda)(\mathcal F^\pm\psi_i)(\lambda)}_{\mathcal G}
<\epsilon. \]
Since $\mathcal F^\pm\psi_i\in C(I;\mathcal G)$, 
we can replace the above $\chi_i$ by some polynomials $\widetilde\chi_i$
due the Weierstrass approximation theorem with an additional $\epsilon$ error. 
Then set 
\[\psi =\sum_{i=1}^n\widetilde\chi_i(H)\psi_i\in\mathcal H_{I},\]
and we obtain by the result from Step III that for any $\lambda\in I$
\[\left\|\Xi(\lambda)-(\mathcal F^\pm\psi)(\lambda)\right\|_{\mathcal G}<2\epsilon. \]
Hence (\ref{item:230205a}) and (\ref{item:230205b}) are done.

\smallskip
\noindent
\textit{Step V.}\ 
The assertion (\ref{item:230205}) is clear from the definition of the generalized Fourier transforms,
see e.g.\ \eqref{eq:diag90}, and Lemma~\ref{lemma:uniq-asympt-norm} with $U=\mathbb S^{d-1}$. 
We are done. 
\end{proof}

\section{Time-dependent scattering theory}\label{sec:Time-dependent theory}

In this section we discuss the time-dependent scattering theory, 
proving Theorems~\ref{thm:2307032}, \ref{thm:230703} and \ref{thm:time-depend-theory2}.
Our arguments are heavily dependent on the stationary theory from the
previous sections in parallel to 
 \cite{IS3}, however for $2$-admissible
potentials the low degree of smoothness entails some complication. 
In Subsection~\ref{subsec:23070422} we prove Theorem~\ref{thm:2307032} by rigorously justifying the Legendre transform, 
and then we compare the asymptotics of solutions to the time-dependent eikonal equation. 
With these results in place we show Theorem~\ref{thm:230703} in
Subsection~\ref{subsec:23070423}, hence obtaining 
the existence and covariance of the time-dependent wave operators. 
Finally Theorem~\ref{thm:time-depend-theory2}, in particular including asymptotic completeness, 
is shown by comparing the time-dependent wave operators with the generalized Fourier transforms.

\subsection{Time-dependent eikonal equation}\label{subsec:23070422}

Here we prove Theorem~\ref{thm:2307032}. 
We justify the Legendre transform, 
investigating its properties. 

\begin{proof}[Proof of Theorem~\ref{thm:2307032}]
Fix $\mu,\mu',I,R$ and $S=(2\lambda)^{1/2}|x|(1+s)$ as in the assertion. 
We first claim, letting $R>0$ be larger if necessary, that for each $(t,x)\in\Omega_\mu$ 
we can find a unique critical point of the function \eqref{eq:230705} in the
variable $\lambda
\in I$. 
In fact, we can compute 
\begin{align}
\begin{split}
\partial_\lambda \widetilde K
&=(2\lambda)^{-1/2}|x|\bigl(1+s+2\lambda (\partial_\lambda s)\bigr)-t
,
\\
\partial_\lambda^2 \widetilde K
&=
-(2\lambda)^{-3/2}|x|\bigl(1+s-4\lambda(\partial_\lambda s)-4\lambda^2 (\partial_\lambda^2 s)\bigr)
,
\end{split}
\label{eq:230724}
\end{align}
and thus by Corollary~\ref{cor:epsSmall} for sufficiently large $R>0$ 
\begin{align*}
(\partial_\lambda \widetilde K)(\mu'^2/2,t,x)
\ge (\tfrac{\mu+\mu'}2)^{-1}|x|-t>0,\quad 
\lim_{\lambda\to\infty} (\partial_\lambda \widetilde K)(\lambda,t,x)=-t<0,
\end{align*}
and for any $\lambda\in I$
\begin{align*}
(\partial_\lambda^2 \widetilde K)(\lambda,t,x)\le -\tfrac12(2\lambda)^{-3/2}|x|<0
.
\end{align*}
This implies that there uniquely exists $\lambda_{\mathrm c}=\lambda_{\mathrm c}(t,x)\in I$ such that 
 \begin{equation}\label{eq:implicit}
(\partial_\lambda \widetilde K)(\lambda_{\mathrm c},t,x)
=(2\lambda_{\mathrm c})^{-1/2}|x|\bigl(1+s(\lambda_{\mathrm c},x)+2\lambda_{\mathrm c}(\partial_\lambda s)(\lambda_{\mathrm c},x)\bigr)-t
=0
.
 \end{equation}
By the implicit function theorem the function $\lambda_{\mathrm c}\in C^{l-1}(\Omega_{\mu})$. 
Now we set 
\begin{equation}\label{eq:Kdefs}
 K=\widetilde K(\lambda_\c,\cdot,\cdot), 
\end{equation} 
and then we can easily see by \eqref{eq:implicit} that 
\begin{align}
\partial_tK&=-\lambda_\c,\quad 
\nabla_x K=(\nabla_x S)(\lambda_{\mathrm c},\cdot),\label{eq:eik222}
\end{align} 
which verifies $K\in C^l(\Omega_\mu)$ and \eqref{eq:eik223}. We have
proven Theorem~\ref{thm:2307032} \eqref{item:22120919ax}.

Next we prove the assertion (\ref{item:22120919dx}). 
For that we claim for any $k+|\alpha|\le 1$ 
\begin{align}
\bigl|\partial_t^k\partial_x^\alpha \bigl(\lambda_{\mathrm c}-\tfrac{x^2}{2t^2}\bigr)\bigr|
\le C_1t^{-k}\langle x\rangle^{-\sigma-|\alpha|}
.
\label{eq:23070618}
\end{align}
In fact, note that \eqref{eq:implicit} and
Corollary~\ref{cor:epsSmall} imply (uniformly in
 large $R>0$)
\begin{align}
c_2t^{-2}x^2\le \lambda_{\mathrm c}\le C_2t^{-2}x^2. 
\label{eq:23070613}
\end{align} 
Note also that $\lambda_{\mathrm
 c}=\tfrac{x^2}{2t^2} $ for $|x|\leq R$, so to show the bounds
 \eqref{eq:23070618} and 
 \eqref{eq:errestza2x} we can indeed assume that $|x|> R$.
Then by combining \eqref{eq:implicit}, Theorem~\ref{thm:main result2}, Corollary~\ref{cor:epsSmall} and \eqref{eq:23070613} 
we obtain 
\begin{align}
\bigl|\lambda_{\mathrm c}-\tfrac{x^2}{2t^2}\bigr|
=\tfrac{x^2}{2t^2}\bigl|s+2\lambda_{\mathrm c} (\partial_\lambda s)\bigr|
\bigl|2+s+2\lambda_{\mathrm c} (\partial_\lambda s)\bigr|
\le C_3\langle x\rangle^{-\sigma}
. 
\label{eq:23070617}
\end{align}
This shows the claim for $k+|\alpha|=0$. 
For $k+|\alpha|=1$ we compute the derivatives of $\lambda_{\mathrm c}$ by 
the Leibniz rule applied to \eqref{eq:implicit}, equivalently written as 
\[
1=\tfrac{x^2}{2t^2}\lambda_{\mathrm c}^{-1}(1+s+2\lambda_{\mathrm c}\partial_\lambda s)^2
. 
\]
In fact, we have for $k=1$ and $|\alpha|=0$ 
\begin{align*}
0&=
-\tfrac{x^2}{t^3}\lambda_{\mathrm c}^{-1}(1+s+2\lambda_{\mathrm c}\partial_\lambda s)^2
\\&\phantom{{}={}}{}
-\tfrac{x^2}{2t^2}\bigl(\partial_t\lambda_{\mathrm c}\bigr)
\lambda_{\mathrm c}^{-2}\bigl(1+s-4\lambda_{\mathrm c}\partial_\lambda s-4\lambda_{\mathrm c}^2\partial_\lambda^2 s\bigr)
(1+s+2\lambda_{\mathrm c}\partial_\lambda s)
, 
\end{align*}
and for $k=0$ and $|\alpha| =1$ 
\begin{align*}
0&=
\big(\partial_x^\alpha\tfrac{x^2}{2t^2}\bigr)\lambda_{\mathrm c}^{-1}(1+s+2\lambda_{\mathrm c}\partial_\lambda s)^2
\\&\phantom{{}={}}{}
-\tfrac{x^2}{2t^2}\bigl(\partial_x^\alpha\lambda_{\mathrm c}\bigr)
\lambda_{\mathrm c}^{-2}\bigl(1+s-4\lambda_{\mathrm c}\partial_\lambda s-4\lambda_{\mathrm c}^2\partial_\lambda^2 s\bigr)
(1+s+2\lambda_{\mathrm c}\partial_\lambda s)
\\&\phantom{{}={}}{}
+\tfrac{x^2}{t^2}\lambda_{\mathrm c}^{-1}
\bigl(\partial_x^\alpha s+2\lambda_{\mathrm c}\partial_\lambda\partial_x^\alpha s\bigr)(1+s+2\lambda_{\mathrm c}\partial_\lambda s)
.
\end{align*}
Possibly by taking $R>0$ larger from the beginning (if necessary) and
by using Corollary~\ref{cor:epsSmall}, we can 
for $k=1$ and $|\alpha|=0$ write 
\begin{align*}
\partial_t\lambda_{\mathrm c}
&=
-\tfrac{2}{t}\lambda_{\mathrm c}(1+s+2\lambda_{\mathrm c}\partial_\lambda s)
\bigl(1+s-4\lambda_{\mathrm c}\partial_\lambda s-4\lambda_{\mathrm c}^2\partial_\lambda^2 s\bigr)^{-1}
, 
\end{align*}
and for $k=0$ and $|\alpha| =1$ 
\begin{align*}
\partial_x^\alpha\lambda_{\mathrm c}
&=
\tfrac{2t^2}{x^2}\big(\partial_x^\alpha\tfrac{x^2}{2t^2}\bigr)\lambda_{\mathrm c}(1+s+2\lambda_{\mathrm c}\partial_\lambda s)
\bigl(1+s-4\lambda_{\mathrm c}\partial_\lambda s-4\lambda_{\mathrm c}^2\partial_\lambda^2 s\bigr)^{-1}
\\&\phantom{{}={}}{}
+2\lambda_{\mathrm c}\bigl(\partial_x^\alpha s+2\lambda_{\mathrm c}\partial_\lambda\partial_x^\alpha s\bigr)
\bigl(1+s-4\lambda_{\mathrm c}\partial_\lambda s-4\lambda_{\mathrm c}^2\partial_\lambda^2 s\bigr)^{-1}
.
\end{align*}
The claim \eqref{eq:23070618} follows from the above expressions, \eqref{eq:23070617} and Theorem~\ref{thm:main result2}.

One can verify the assertion (\ref{item:22120919dx}) by \eqref{eq:230705}, \eqref{eq:Kdefs}, \eqref{eq:23070618}, \eqref{eq:eik222} 
and Theorem~\ref{thm:main result2}. 
While the bounds \eqref{eq:errestza2x} follow
 immediately unless $k=0$ and $|\alpha| =2$, the latter case
 requires some other computations. We omit the details of proof for that case.
\end{proof}

We next investigate the asymptotics of general solutions to \eqref{eq:4cb}. 

\begin{lemma}\label{lem:2307172208}
Let $\mu,T,\Omega_{\mu,T}$ and $K$ satisfy the assumption of Theorem~\ref{thm:230703},
and let $K_{\mathrm L}$ be the Legendre transform of the function
$S_{\mathrm L}$ taken from Subsection~\ref{subsec:230329} with 
$I=I_{\mu'}=[\mu'^2/2,\infty)$, $\mu'\in (0,\mu)$, and $l'>1+2/\rho$.
In addition, let $y$ be the flow associated with $S_{\mathrm L}$ 
as in Subsection~\ref{subsec:Generalized Fourier transform for rho>1/2 I}, and set 
\begin{equation*}
\tau(\lambda,t,\theta)=\int_0^t|\nabla S_{\mathrm L}(\lambda,y(\lambda,s,\theta))|^{-2}\,\mathrm ds. 
\end{equation*}
Then the following assertions hold.
\begin{enumerate}[(1)]
\item \label{item:sigma10}
There exists the limit 
\begin{align*}
\Xi_{\mathrm L}(\lambda,\theta)
:=
&\lim_{t\to\infty}\bigl(K(\tau(\lambda,t,\theta),y(\lambda,t,\theta))
-K_{\mathrm L}(\tau(\lambda,t,\theta),y(\lambda,t,\theta))
\bigr)
\end{align*}
taken locally uniformly in $(\lambda,\theta)\in J\times \mathbb S^{d-1}$, where $J=[\mu^2/2,\infty)$.
In particular, 
if $K_1$ also satisfies the assumption of Theorem~\ref{thm:230703},
there exists the limit 
\begin{align*}
\Xi(\lambda,\theta)
:=
\lim_{t\to\infty}\bigl(K_1(\tau(\lambda,t,\theta),y(\lambda,t,\theta))
-K(\tau(\lambda,t,\theta),y(\lambda,t,\theta))
\bigr)
\end{align*}
taken locally uniformly in $(\lambda,\theta)\in J\times \mathbb S^{d-1}$. 
\item \label{item:sigma20} The quantities in \ref{item:sigma10} can {also} be computed
 as the limits
\begin{align*}
\Phi_L(\lambda,\omega)&:=\Xi_L(\lambda,\theta_+(\lambda,\omega))=\lim_{\tau\to \infty}\parb{ K(\tau,
 (2\lambda)^{1/2}\tau\omega)-K_L(\tau, (2\lambda)^{1/2}\tau\omega)},\\
\Phi(\lambda,\omega)&:=\Xi(\lambda,\theta_+(\lambda,\omega))=\lim_{\tau\to \infty}\parb{ K_1(\tau,
 (2\lambda)^{1/2}\tau\omega)-K(\tau, (2\lambda)^{1/2}\tau\omega)},
\end{align*} both taken locally uniformly in $(\lambda,\omega)\in J\times\mathbb
 S^{d-1}$.
\end{enumerate}
\end{lemma}
\begin{remarks*} The requirement
 $l'>1+2/\rho$ (not used in the proof) will be needed in the proof of Lemma
 \ref{lem:time-depend-theory}, 
cf. Remark \ref{remark:l_prime}. The above function $\tau$ should be considered as the
 `physical time', cf. 
 Remark \ref{rem:230218} \ref{item:geom2}.
\end{remarks*}
\begin{proof} 
\textit{Step I.}\ 
The second assertion of \ref{item:sigma10} is clear from the first
one, hence we only prove the first assertion of \ref{item:sigma10}. 
It suffices to show existence of the limits
\begin{subequations}
\begin{align}
{\Psi(\lambda,\theta):=}
\lim_{t\to\infty}\bigl(K(\tau(\lambda,t,\theta),y(\lambda,t,\theta))
-S_{\mathrm L}(\lambda,y(\lambda,t,\theta))
+\lambda \tau(\lambda,t,\theta) \bigr)
\label{eq:23071720}
\end{align}
and 
\begin{align}
{\Psi_{\mathrm L}(\lambda,\theta):=}
\lim_{t\to\infty}\bigl(K_{\mathrm L}(\tau(\lambda,t,\theta),y(\lambda,t,\theta))
-S_{\mathrm L}(\lambda,y(\lambda,t,\theta))
+\lambda \tau(\lambda,t,\theta) \bigr)
\label{eq:23071720b}
\end{align} 
 taken {locally} uniformly in $(\lambda,\omega)\in J\times \mathbb S^{d-1}$. 
We can prove them in the same manner, but it follows 
easily from the proof of \eqref{eq:23071720} that in fact 
\begin{equation}\label{eq:23071720bb}
  K_{\mathrm L}(\tau(\lambda,t,\theta),y(\lambda,t,\theta))
-S_{\mathrm L}(\lambda,y(\lambda,t,\theta))
+\lambda \tau(\lambda,t,\theta) =0, 
\end{equation}
hence we discuss only \eqref{eq:23071720}.
\end{subequations}

Note that all the arguments below are {locally} uniform in $(\lambda,\theta)\in J\times\mathbb S^{d-1}$. 
Omitting the arguments, and using \eqref{eq:eik223}, \eqref{eq:flow} and \eqref{eq:2212116},
we can compute the $t$-derivative for large $t>0$ as 
\begin{align}
\tfrac{\mathrm d}{\mathrm dt}(K-S_{\mathrm L}+\lambda\tau)
=
-\tfrac12|\nabla S_{\mathrm L}|^{-2}(\nabla K-\nabla S_{\mathrm L})^{2}
-|\nabla S_{\mathrm L}|^{-2}V_{\mathrm S}.
\label{eq:23071721}
\end{align}
The second term on the right-hand side of \eqref{eq:23071721} is
obviously integrable at infinity,
and thus it suffices to show that for some $\delta>0$
\begin{align}
u
:=(\nabla K-\nabla S_{\mathrm L})^{2}
=\mathcal O(t^{-1-\delta}).
\label{eq:2307172154}
\end{align}
For that, similarly to \eqref{eq:23071721}, we further differentiate it for large $t> 0$ as 
\begin{align}\label{eq:Dife1}
\tfrac12|\nabla S_{\mathrm L}|^2\tfrac{\mathrm d}{\mathrm dt}u
=-(\nabla K-\nabla S_{\mathrm L})(\nabla^2K)(\nabla K-\nabla S_{\mathrm L})
-(\nabla V_{\mathrm S})\cdot(\nabla K-\nabla S_{\mathrm L}).
\end{align}
By \eqref{eq:errestza2}, \eqref{eq:errestza2x} and \eqref{eq:230115}, 
 it follows that for some (small) $\delta>0$ and (big) $C_1>0$ and any large $t>0$ 
\begin{align*}
\tfrac{\mathrm d}{\mathrm dt}u
\le -(2-\delta)t^{-1}u+C_1t^{-2-\delta}.
\end{align*}
This certainly implies \eqref{eq:2307172154}. Thus we are done with
\ref{item:sigma10}.

\smallskip
\noindent 
\textit{Step II.}\ 
To prove the assertion \ref{item:sigma20} we discuss the following change of variables. 
We claim that for all large $\tau>0$ and any $(\lambda,\omega)\in J\times \mathbb S^{d-1}$ there exist $\bar\lambda$ and $\bar\theta$ such that
\begin{subequations}
\begin{equation}\label{eq:Impli1}
 (\tau,
 (2\lambda)^{1/2}\tau\omega)=(\tau(\bar\lambda,t,\bar\theta),y(\bar\lambda,t,\bar\theta))\quad\text {with}\quad
 t=S_L(\bar\lambda,(2\lambda)^{1/2}\tau\omega), 
\end{equation} 
{and} that
\begin{equation}\label{eq:Impli2}
 \lim_{\tau\to \infty }\,\bar\lambda(\tau)=\lambda\quad \mand \quad\lim_{\tau\to \infty }\,\bar\theta(\tau)=\theta_+(\lambda,\omega).
\end{equation} 
\end{subequations}
{
First} we solve  for fixed $\bar \lambda$ and $\bar\theta$ the equation
\begin{equation*}
\tau=\int_0^t|\nabla S_{\mathrm
 L}(\bar\lambda,y(\bar\lambda,s,\bar
\theta))|^{-2}\,\mathrm ds. 
\end{equation*} for $t=t(\tau,\bar \lambda,\bar\theta)$
(with a $C^1$-dependence by the implicit function theorem). 
For fixed large $\tau$ we then need to solve the equation
\begin{equation}\label{eq:solcHJ}
 (2\lambda)^{1/2}\omega=y(\bar\lambda,t,\bar\theta)/\tau=:F_\tau(\bar\lambda,\bar\theta)
\end{equation} for $\bar \lambda$ and $\bar\theta$. This can be done
by using the inverse
function theorem in a version applicable to
parameter-dependent problems (note that $\tau$ is the relevant
parameter), for example   the version stated as \cite[Theorem
D.1]{Is2}.  We need
to verify that the derivative of
$F_\tau$ at $z_0:=(\lambda,\theta_+(\lambda,\omega))$ is
non-degenerate near  infinity. 

First note that
\begin{equation*}
 F_\infty(\bar\lambda,\bar\theta):=\lim_{\tau\to
 \infty}F_\tau(\bar\lambda,\bar\theta)=(2\bar\lambda)^{1/2}\omega_+(\bar\lambda,\bar\theta)
 \text{ in a neighbourhood of }z_0,
\end{equation*} in particular that 
\begin{equation*}
 F_\infty(z_0)=(2\lambda)^{1/2}\omega.
\end{equation*}
 Next note
that
\begin{equation*}
 \omega_+(\bar\lambda,\bar\theta)=\int_0^\infty\,
 \tfrac{\partial }{\partial s}\omega(\bar\lambda,s,\bar\theta)\, \d
 s+\bar\theta\quad\mand\quad\partial_{\bar\lambda} \omega_+(\bar\lambda,\bar\theta)=\int_0^\infty\,
 \partial_{\bar\lambda}\partial_{s} \omega(\bar\lambda,s,\bar\theta)\, \d s.
\end{equation*} By the latter formula, the representation \eqref{eq:deta} and Corollary
\ref{cor:epsSmall} it follows that $\partial_{\bar\lambda}\omega_+(\bar\lambda,\bar\theta)$ is small when the
parameter $R$ is sufficiently big. When combined with
  \eqref{eq:fixedP3} this leads to the
  conclusion that the map
$(\bar\lambda,\bar\theta)\to
F_\infty(\bar\lambda,\bar\theta)=(2\bar\lambda)^{1/2}\omega_+(\bar\lambda,\bar\theta)$ is non-degenerate
$C^1$ near $z_0$ (uniformly in large $R$).

Now to  solve \eqref{eq:solcHJ}, by  applying \cite[Theorem D.1]{Is2},  
all that remains to be seen is that the limits 
\begin{equation*}
 \lim_{\tau\to
 \infty}\partial_{\bar\lambda,\bar\theta}{F_\tau(\bar\lambda,\bar\theta)}\text{
 exist, uniformly in a neighbourhood of }z_0.
\end{equation*} We skip the details of the verification of this
uniform convergence, and conclude
the solvability of \eqref{eq:solcHJ}. This justifies 
 \eqref{eq:Impli1} and \eqref{eq:Impli2}.

\smallskip
\noindent
{\textit{Step III.}\ 
Now we prove the first assertion of \ref{item:sigma20} (this implies the second one). 
We proceed partly in parallel to the proof of Lemma \ref{lemma:uniq-asympt-norm} \ref{item:sigma2}.  
It suffices to show that with notation from Step~I 
\begin{align*}
\Psi(\lambda,\theta_+(\lambda,\omega))
-\Psi_{\mathrm L}(\lambda,\theta_+(\lambda,\omega))
&=
\lim_{\tau\to\infty}\bigl(K(\tau,(2\lambda)^{1/2}\tau\omega)
-K_{\mathrm L}(\tau,(2\lambda)^{1/2}\tau\omega)
\bigr)
\end{align*} 
locally uniformly in $(\lambda,\omega)\in J\times\mathbb S^{d-1}$. 
Take any compact subset $L\subset J\times \mathbb S^{d-1}$. 
Let $(\bar\lambda,\bar\theta)=(\bar\lambda(\tau),\bar\theta(\tau))$ be the change of variables from Step~II. 
Using the locally uniform limit
\eqref{eq:240722} along with
\eqref{eq:Impli1} and \eqref{eq:Impli2},  and noting that $\tau\to\infty$ corresponds to $t\to\infty$, 
it follows that for some compact subset $L'\subset J\times \mathbb S^{d-1}$ 
\begin{align*}
&
\lim_{\tau\to\infty}\sup_{(\lambda,\omega)\in L}\bigl|
\Psi(\lambda,\theta_+(\lambda,\omega))
-\Psi_{\mathrm L}(\lambda,\theta_+(\lambda,\omega))
-K(\tau,(2\lambda)^{1/2}\tau\omega)+K_{\mathrm L}(\tau,(2\lambda)^{1/2}\tau\omega)\bigr)
\bigr|
\\&
=
\lim_{\tau\to\infty}\sup_{(\lambda,\omega)\in L}\bigl|
\Psi(\bar\lambda,\theta(\bar\lambda,(2\lambda)^{1/2}\tau\omega))
-\Psi_{\mathrm L}(\bar\lambda,\theta(\bar\lambda,(2\lambda)^{1/2}\tau\omega))
\\[-.5em]&\qquad\qquad\qquad\qquad\qquad\qquad\qquad
-K(\tau,(2\lambda)^{1/2}\tau\omega)+K_{\mathrm L}(\tau,(2\lambda)^{1/2}\tau\omega)
\bigr|
\\[.5em]&
=
\lim_{t\to\infty}\sup_{(\bar\lambda,\bar\theta)\in L'}\bigl|
\Psi(\bar\lambda,\theta(\bar\lambda,y(\bar\lambda,t,\bar\theta)))
-\Psi_{\mathrm L}(\bar\lambda,\theta(\bar\lambda,y(\bar\lambda,t,\bar\theta)))
\\[-.5em]&\qquad\qquad\qquad\qquad\qquad
-K(\tau(\bar\lambda,t,\bar\theta),y(\bar\lambda,t,\bar\theta))+K_{\mathrm L}(\tau(\bar\lambda,t,\bar\theta),y(\bar\lambda,t,\bar\theta))
\bigr|
\\[.5em]&
=
\lim_{t\to\infty}\sup_{(\bar\lambda,\bar\theta)\in L'}\bigl|
\Psi(\bar\lambda,\bar\theta)
-\Psi_{\mathrm L}(\bar\lambda,\bar\theta)
\\[-.5em]&\qquad\qquad\qquad\qquad\qquad
-K(\tau(\bar\lambda,t,\bar\theta),y(\bar\lambda,t,\bar\theta))+K_{\mathrm L}(\tau(\bar\lambda,t,\bar\theta),y(\bar\lambda,t,\bar\theta))
\bigr|
\\[.5em]&
=0
.
\end{align*}
Thus we are done.}
\end{proof}

\subsection{Time-dependent wave operators}\label{subsec:23070423}
 
Here we prove Theorems~\ref{thm:230703} and \ref{thm:time-depend-theory2}.

\subsubsection{Existence} 

\begin{proof}[Proof of Theorem~\ref{thm:230703}]
 We first prove the existence of the strong limits
\eqref{eq:wave7} by the
  familiar Cook--Kuroda method. 
Thanks to a density argument based on uniform boundedness of $U^\pm(t)$ and $\mathrm e^{\pm\mathrm itH}$ in $t>T$,
it suffices to show that for any $h\in C_\c^\infty(\R_+\times\mathbb S^{d-1})$ 
with $\mathop{\mathrm{supp}}h\subseteq J\times\mathbb S^{d-1}$, there exist the limits 
\begin{align}
\lim _{t\to \infty} \e^{\pm \i tH}U^\pm(t)h
.
\label{eq:230608}
\end{align}
For that we show integrability of 
\[
\bigl\|\tfrac{\partial}{\partial t}\e^{\pm \i tH}U^\pm(t)h\bigr\|_{\mathcal H}
=
\bigl\|\bigl(\tfrac{\partial}{\partial t}\pm \i H\bigr)U^\pm(t)h\bigr\|_{\mathcal H}
\]
at infinity. 
Using the Hamilton--Jacobi equation  \eqref{eq:4cb} and letting $\omega'=(\omega_2',\dots,\omega_d')$ be any local coordinates of $\mathbb S^{d-1}$, 
we compute 
\begin{align*}
 \bigl(\tfrac{\partial}{\partial t}\pm \i H\bigr)U^\pm(t)h
&=
\e^{\mp 3\pi\i/4}t^{-1}\abs{x}^{1-d/2}\e^{\pm \i K}
\\&\phantom{{}={}}{}\cdot
\Bigl[\mp \tfrac{\mathrm i}2 x^2 t^{-4}(\partial_\lambda^2 h)
\mp \tfrac{\mathrm i}2 |x|^{-2}\bigl(\Delta_{\mathbb S^{d-1}}h\bigr)
+(\partial_i\omega_\alpha')(\partial_iK) (\partial_\alpha h)
\\&\phantom{{}={}}{}\quad
+\Bigl(-t^{-3}x^2\mp \i t^{-2}+ t^{-2}x\cdot (\nabla_x K)\Bigr)(\partial_\lambda h)
\\&\phantom{{}={}}{}\quad
+\Bigl(-t^{-1}\pm \i \tfrac d4\bigl(1-\tfrac d2\bigr)|x|^{-2}
\pm\mathrm iq
\\&\phantom{{}={}}{}\quad\quad\ \ 
+\bigl(1-\tfrac d2\bigr) |x|^{-2}x\cdot (\nabla_x K)
+\tfrac12 (\Delta_x K)
\Bigr)h
\Bigr]
.
\end{align*} 
By the assumption on $K$ and the support property of $h$ 
the last expression is of order $t^{-1-\min\{\epsilon,\tau\}}$ with values in $\mathcal H$. 
Hence the limits \eqref{eq:230608} exist. 
 Since  $U^\pm(t)$ and $\mathrm
e^{\pm\mathrm itH}$ are  isometries, it is clear that so  are
$W^\pm$. By the above computation we easily see that $W^\pm M_\lambda
\subseteq H W^\pm$. In particular the mapping property  stated in the
last part of 
 Theorem~\ref{thm:230703} \eqref{item:23071721}
 follows. 

The assertion   Theorem~\ref{thm:230703} \eqref{item:23071722} follows
readily from
\eqref{item:23071721} and  Lemma \ref{lem:2307172208} \ref{item:sigma20}.
\end{proof}
 
We  present alternative representations of $W^\pm$, which  
 will be useful in the proof of their  completeness. 
See also \cite[Lemma 3.3]{IS3}.

\begin{lemma}\label{lem:time-depend-theory} 
In the setting of Lemma~\ref{lem:2307172208}, 
define for any $h\in C^\infty_{\mathrm c}(\R_+\times\mathbb S^{d-1})$
with $\mathop{\mathrm{supp}}h\subseteq J\times \mathbb S^{d-1}$  the evolutions
\begin{align*}
 (\widetilde U^\pm_{\mathrm L}(t)h)(x)
&=
(\pm 2\pi\mathrm i)^{-1}
\int_0^\infty \e^{\mp \i \lambda t }\phi_\pm^{S_{\mathrm L}}[h(\lambda,\cdot)](\lambda,x)\,\mathrm d\lambda
, \quad (t,x)\in\R_+\times \mathbb R^d,  
\end{align*} 
where $\phi_\pm^{S_{\mathrm L}}[\cdot]$ is from \eqref{eq:230126bb}. 
Then for each $t>0$ it follows that  $\widetilde U^\pm_{\mathrm L}(t)h\in\mathcal
H$. Moreover
\begin{equation*}
 W^\pm\mathrm e^{\mp\mathrm i\Phi_{\mathrm L}}h=\lim
 _{t\to \infty} \e^{\pm \i tH}\widetilde  U^\pm_{\mathrm L}(t)h 
,
\end{equation*} 
respectively.
 \end{lemma}
\begin{remark}\label{rem230718}
The above evolution map $\widetilde U^\pm_{\mathrm L}(t)$ motivates our free comparison dynamics $U^\pm(t)$. 
In fact, due to the stationary phase theorem the leading term is very similar to that of $U^\pm(t)$, see the proof below. 
Note however  that we need higher order derivatives of the phase function. 
This is why we substitute $S_{\mathrm L}$ for $S$ in $\widetilde U^\pm_{\mathrm L}(t)$. 
The error is compensated by the factor $\mathrm e^{\mp\mathrm i\Phi_{\mathrm L}}$.
\end{remark}
\begin{proof}
\textit{Step I.}\ 
To show that $\widetilde  U^\pm_{\mathrm L}(t)h\in \mathcal H$ we introduce 
\begin{equation*}
\widetilde K_{\mathrm L}(\lambda,t,x)=S_{\mathrm L}(\lambda,x)-\lambda t,
\end{equation*}
 cf.\ \eqref{eq:230705}. 
First fix any $t>0$, and let $R'>0$ be sufficiently large. 
Then by the expressions \eqref{eq:230724} corresponding to $\widetilde K_{\mathrm L}$ 
we have for any $(\lambda,\hat x)\in\mathop{\mathrm{supp}}h$ with $|x|>R'$
\begin{align}
\partial_\lambda \widetilde K_{\mathrm L}(\lambda,t,x)\ge c_1\langle x\rangle
,\quad 
\bigl|\partial_\lambda^2 \widetilde K_{\mathrm L}(\lambda,t,x)\bigr|\le C_1\langle x\rangle
. 
\label{eq:23072410}
\end{align}
Thus we can integrate by parts as 
\begin{align}
 \label{eq:appDy1}
\widetilde  U^\pm_{\mathrm L}(t)h
&=
(2\pi)^{-1/2}\chi\abs{x}^{-(d-1)/2}
\int_0^\infty \e^{\pm \i \widetilde K_{\mathrm L}}\partial_\lambda\Bigl[(\partial_\lambda\widetilde K_{\mathrm L})^{-1}(2\lambda)^{-1/4}h\Bigr]\,\mathrm d\lambda
,
\end{align} 
which implies, due to \eqref{eq:23072410} again, that for any $|x|>R'$
\begin{align*}
\bigl|(\widetilde  U^\pm_{\mathrm L}(t)h)(x)\bigr|
\le C_2\langle x\rangle^{-(d+1)/2}.
\end{align*}
Since $\widetilde  U^\pm_{\mathrm L}(t)h$ is uniformly bounded for $|x|\le R'$, we 
conclude that  $\widetilde U^\pm_{\mathrm L}(t)h\in \mathcal H$.

We  prove the second assertion in several steps. 

\smallskip
\noindent
\textit{Step II.}\ 
By the uniform boundedness of $\mathrm e^{\pm\mathrm itH}$ and Lemma~\ref{lem:2307172208} it suffices to show that as $t\to\infty$ 
\begin{align*}
(U^\pm_{\mathrm L}(t)h)(x)
=
\e^{\mp 3\pi\i/4}t^{-1}\abs{x}^{1-d/2}\e^{\pm \i K_{\mathrm L}(t,x)}h\parb{ x^2/(2t^2), \hat x }
+o(t^0)\ \ \text{in }\mathcal H. 
\end{align*}
Recall that  $S_{\mathrm L}$ is  defined for $I=[\mu'^2/2,\infty)$ as in  Lemma~\ref{lem:2307172208}.
Let  $\mu''=(\mu'+\mu)/2$, and decompose 
$\R_+\times \R^d= \Omega_{\mu''}\cup \Omega^c_{\mu''}$ and correspondingly 
\begin{align}
\label{eq:appDy12}
(\widetilde U^\pm_{\mathrm L}(t)h)(x)
&=
1_{\Omega_{\mu''}}(t,x)(\widetilde U^\pm_{\mathrm L}(t)h)(x)
+1_{\Omega^c_{\mu''}}(t,x)(\widetilde U^\pm_{\mathrm L}(t)h)(x)
.
\end{align} 
{Here let us prove that the} second term of \eqref{eq:appDy12} is negligible. 
In fact, using \eqref{eq:230724} for $\widetilde K_{\mathrm L}$, 
we can estimate for any $\lambda\ge \mu^2/2$ and $(t,x)\notin \Omega_{\mu''}$ with sufficiently large $t$ 
\begin{align*}
\partial_\lambda \widetilde K_{\mathrm L}(\lambda,t,x)\le -c_2t,
\quad 
\bigl|\partial_\lambda^2 \widetilde K_{\mathrm L}(\lambda,t,x)\bigr|\le C_3\langle x\rangle.
\end{align*}
This allows us to integrate by parts in the same way as for
\eqref{eq:appDy1} and deduce the pointwise bound 
\begin{align*}
\bigl|1_{\Omega^c_{\mu''}}(t,x) (\widetilde U^\pm_{\mathrm L}(t)h)(x)\bigr|
\le C_4t^{-1}\langle x\rangle^{-(d-1)/2}.
\end{align*}
By integration we then conclude the norm-bound
\begin{equation*}
1_{\Omega^c_{\mu''}} \widetilde U^\pm_{\mathrm L}(t)h=\mathcal O\bigl(t^{-1/2
}\bigr)\ \ \text{in }\mathcal H.
\end{equation*}
 
\smallskip
\noindent
\textit{Step III.}\ 
As for the first term of \eqref{eq:appDy12} we decompose it as follows.  
Using that $\mu''>\mu'$, we can  for any $(t,x)\in\Omega_{\mu''}$ 
find a unique critical point $\lambda_{\mathrm L,c}=\lambda_{\mathrm L,c}(t,x)$ of $\widetilde K_{\mathrm L}$.
Then take any $\eta_0\in C^\infty_{\mathrm c}(\mathbb R)$ such that
$\eta_0(s)=1$ for $|s|\leq  1/2$ while $\eta_0(s)=0$ for $|s|\geq  1$, and set 
\[
\eta(\lambda,t,x)=\eta_0\bigl(\langle x\rangle^{\delta}(\lambda-\lambda_{\mathrm L,c})\bigr),\quad 
\tfrac12-\tfrac{\rho}6<\delta <\tfrac12-\tfrac1{4(l'-1)}.
\]
We recall that $l'$ is the fixed integer obeying the condition $l'>1+2/\rho$. Obviously we can
find $\delta$ fulfilling these constraints, henceforth taken fixed. 
We now decompose 
\begin{align*}
1_{\Omega_{\mu''}}\widetilde U^\pm_{\mathrm L}(t)h
&=
\mp\mathrm i (2\pi)^{-1/2}1_{\Omega_{\mu''}}\chi\abs{x}^{-(d-1)/2}
\int_0^\infty
(\eta+(1-\eta))
\mathrm e^{\pm\mathrm i\widetilde K_{\mathrm L}}
(2\lambda)^{-1/4}h\,\mathrm d\lambda
\\&
=: \psi_1(t,x)+\psi_2(t,x)
.
\end{align*}

\smallskip
\noindent
\textit{Step IV.}\ 
The second term $\psi_2$ is negligible. 
In fact, for any $k=2,\dots,l'$, sufficiently large $t$ 
and $(\lambda,t,x)\in\mathop{\mathrm{supp}}(1-\eta)$ with
$(\lambda,\hat x)\in\mathop{\mathrm{supp}}h$, 
$(t,x)\in \Omega_{\mu''}$
we can bound  
\begin{align*}
\bigl|\partial_\lambda \widetilde K_{\mathrm L}(\lambda,t,x)\bigr|
\ge c_2\langle x\rangle^{1-\delta},
\quad 
\bigl|\partial_\lambda^k \widetilde K_{\mathrm L}(\lambda,t,x)\bigr|\le C_3\langle x\rangle^{1+\sigma-m(k)+k}
.
\end{align*}
This and the lower bound $\delta>1/2-\rho/6>1/2-\rho/2$ imply that
each time we integrate $\psi_2$ by parts as in \eqref{eq:appDy1} we at
least gain a decay of order $\langle x\rangle^{-1+2\delta}$. 
 Hence by doing the integration by parts in total $(l'-1)$ times we
 can 
 conclude that 
\begin{align*}
\psi_2=o(t^{0})\ \ \text{in }\mathcal H.
\end{align*}

\smallskip
\noindent
\textit{Step V.}\ 
It remains to investigate $\psi_1$. Let us expand the phase function as 
\[
\widetilde K_{\mathrm L}(\lambda,t,x)
=K_{\mathrm L}(t,x)+\tfrac12A(t,x)(\lambda-\lambda_{\mathrm L,c})^2
+B(\lambda,t,x)(\lambda-\lambda_{\mathrm L,c})^3
\]
with 
\begin{align*}
A(t,x)&=(\partial_\lambda^2S_{\mathrm L})(\lambda_{\mathrm L,c},x)
,\\
B(\lambda,t,x)
&=\tfrac12\int_0^1(1-\tau)^2(\partial_\lambda^3S_{\mathrm L})(\lambda_{\mathrm L,c}+\tau(\lambda-\lambda_{\mathrm L,c}),x)\,\mathrm d \tau
.
\end{align*}
We substitute the expression into $\psi_1$, and split the integral as 
\begin{align*}
\psi_1
&=
\mp\mathrm i (2\pi)^{-1/2}1_{\Omega_{\mu''}}\chi\abs{x}^{-(d-1)/2}
\mathrm e^{\pm\mathrm iK_{\mathrm L}}
\biggl(
(2\lambda_{\mathrm L,c})^{-1/4}h(\lambda_{\mathrm L,c},\cdot)
\int_0^\infty\e^{\pm \i A(\lambda-\lambda_{\mathrm L,c})^2/2}\eta\,\mathrm d\lambda
\\&\qquad 
+\int_0^\infty
\e^{\pm \i A(\lambda-\lambda_{\mathrm L,c})^2/2}
\eta\bigl[\e^{\pm \i B(\lambda-\lambda_{\mathrm L,c})^3}
(2\lambda)^{-1/4}h-(2\lambda_{\mathrm L,c})^{-1/4}h(\lambda_{\mathrm L,c},\cdot)\bigr]\,\mathrm d\lambda
\biggr)
.
\end{align*} 
Let us denote the last integral by $T$, and show that its contribution is negligible. 
By an integration by parts we can rewrite it as 
\begin{align*}
T
=
-
\int_{|\lambda-\lambda_{\mathrm L,c}|\leq \langle x\rangle^{-\delta}}
&\biggl(
\int_{\lambda_{\mathrm L,c}}^\lambda \e^{\pm \i A(\lambda'-\lambda_{\mathrm L,c})^2/2}\,\mathrm d\lambda'\biggr)
\\&{}
\cdot 
\tfrac{\partial}{\partial\lambda}
\eta\bigl[\e^{\pm \i B(\lambda-\lambda_{\mathrm L,c})^3}
(2\lambda)^{-1/4}h-(2\lambda_{\mathrm L,c})^{-1/4}h(\lambda_{\mathrm L,c},\cdot)\bigr]\,\mathrm d\lambda
.
\end{align*}
Then by the van der Corput Lemma, cf.\ \cite [p.\ 332]{St}, 
and the assumed support property of $h$, it follows that   on $\mathop{\mathrm{supp}}(1_{\Omega_{\mu''}}\chi)$ 
\[T=\mathcal O\bigl(\langle
x\rangle^{-1/2-\delta'}\bigr),\quad \delta'=\delta-\tfrac12+\tfrac{\rho}6,\]
so that indeed 
\begin{align*}
\psi_1
\pm\mathrm i (2\pi)^{-1/2}1_{\Omega_{\mu''}}\abs{x}^{-(d-1)/2}
\mathrm e^{\pm\mathrm iK_{\mathrm L}}
(2\lambda_{\mathrm L,c})^{-1/4}h(\lambda_{\mathrm L,c},\cdot)\int_0^\infty \e^{\pm \i A(\lambda-\lambda_{\mathrm L,c})^2/2}\eta \,\mathrm d\lambda
=o(t^0)
\end{align*}
as a vector in $\mathcal H$. 

\smallskip
\noindent
\textit{Step VI.}\ 
Finally  we remove $\eta$ from the last integral with another admissible error, and then 
implement the Gaussian integral to obtain 
\begin{align*}
\psi_1
&=
\mathrm e^{\mp\mathrm i3\pi/4}\abs{x}^{-(d-1)/2}
\mathrm e^{\pm\mathrm iK_{\mathrm L}}
|A|^{-1/2}
(2\lambda_{\mathrm L,c})^{-1/4}h(\lambda_{\mathrm L,c},\cdot)
+o(t^{0})\ \ \text{in }\mathcal H
.
\end{align*} 
Thus, using Lemma~\ref{lem:2307172208}  and 
\[A= -(2\lambda_{\mathrm L,c})^{-3/2}|x|+\mathcal O\bigl(\langle x\rangle^{1-\sigma}\bigr)
,\quad 
\lambda_{\mathrm L,c}=x^2/(2t^2)+\mathcal O\bigl(\langle x\rangle^{-\sigma}\bigr),\]
we obtain the second assertion of the lemma. 
\end{proof}

\subsubsection{Asymptotic completeness}\label{subsubsec:230726}

Now we are ready to prove the first part of Theorem~\ref{thm:time-depend-theory2}.

\begin{proof}[Proof of Theorem~\ref{thm:time-depend-theory2} (\ref{item:23071912})]
It suffices to show the identity \eqref{eq:fund2} since 
then the asymptotic completeness is obvious by the unitarity of $\mathcal F^\pm$. 
Let $S_{\mathrm L}$ be defined for an interval
$I_{\mu'}=[\mu'^2/2,\infty)$, $0<\mu'< \mu$,  sufficiently big to include the given closed interval $I$. 
Then for any $h\in C^\infty_{\mathrm c}(\R_+\times\mathbb S^{d-1})$ with $\mathop{\mathrm{supp}}h\subseteq (I\cap J)\times\mathbb S^{d-1}$ and 
for any $\psi\in C^\infty_{\mathrm c}(\mathbb R^d)$ we can compute by Lemma~\ref{lem:time-depend-theory} 
\begin{align*}
\bigl\langle \psi, W^\pm\mathrm e^{\mp\mathrm i\Phi_{\mathrm L}} h\bigr\rangle_{\mathcal H}
&
=\lim_{\epsilon\to 0_+}
\int_0^\infty\epsilon \mathrm e^{-\epsilon t}\bigl\langle \psi, 
\mathrm e^{\pm\mathrm itH}\widetilde  U^\pm_{\mathrm L}(t) h\bigr\rangle_{\mathcal H}\,\mathrm dt
\\&
=\lim_{\epsilon\to 0_+}(\pm 2\pi\mathrm i)^{-1}
\int_0^\infty\biggl(\int_0^\infty\bigl\langle \epsilon \mathrm e^{\mp\mathrm it(H-\lambda\mp\mathrm i\epsilon)}\psi, 
\phi_\pm^{S_{\mathrm L}}[h(\lambda,\cdot)]\bigr\rangle_{\mathcal H}\,\mathrm d\lambda\biggr)\,\mathrm dt
.
\end{align*}
By Fubini's theorem and \eqref{eq:23062917} we can further proceed as 
\begin{align*}
\bigl\langle \psi, W^\pm\mathrm e^{\mp\mathrm i\Phi_{\mathrm L}} h\bigr\rangle_{\mathcal H}
&
=\lim_{\epsilon\to 0_+}
(\pm 2\pi\mathrm i)^{-1}
\int_0^\infty\bigl\langle \mp\mathrm i\epsilon R(\lambda\pm\mathrm i\epsilon)\psi , 
\phi_\pm^{S_{\mathrm L}}[h(\lambda,\cdot)]\bigr\rangle_{\mathcal H}\,\mathrm d\lambda
\\&
=\lim_{\epsilon\to 0_+} 
(\pm 2\pi\mathrm i)^{-1}
\int_0^\infty\bigl\langle \psi - (H-\lambda)R(\lambda\pm\mathrm i\epsilon)\psi 
,\phi_\pm^{S_{\mathrm L}}[h(\lambda,\cdot)]\bigr\rangle_{\mathcal H} 
\,\mathrm d\lambda
\\&
=\lim_{\epsilon\to 0_+} 
(\pm 2\pi\mathrm i)^{-1}
\int_0^\infty
\bigl\langle \psi 
,\phi_\pm^{S_{\mathrm L}}[h(\lambda,\cdot)]-R(\lambda\mp\mathrm i\epsilon)\psi_\pm^{S_{\mathrm L}}[h(\lambda,\cdot)]\bigr\rangle_{\mathcal H} 
\,\mathrm d\lambda
\\&
= 
\int_0^\infty
\bigl\langle \psi 
,F^\pm(\lambda)^*\big(\mathrm e^{\mp\mathrm i\Theta_{\mathrm L}(\lambda,\cdot)}h(\lambda,\cdot)\bigr)\bigr\rangle_{\mathcal H} 
\,\mathrm d\lambda
\\&
= 
\int_0^\infty
\bigl\langle \mathrm e^{\mp\mathrm i\Theta_{\mathrm L}(\lambda,\cdot)}F^\pm(\lambda)\psi 
,h(\lambda,\cdot)\bigr\rangle_{\mathcal G} 
\,\mathrm d\lambda
\\&
= 
\bigl\langle \mathrm e^{\pm\mathrm i\Theta_{\mathrm L}}\mathcal F^\pm\psi 
,h\bigr\rangle_{\widetilde{\mathcal H}} 
.
\end{align*}
This implies, as operators $\mathcal H_{I\cap J}\to\widetilde{\mathcal H}_{I\cap J}$,
\[
(W^\pm)^*
=\mathrm e^{\pm\mathrm i(\Theta_{\mathrm L}-\Phi_{\mathrm L})}\mathcal F^\pm
.
\]
We have proven \eqref{eq:fund2} with $\Psi=\Phi_{\mathrm L}-\Theta_{\mathrm L}$, 
and we are done with (\ref{item:23071912}). 
\end{proof}

To prove the remaining assertion
  Theorem~\ref{thm:time-depend-theory2} (\ref{item:23071914}) we use
an approximation argument to be studied in the  following. 
Fix any $\rho\in(0,\sigma)$ and $\delta\in (0,\sigma-\rho)$,
and let $\epsilon=1/j$ for any $j\in\mathbb N$. 
For these parameters $\rho, \epsilon$ and $\delta$ used as inputs in 
Lemma~\ref{lem:230111} (\ref{item:2307261}) we decompose $V$ accordingly 
and the denote the decomposition as 
\[V=V_{\mathrm S,j}+V_{\mathrm L,j}.\]

\begin{lemma}\label{lem:230726143}
Let $\mu>\mu'>0$, $I=[\mu'^2/2,\infty)$ and $J=[\mu^2/2,\infty)$, and
fix a sufficiently large $R>0$. 
Let $S,S_{\mathrm L,j}$ be determined as in Theorem~\ref{thm:main result2} for $V,V_{\mathrm L,j}$, 
and let $K,K_{\mathrm L,j}$ be their Legendre transforms as in Theorem~\ref{thm:2307032}, respectively.
Then
\begin{align}
\lim_{j\to\infty}\Theta_{\mathrm L,j}&=0\ \ \text{locally  uniformly on }I\times\mathbb S^{d-1}
\label{eq:2307270}
\end{align}
with $\Theta_{\mathrm L,j}$ being the limit from Lemma~\ref{lemma:uniq-asympt-norm} for $S$ and $S_{\mathrm L,j}$,
and 
\begin{align}
\lim_{j\to\infty}\Phi_{\mathrm L,j}&=0\ \ \text{locally uniformly on }J\times\mathbb S^{d-1}
\label{eq:2307271}
\end{align}
with $\Phi_{\mathrm L,j}$ being the limit from Lemma~\ref{lem:2307172208} for $K$ and $K_{\mathrm L,j}$.
\end{lemma}
\begin{remark}
In the construction of $S,S_{\mathrm L,j},K,K_{\mathrm L,j}$ 
we can use the same $R$ since it depends only on sizes of zeroth to second derivatives of $V,V_{\mathrm L,j}$, 
which are uniformly estimated due to Lemma~\ref{lem:230111}
(\ref{item:2307261}). 
See the proofs of Theorems~\ref{thm:main result2} and \ref{thm:2307032}.
\end{remark}
\begin{proof} 
\textit{Step I.}\ 
To prove \eqref{eq:2307270} {we can partially mimic  the proof of Lemma~\ref{lemma:uniq-asympt-norm}. 
For details of the following computations see the proof there,
though one should note that $S$ and $S_{\mathrm L,j}$ now are   more specific than before}. 

Let $y_j$ be the flow from Subsection~\ref{subsec:Generalized Fourier transform for rho>1/2 I} associated with $S_{\mathrm L,j}$,
and we discuss the limit 
\[
{\Sigma_{\mathrm L,j}(\lambda,\theta):=}\lim_{t\to \infty}\parb{ S(\lambda,y_j(\lambda,t,\theta))-S_{\mathrm L,j}(\lambda,y_j(\lambda,t,\theta))}
. 
\]
Omitting the argument, we can compute its derivative in $t$ as 
\begin{align}
\tfrac{\mathrm d}{\mathrm dt}(S-S_{\mathrm L,j})
=-\tfrac12|\nabla S_{\mathrm L,j}|^{-2}
(\nabla S-\nabla S_{\mathrm L,j})^2
-|\nabla S_{\mathrm L,j}|^{-2}\chi_R V_{\mathrm S,j}
.
\label{eq:2307272}
\end{align}
Hence we are led to consider 
\[
u=(\nabla S-\nabla S_{\mathrm L,j})^2. 
\]
Its $t$-derivative   is computed as
\begin{align*}
 \tfrac12|\nabla S_{\mathrm L,j}|^{2}\tfrac{\mathrm d}{\mathrm dt}u
&=
-(\nabla S-\nabla S_{\mathrm L,j})\cdot(\nabla^2 S)(\nabla S-\nabla S_{\mathrm L,j})
-(\nabla \chi_RV_{\mathrm S,j})\cdot(\nabla S-\nabla S_{\mathrm L,j}),
\end{align*}
and thus we can deduce that for any  (small)
  $\delta>0$  there exist  $C_1,T>0$   independent of $j\in\mathbb
  N$, such that 
\[
\tfrac{\mathrm d}{\mathrm dt}u
\le 
-(2-\delta)
  t^{-1}u+C_1j^{-1} t^{-2-\delta};\quad t\geq T
.
\] By integration we then deduce that 
\begin{align}
u(t)\le t^{\delta-2}T^{2-\delta}u(T)+C_1j^{-1}t^{-1-\delta};\quad t\geq T
.
\label{eq:230727322}
\end{align}

By another  integration   \eqref{eq:2307272} and  \eqref{eq:230727322} imply a uniform bound 
\begin{align*}
|S-S_{\mathrm L,j}|\le |(S-S_{\mathrm L,j})(T) |+ C_2 T u(T)+ C_3j^{-1}.
\end{align*} On the other hand by integrating a version of \eqref{eq:23021010} from
$t=0$ (where $u$ vanishes) to $t=T$  (for fixed $T$) we deduce that   
\begin{equation}\label{eq:Sminus}
|(S-S_{\mathrm L,j})(T) |+ C_2 T u(T)\leq  C_4j^{-1/2}.
\end{equation}    (In fact the bound holds with $j^{-1/2}$ replaced by
$j^{-1}$ if we invoke  the variational principle of Lemma
\ref{lemma:existence} to bound the first term.)

 We  conclude 
that for large $t$
\begin{align*}
|S-S_{\mathrm L,j}|\le C_5j^{-1/2}.
\end{align*}
In particular  $ 
|\Sigma_{\mathrm L,j}|, \,|\Theta_{\mathrm L,j}|\le C_3j^{-1/2}$, and  \eqref{eq:2307270} follows.

\smallskip
\noindent
\textit{Step II.}\ 
The  proof of \eqref{eq:2307271} is  similar.
 We  consider versions of 
 \eqref{eq:23071720} and \eqref{eq:23071720b} using  versions of
 \eqref{eq:23071720bb} and \eqref{eq:Dife1}, 
  omitting  here the details.
 \end{proof}

\begin{proof}[Proof of Theorem~\ref{thm:time-depend-theory2} (\ref{item:23071914})]
Let $K$ be the Legendre transform of $S$ from Theorem~\ref{thm:main
  result2}, as given in (\ref{item:23071914}).
Note that the phase corrections $\mathrm e^{\mp
    \i\Psi}=\mathrm e^{\pm\mathrm i(\Theta_{\mathrm L}-\Phi_{\mathrm
      L})}$ are independent of choice of $V_{\mathrm L}$
since so are $\mathcal F^\pm$ and $W^\pm$. 
Thus it suffices to choose $V_{\mathrm L,j}$ such that 
the associated difference $\Theta_{\mathrm L,j}-\Phi_{\mathrm L,j}$ converges to $0$ as $j\to\infty$. 
However this obviously follows from Lemma~\ref{lem:230726143}. 
\end{proof}


\begin{thebibliography}{ACH}


\bibitem[ACH]{ACH} S. Agmon, J. Cruz, I. Herbst:
\emph{Generalized Fourier transform for Schr\"odinger operators with
 potentials of order zero}, J. Funct. Anal. \textbf{167} (1999),
345{--}369.

\bibitem[AIIS1]{AIIS}
T. Adachi, K. Itakura, K. Ito, E. Skibsted, 
\emph{Stationary scattering theory for $1$-body Stark operators, I},
Pure Appl.\ Funct.\ Anal.\ \textbf{7} (2022), no.\ 3, 825--861.

 \bibitem[AIIS2]{AIIS2}
T. Adachi, K. Itakura, K. Ito, E. Skibsted, 
\emph{New methods in spectral theory of $N$-body Schr{\"o}dinger
 operators}, Rev. Math. Phys. \textbf{33} (2021), 48 pp. 



\bibitem[CS]{CS} J. Cruz, E. Skibsted, \textit{Global solutions to
 the eikonal equation}, J. Differential Equations {\bf 255} (2013),
 4337--4377.

\bibitem[De]{De}
K. Deimling, Nonlinear Functional Analysis, Springer-Verlag, Berlin, 1982.

\bibitem[DG]{DG}
J. Derezi{\'n}ski and C. G{\'e}rard, \emph{Scattering theory of
 classical and quantum {$N$}-particle systems}, Texts and Monographs in
 Physics, Berlin, Springer 1997.


\bibitem[Ev]{Ev}
L.C. Evans, \emph{Partial differential equations},
 Graduate Studies in Mathematics {\bf 19}, Providence, AMS 1998.


\bibitem[GY]{GY} Y. G\^atel, D. Yafaev, \emph{On the solutions of
 the Schr\"odinger equation with radiation conditions at infinity: 
 the long-range case}, Ann. Inst. Fourier, Grenoble {\bf 49}
 no. 5 
 (1999), 1581--1602.

\bibitem[HS]{HS} 
I. Herbst, E. Skibsted, 
\emph{Time-dependent approach to radiation conditions}, 
Duke Math.\ J. {\bf 64} no. 1 (1991), 119--147.


\bibitem[H{\"o}1]{H0} L. H{\"o}rmander, \emph{The existence of wave
 operators in scattering theory}, Math. Z. \textbf{146} (1976), 68--91.

\bibitem[H{\"o}2]{H1} L. H{\"o}rmander, \emph{The analysis of linear
 partial differential operators. {I--IV}}, Berlin, Springer
 1983--85.


\bibitem[II]{II}
T. Ikebe, H. Isozaki, \emph{A stationary approach to the existence and
 completeness of long-range operators}, Integral equations and
operator theory \textbf{5} 
 (1982), 18--49.


\bibitem[Ir]{Ir} M. C. Irwin:
\emph{Smooth Dynamical Systems}, London, Academic Press, 1980.

\bibitem[Is1]{Is}
H. Isozaki, 
\emph{Eikonal equations and spectral representations for long-range Schr\"odinger Hamiltonians}, 
J. Math.\ Kyoto Univ. \textbf{20} (1980), 243--261.

\bibitem[Is2]{Is2}
H. Isozaki, 
\emph{On the generalized Fourier transforms
associated with Schr\"odinger operators
with long-range perturbations}, J. Reine Angew. Math. \textbf{337} (1982), 18--67.

\bibitem[IS1]{IS2} K. Ito, E. Skibsted, \emph{Stationary scattering
 theory on manifolds}, Ann. Inst. Fourier, Grenoble (2021), 55 pp.


 \bibitem[IS2]{IS3} 
K. Ito, E. Skibsted, 
\emph{Time-dependent scattering theory on manifolds},
J. Funct. Anal. \textbf{277} (2019), 1423--1468. 

\bibitem[IS3]{IS4} 
K. Ito, E. Skibsted, 
\emph{Scattering theory for Riemannian Laplacians},
J. Funct. Anal. \textbf{264} (2013), 1929--1974. 
 


\bibitem[RS]{RS}
M.~Reed and B.~Simon, \emph{Methods of modern mathematical physics {I}-{IV}},
 New York, Academic Press 1972-78.

\bibitem[Sa]{Sa} 
Y. Sait{\=o}, 
\emph{Spectral representations for Schr\"odinger operators with a long-range potentials},
Lecture Notes in Mathematics {\bf 727}, Berlin, Springer 1979.


\bibitem[Sk1]{Sk1} 
E. Skibsted, 
\emph{Green functions and completeness; the $3$-body problem
 revisited}, preprint 2 August  2024, http://arxiv.org/abs/2205.15028v2



\bibitem[Sk2]{Sk} 
E. Skibsted, 
\emph{Renormalized two-body low-energy scattering}, 
Journal d'Analyse Math\'ematique {\bf 122} (2014), 25--68.

\bibitem[St]{St} 
E.M. Stein, \emph{Harmonic analysis: real-variable
 methods, orthogonality and oscillatory integrals}, Princeton
 University Press, Princeton New Jersey 1993.

\bibitem[Yo]{Yo}
K. Yosida, \emph{Functional Analysis}, Springer, Berlin, 1965.

\end{thebibliography}
\end{document}